%% file: main.tex
\documentclass{article}

\usepackage{microtype}
\usepackage{graphicx}
\usepackage{subfigure}
\usepackage{booktabs} % for professional tables

\usepackage[final]{hyperref}

\usepackage[accepted]{icml2025} %

\usepackage{amsmath}
\usepackage{amssymb}
\usepackage{mathtools}
\usepackage{amsthm}

\usepackage[capitalize,noabbrev]{cleveref}

\theoremstyle{plain}
\newtheorem{theorem}{Theorem}[section]
\newtheorem{proposition}[theorem]{Proposition}

\theoremstyle{definition}

\theoremstyle{remark}

\usepackage[textsize=tiny]{todonotes}
\input{util}

\icmltitlerunning{Sublinear Adaptive Algorithm for Submodular Maximization}

\begin{document}

\twocolumn[
\icmltitle{Breaking Barriers: Combinatorial Algorithms for Non-Monotone Submodular \\
 Maximization with Sublinear Adaptivity and $1/e$ Approximation}




\icmlsetsymbol{equal}{*}

\begin{icmlauthorlist}
\icmlauthor{Yixin Chen}{yyy}  
\icmlauthor{Wenjing Chen}{yyy}
\icmlauthor{Alan Kuhnle}{yyy}
\end{icmlauthorlist}

\icmlaffiliation{yyy}{Department of Computer Science \& Engineering, Texas A\&M University, College Station, Texas, USA}

\icmlcorrespondingauthor{Alan Kuhnle}{kuhnle@tamu.edu}

\icmlkeywords{Machine Learning, ICML}

\vskip 0.3in
]



\printAffiliationsAndNotice{}  % leave blank if no need to mention equal contribution

\begin{abstract}
% With the explosion of data in modern applications, practical algorithms have garnered increasing attention under various large-scale data settings.
% This work focuses on developing parallel combinatorial approximation algorithms for maximizing a non-monotone submodular function subject to a size constraint $k$
% with a ground set of size $n$.
% The current state-of-the-art approximation ratio for this problem is $1/e$, achieved by
% a continuous algorithm~\citep{Ene2020a} with adaptivity $\oh{\log(n)}$.
% We propose two parallel combinatorial algorithms, both
% achieving $\oh{\log(n)\log(k)}$ adaptivity and 
% $\oh{n\log(n)\log(k)}$ query complexity.
% These algorithms improve the best-known deterministic approximation ratio from $0.125-\epsi$ to $0.25-\epsi$
% and the best-known randomized approximation ratio from $0.25-\epsi$ to $1/e-\epsi\approx 0.367-\epsi$,
% breaking the barrier between continuous and combinatorial approaches.
% Empirical evaluations demonstrate the effectiveness of our methods, 
% achieving competitive objective values, 
% with the first algorithm excelling in query efficiency.
With the rapid growth of data in modern applications, parallel algorithms 
for maximizing non-monotone submodular functions have gained significant attention.
In the parallel computation setting,
\color{black}
the state-of-the-art approximation ratio of $1/e$ is 
achieved by a continuous algorithm~\citep{Ene2020a} 
with adaptivity $\oh{\log(n)}$.
In this work, we focus on size constraints 
and present the first combinatorial algorithm matching this bound
-- a randomized parallel approach achieving $1/e-\epsi$ approximation ratio.
This result bridges the gap between continuous and combinatorial approaches for this problem.
As a byproduct, we also develop a simpler
$(1/4-\epsi)$-approximation algorithm  
with high probability ($\ge 1-1/n$).
Both algorithms achieve $\oh{\log(n)\log(k)}$ adaptivity and 
$\oh{n\log(n)\log(k)}$ query complexity.
Empirical results show our algorithms achieve competitive objective values, 
with the $(1/4-\epsi)$-approximation algorithm particularly efficient in queries.
\end{abstract}

\input{1_intro}
\input{2_alg}
\input{3_exp}
% \section{Conclusion}
% The state-of-the-art $1/e$ approximation ratio for sublinear adaptive algorithms 
% is achieved by a continuous algorithm~\citep{Ene2020a}.
% For combinatorial algorithms with sublinear adaptivity,
% the best-known result is a randomized $1/4$ approximation ratio~\citep{Cui2023}.
% In this work, we present a sublinear adaptive approximation algorithm 
% achieving $1/4-\epsi$ approximation ratio with high probability,
% and further improve this ratio achieved to $1/e$, 
% breaking the barrier between continuous and combinatorial algorithms. 
% These advancements are made by a novel blending analysis technique,
% which offers a fresh perspective for analyzing greedy-based algorithms.

\clearpage
\section*{Impact Statement}
This paper presents work whose goal is to advance the field of 
Machine Learning. There are many potential societal consequences 
of our work, none which we feel must be specifically highlighted here.
\bibliography{yixin-refs}
\bibliographystyle{icml2025}

%%%%%%%%%%%%%%%%%%%%%%%%%%%%%%%%%%%%%%%%%%%%%%%%%%%%%%%%%%%%%%%%%%%%%%%%%%%%%%%
%%%%%%%%%%%%%%%%%%%%%%%%%%%%%%%%%%%%%%%%%%%%%%%%%%%%%%%%%%%%%%%%%%%%%%%%%%%%%%%
% APPENDIX
%%%%%%%%%%%%%%%%%%%%%%%%%%%%%%%%%%%%%%%%%%%%%%%%%%%%%%%%%%%%%%%%%%%%%%%%%%%%%%%
%%%%%%%%%%%%%%%%%%%%%%%%%%%%%%%%%%%%%%%%%%%%%%%%%%%%%%%%%%%%%%%%%%%%%%%%%%%%%%%
\newpage
\appendix
\onecolumn
\input{apx_tech}
\input{apx}

\input{apx_exp}

\end{document}

%% file: util.tex
\usepackage{array, makecell} 
\usepackage{xspace}
\usepackage{thm-restate}
\usepackage{multirow}
\usepackage{tablefootnote}
\allowdisplaybreaks

\usepackage[ruled,lined,noend,linesnumbered]{algorithm2e}
\DontPrintSemicolon
\SetAlgoProcName{Paradigm}{anautorefname}
\SetKwInOut{Init}{Initialize}
\SetKwProg{Fn}{Procedure}{:}{}

\newcommand\numberthis{\addtocounter{equation}{1}\tag{\theequation}}

\newcommand{\rev}{\color{black}}

\newcommand{\uni}{\mathcal U}
\newcommand{\reals}{\mathbb{R}_{\ge 0}}

\newcommand{\st}{\textit{s.t.}\xspace}

\newcommand{\ff}[1]{ f \left( #1 \right) }

\newcommand{\marge}[2]{\Delta \left( #1 \, \middle| \, #2 \right) }

\newcommand{\ex}[1]{\mathbb{E}\left[ #1 \right]}

\newcommand{\exc}[2]{ \mathbb{E}\left[\left. #1 \, \right| \; #2 \right] }
\newcommand{\oh}[1]{ \mathcal O \left( #1 \right) }
\newcommand{\epsi}[0]{ \varepsilon }
\newcommand{\opt}{\text{OPT}}
\newcommand{\prob}[1]{ \mathbb{P} \left[ #1 \right] }

\newcommand{\sm}{\textsc{SMCC}\xspace}
\renewcommand{\restriction}{\mathord{\upharpoonright}}

\DeclareMathOperator*{\argmax}{arg\,max}
\DeclareMathOperator*{\argmin}{arg\,min}

\newcommand{\nmon}{\textsc{SM-Gen}\xspace}
\newcommand{\mon}{\textsc{SM-Mon}\xspace}

\newcommand{\maxcut}{$\texttt{maxcut}$\xspace}
\newcommand{\revmax}{$\texttt{revmax}$\xspace}

\declaretheorem[style=definition,sibling=theorem]{lemma}
\declaretheorem[style=definition,numberwithin=section]{claim}
\newcommand{\rg}{\textsc{RandomGreedy}}

\newcommand{\ig}{\textsc{InterlaceGreedy}\xspace}
\newcommand{\itg}{\textsc{InterpolatedGreedy}\xspace}
\newcommand{\ptgone}{\textsc{ParallelInterlaceGreedy}\xspace}
\newcommand{\ptgtwo}{\textsc{ParallelInterpolatedGreedy}\xspace}
\newcommand{\ptgoneshort}{\textsc{PIG}\xspace}
\newcommand{\ptgtwoshort}{\textsc{PItG}\xspace}

\newcommand{\ts}{\textsc{ThreshSeq}\xspace}

\newcommand{\dist}{\textsc{Distribute}\xspace}
\newcommand{\prefix}{\textsc{Prefix-Selection}\xspace}
\newcommand{\update}{\textsc{Update}\xspace}

\newcommand{\randomgreedy}{\textsc{RandomGreedy}\xspace}

\newcommand{\fast}{\textsc{Fast}\xspace}
\newcommand{\lspgb}{\textsc{LS+PGB}\xspace}
\newcommand{\parskp}{\textsc{ParSKP}\xspace}
\newcommand{\parssp}{\textsc{ParSSP}\xspace}

\newcommand{\frg}{\textsc{FastRandomGreedy}\xspace}
\newcommand{\anm}{\textsc{AdaptiveNonmonotoneMax}\xspace}
\newcommand{\unc}{\textsc{UnconstrainedMax}\xspace}

%% file: 1_intro.tex
\section{Introduction}
\textbf{Submodular Optimization.} Submodular optimization is a powerful framework for solving combinatorial optimization problems that exhibit diminishing returns~\citep{Nemhauser1978,feige1995approximating,cornuejols1977exceptional}. In a monotone setting, adding more elements to a solution always increases its utility, but at a decreasing rate. On the other hand, non-monotone objectives may have elements that, when added, can reduce the utility of a solution. This versatility makes submodular optimization widely applicable across various domains. For instance, in data summarization~\citep{DBLP:conf/aaai/MirzasoleimanJ018,DBLP:conf/nips/TschiatschekIWB14} and feature selection~\citep{DBLP:journals/corr/abs-2202-00132}, it helps identify the most informative subsets efficiently. In social network analysis~\citep{DBLP:conf/kdd/KempeKT03}, it aids in influence maximization by selecting a subset of individuals to maximize information spread. Additionally, in machine learning~\citep{DBLP:conf/acl/BairiIRB15,DBLP:conf/nips/ElenbergDFK17,DBLP:conf/iclr/PrajapatMZ024}, submodular functions are used for diverse tasks like active learning, sensor placement, and diverse set selection in recommendation systems. These applications demonstrate the flexibility and effectiveness of submodular optimization in addressing real-world problems with both monotone and non-monotone objectives. % From ChatGPT

\textbf{Problem Definition and Greedy Algorithms.} In this work, we consider the size-constrained maximization of a submodular function: given a submodular function $f$ on ground set of size $n$, and given an integer $k$, find $\argmax_{S \subseteq \uni, |S| \le k} f(S)$. If additionally $f$ is assumed to be monotone, we refer to this problem as \mon{}; otherwise, we call the problem \nmon{}. 
For \mon{}, a standard greedy algorithm gives the optimal\footnote{Optimal in polynomially many queries to $f$ in the value query model \citep{Nemhauser1978a}.} approximation ratio
of $1 - 1/e \approx 0.63$ \citep{Nemhauser1978} in
at most $kn$ queries\footnote{Typically, queries to $f$ dominate other parts of the computation, so in this field time complexity of an algorithm is usually given as oracle complexity to $f$.} to $f$.
In contrast, standard greedy can't achieve any constant
approximation factor in the non-monotone case; however,
an elegant, randomized variant of greedy, the \randomgreedy{}
of \citet{Buchbinder2014a} obtains the same greedy
ratio (in expectation) of $1-1/e$ for monotone objectives, and $1/e \approx 0.367$ for the general,
non-monotone case, also in at most $kn$ queries. 
However, as modern instance sizes in applications have become very large,
$kn$ queries is too many. In the worst case, $k = \Omega(n)$, and
the time complexity of these greedy algorithms becomes quadratic.
To improve efficiency, \citet{DBLP:journals/mor/BuchbinderFS17} 
leverage sampling technique to get \frg,
reducing the query complexity to $\oh{n}$ when
$k \ge 8\epsi^{-2}\log(2\epsi^{-1})$.

\begin{table*}[ht]
  \centering \scriptsize
  \caption{Theoretical comparison of greedy algorithms and parallel algorithms with sublinear adaptivity.}
    \begin{tabular}{clll}
    \toprule
    Algorithm  & Ratio& Queries & Adaptivity \\%& Randomized?\\
    \midrule
    \frg~\citep{DBLP:journals/mor/BuchbinderFS17} & \boldmath{$1/e-\epsi$} & $\oh{n}$ & $k$ \\
    \textsc{ANM}~\citep{fahrbach2018non}& $0.039 - \epsi$
     & \boldmath{$\oh{{n \log (k)}}$} & \boldmath{$\oh{\log (n)}$}\\
     \citet{Ene2020a} & \boldmath{$1/e - \epsi$} & $\oh{k^2n\log^2(n)}\ddagger$ & \boldmath{$\oh{\log (n)}$}\\
    {\textsc{ParKnapsack}\citep{amanatidis2021submodular}}& $0.172-\epsi$ 
    & $\oh{kn\log(n)\log(k)} \| \oh{n\log(n)\log^2(k)}$
    & {\boldmath{$\oh{\log(n)}$}} $\|\oh{\log(n)\log(k)}$ \\
    AST \citep{Chen2024} & $1/6 - \epsi $ & \boldmath{$\oh{n \log (k)}$} & \boldmath{$\oh{\log(n)}$} \\
    ATG \citep{Chen2024} & $0.193 - \epsi$ & \boldmath{$\oh{n \log (k)}$} & $\oh{\log(n) \log(k)}$ \\
    \parskp \citep{Cui2023}& $0.125-\epsi$ & $\oh{kn \log^2(n)} \| \oh{n\log^2(n)\log (k) }$ & $\oh{\log(n)} \| \oh{ \log(n) \log (k) } $ \\
    \parssp \citep{Cui2023}& $0.25-\epsi$ & $\oh{kn\log^2(n)} \| \oh{n\log^2(n)\log (k) }$ & $\oh{\log^2(n)} \| \oh{ \log^2(n) \log (k) } $ \\
    \midrule
    \ptgone (Alg.~\ref{alg:ptgone}) & $0.25-\epsi \dagger$  & $\oh{n\log(n)\log(k)}$& $\oh{\log(n)\log(k)}$\\
    \ptgtwo (Alg.~\ref{alg:ptgtwo}) & \boldmath{$1/e - \epsi $} & $\oh{n\log(n)\log(k)}$ & $\oh{\log(n)\log(k)}$ \\
    \bottomrule
    \end{tabular}
    \parbox{\linewidth}{\small $\ddagger$ The parallel algorithm in \citet{Ene2020a} queries to the continuous oracle.}
    \parbox{\linewidth}{\small $\dagger$ The approximation ratio is achieved with high probability (at least $1-1/n$).}
  \label{table:cmp}
  \vspace*{-1.5em}
\end{table*}  

\textbf{Parallelizable Algorithms.}
In addition to reducing the number of queries,
recently, much work has focused on developing parallelizable algorithms for submodular optimization.
One measure of
parallelizability is the
\textit{adaptive complexity} of an algorithm. That is,
the queries to $f$ are divided into adaptive rounds, where within each round the set queried may only depend
on the results of queries in previous rounds;
the queries within each round may be arbitrarily parallelized. Thus, the lower the adaptive complexity, the more
parallelizable an algorithm is. 
Although the initial algorithms with sublinear adaptivity were
impractical, for the monotone case, these works culminated in two practical algorithms:
\fast{} \citep{Breuer2020} and \lspgb{} \citep{Chen2021}, both of which achieve nearly
the optimal ratio in nearly linear time and nearly optimal adaptive rounds. 
For the nonmonotone case, 
the best approximation ratio achieved in sublinear adaptive rounds is $1/e$~\citep{Ene2020a}.
However, this algorithm queries to a continuous oracle,
which needs to be estimated through a substantial number of queries to the original set function oracle.
Although practical, sublinearly adaptive algorithms have also been
developed, the best ratio achieved in nearly linear time is nearly $1/4$ \citep{Cui2023},
significantly worse than the state-of-the-art\footnote{The best known ratio in polynomial time was very recently improved from close to $1/e$ to $0.401$ \citep{Buchbinder2023}},
this $1/4$ ratio also stands as the best even for superlinear time parallel algorithms.
Further references to parallel algorithms and their theoretical guarantees are provided in Table~\ref{table:cmp}.

\begin{algorithm}[ht]
    \KwIn{evaluation oracle $f:2^{\uni} \to \reals$, constraint $k$}
    \Init{$a_0 \gets \argmax_{x\in \uni} \ff{x}$,
    $A\gets B \gets \emptyset$, $D \gets E \gets \{a_0\}$,
    add $2k$ dummy elements to the ground set}
    \For{$i\gets 0$ to $k-1$}{
        $A\gets A+\argmax_{x\in \uni\setminus \left(A\cup B\right)} \marge{x}{A}$\;
        $B\gets B+\argmax_{x\in \uni\setminus \left(A\cup B\right)} \marge{x}{B}$\;
    }
    \For{$i\gets 1$ to $k-1$}{
        $D\gets D+\argmax_{x\in \uni\setminus \left(D\cup E\right)} \marge{x}{D}$\;
        $E\gets E+\argmax_{x\in \uni\setminus \left(D\cup E\right)} \marge{x}{E}$\;
    }
    \Return{$C\gets \argmax\{\ff{A}, \ff{B}, \ff{D}, \ff{E}\}$}
    \caption{$\ig(f,k)$: The \ig Algorithm~\citep{DBLP:conf/nips/Kuhnle19}}
    \label{alg:ig}
\end{algorithm}
\textbf{Greedy Variants for Parallelization.}
To enhance the approximation ratios for combinatorial sublinear adaptive algorithms, 
it is crucial to develop practical \color{black} parallelizable algorithms that serve as universal frameworks for such parallel approaches. 
Among existing methods, \ig~\citep{DBLP:conf/nips/Kuhnle19}, 
with an approximation ratio of $1/4$, and \itg\citep{DBLP:conf/kdd/ChenK23},
with an expected approximation ratio of $1/e$, 
have emerged as promising candidates 
due to their unique and deterministic interlacing greedy procedures. 

\rev
\ig (Alg.~\ref{alg:ig}) operates by first guessing whether the maximum singleton 
$a_{0} = \argmax_{x\in \uni} \ff{x}$ is contained within the optimal solution $O$.
Based on this hypothesis, the algorithm initializes two solution pools differently:
if $a_{0}$ is not in $O$, the pools begin empty ($A = B= \emptyset$);
otherwise, they are initialized with $a_{0}$ ($D = E = \{a_0\}$).
The algorithm then proceeds by alternately selecting elements for each pool 
in a greedy fashion, ultimately returning the best solution found.
We provide its theoretical guarantees below.
\begin{theorem}
Let $f:2^{\uni} \to \reals$ be submodular, let $k\in \uni$,
let $O = \argmax_{|S|\le k} \ff{S}$,
and let $C = \ig(f, k)$. Then
\[\ff{C} \ge \ff{O}/4,\]
and \ig makes $\oh{kn}$ queries to $f$.
\end{theorem}
The key to this guarantee lies in the alternating selection strategy 
between the two pools. 
For any disjoint pools $\{S,T\}$, it holds that
\begin{align*}
\ff{S} + \ff{T} \ge \ff{O\cup S} + \ff{O\cup T} \ge \ff{O},
\end{align*}
where the first inequality follows from the alternating selection strategy,
and the second inequality follows from
submodularity and monotonicity.

\color{black}
% We provide their pseudocodes and theoretical guarantees in Appendix~\ref{apx:pseudocode}. 
% \ig alternates between two sets, adding elements greedily to each set in turn. 
Building on this, \itg\citep{DBLP:conf/kdd/ChenK23} generalizes the approach by maintaining $\ell$ sets, 
each containing $k/\ell$ elements at each iteration
\rev
(see pseudocode and theoretical guarantees in Appendix~\ref{apx:pseudocode}).
\color{black}
However, in their original formulations, both algorithms require an initial step to guess whether the first element added to each solution belongs to the optimal set $O$.
This results in repeated for loop with different initial values in \ig
and a success probability of 
$(\ell+1)^{-\ell}$.
% In this paper, we present an improved analysis for a revised version of \ig and \itg, 
% which eliminates the need for this guessing step, providing a more effective and theoretically sound approach.

% \textbf{The threshold sampling subroutines.}

\subsection{Contributions}
\rev
\begin{figure}[h]
\centering
\includegraphics[width=\linewidth]{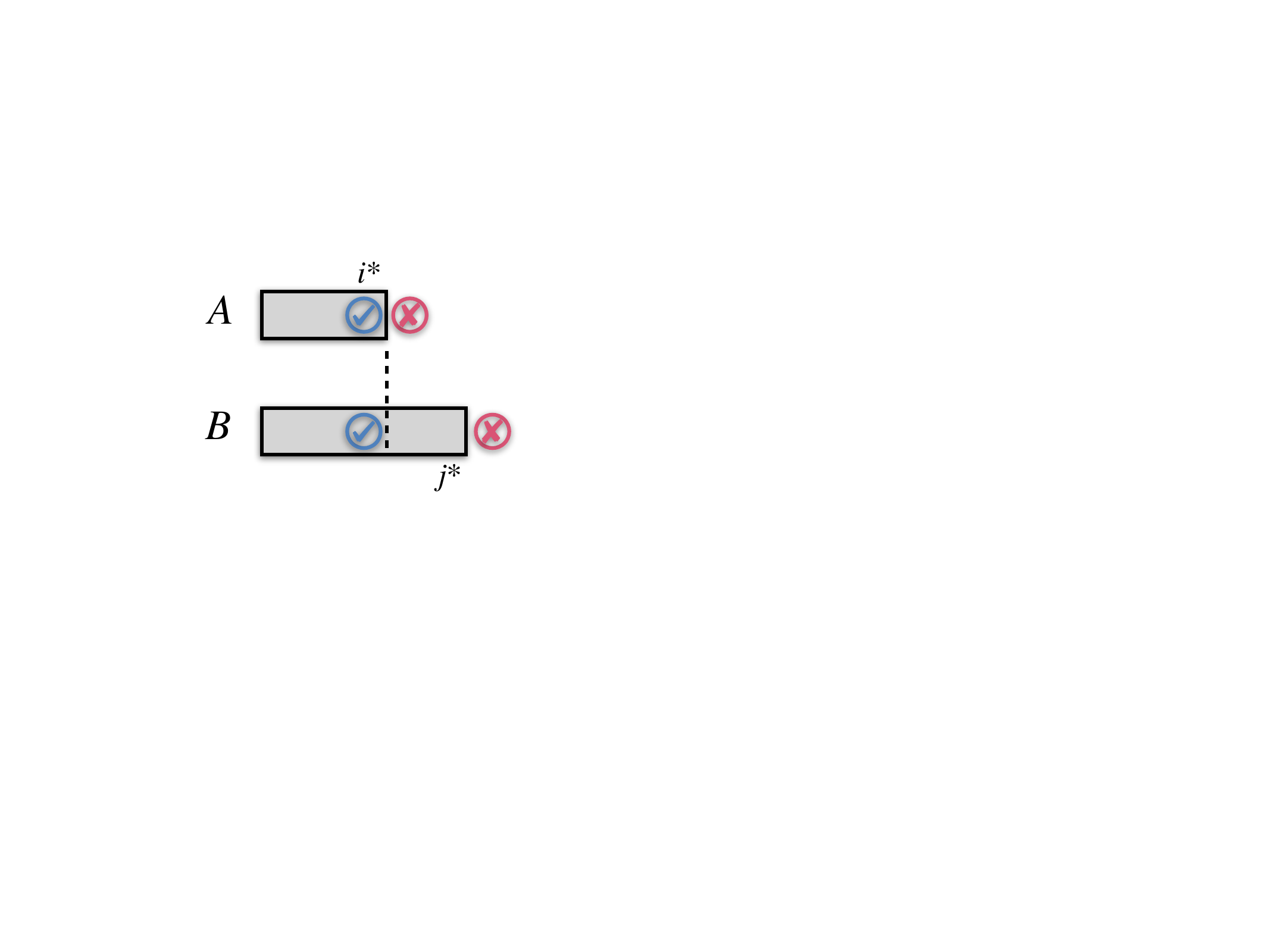}
\caption{This figure illustrates strategies employed by each algorithm. 
The three leftmost algorithms achieve an asymptotic approximation ratio of $1/4$,
while the three rightmost algorithms attain an asymptotic ratio of $1/e$.}
\label{fig:relation}
\end{figure}
\color{black}
\textbf{Technical Contributions: Blended Marginal Gains Strategy.}
Our first contribution is a novel method of analysis
for interlaced greedy-based algorithms,
such as \itg\citep{DBLP:conf/kdd/ChenK23}.
The core innovation introduces multiple upper bounds on the marginal gain 
of the same element regarding the solution, 
offering greater flexibility in the analysis.
This approach yields three notable advances:
First, we present Alg.~\ref{alg:gdtwo} (Section~\ref{sec:gd})
with $\oh{kn}$ query complexity, 
achieving an expected approximation ratios of $1/e-\epsi$,
eliminating the probabilistic guessing step that limited prior method to
$(1+\ell)^{-\ell}$ success probability.
Second, as a byproduct, we also provide a simplified version of
\ig~\citep{DBLP:conf/nips/Kuhnle19} that preserves its theoretical guarantees 
in Appendix~\ref{apx:greedy-1/4}.
Most importantly, these simplified variants establish 
the theoretical foundation for parallel algorithms.
By removing branching dependencies and probabilistic guesswork, 
our framework enables the first efficient parallelization 
of these interlaced greedy approaches while preserving their approximation guarantees.

\textbf{Parallel Algorithms with Logarithmic Adaptivity and 
Nearly-linear Query Complexity
Using a Unified Subroutine.}
In Section~\ref{sec:ptg},
we present two sublinear adaptive algorithms, 
\ptgone (\ptgoneshort) and \ptgtwo (\ptgtwoshort)
that share a unified subroutine, \ptgoneshort.
The core innovation of \ptgoneshort lies in its novel threshold sampling procedure, 
which simultaneously preserves the crucial alternating selection property 
of interlaced greedy methods
while enabling efficient parallel implementation.

Like prior parallel algorithms,
\ptgoneshort maintains descending thresholds for each solution, 
adding elements whose average marginal gain exceeds the current threshold.
However, \ptgoneshort introduces two critical modifications 
to maintain the interlaced greedy structure:
1) strict synchronization of batch sizes across all $\ell$ parallel solutions,
and 2) coordinated element selection to maintain sufficient marginal gain
for each solution.

This design achieves three fundamental properties.
First, it preserves the essential alternating selection
property of the interlaced greedy methods.
Second, through threshold sampling, it geometrically reduces
the size of candidate sets - crucial for achieving sublinear adaptivity.
Third, its efficient batch selection 
ensures each added batch provides sufficient marginal contribution to the solution.
Together, these properties allow \ptgoneshort to match
the approximation guarantees of the vanilla interlaced greedy method
while achieving parallel efficiency.

Leveraging this unified framework, a single call to \rev \ptgoneshort \color{black} achieves an approximation ratio of $(1/4-\epsi)$ with high probability.
Repeated calls to \rev \ptgoneshort \color{black} yields \ptgtwoshort,
further enhance the approximation ratio to an expected $(1/e-\epsi)$.
Both algorithms achieve $\oh{\log(n) \log(k)}$ adaptivity
and $\oh{n\log(n)\log(k)}$ query complexity,
making them efficient for large-scale applications.

\textbf{Empirical Evaluation.}
Finally, we evaluate the performance of our parallel algorithms
in Section~\ref{sec:exp} and Appendix~\ref{apx:exp}
across two applications and four datasets.
The results demonstrate that our algorithms achieve competitive objective values.
Notably, \ptgone outperforms other algorithms in terms of query efficiency,
highlighting its practical advantages.

%%% Local Variables: 
%%% mode: latex
%%% TeX-master: "main.tex"
%%% End: 

%% file: 2_alg.tex
\section{Preliminary}
\textbf{Notation.}
We denote the marginal gain of adding $A$ to $B$ by $\marge{A}{B} = \ff{A\cup B} - \ff{B}$.
For every set $S\subseteq U$ and an element $x\in \uni$,
we denote $S\cup \{x\}$ by $S+x$ and $S\setminus \{x\}$ by $S-x$.

\textbf{Submodularity.}
A set function $f:2^\uni\to \reals$ is submodular, 
if $\marge{x}{S}\ge \marge{x}{T}$ for all $S\subseteq T\subseteq \uni$
and $x\in \uni\setminus T$,
or equivalently, for all $A, B\subseteq \uni$,
it holds that $\ff{A} + \ff{B}\ge \ff{A\cup B} + \ff{A\cap B}$.
With a size constraint $k$,
let $O =\argmax_{S\subseteq \uni, |S| \le k} \ff{S}$.

In Appendix~\ref{apx:prop}, we provide several key propositions
derived from submodularity that streamline the analysis.

\textbf{Organization.}
In Section~\ref{sec:gd}, the blending technique is introduced
and applied to \ig and \itg,
with detailed analysis provided in Appendix~\ref{apx:greedy-1/4}
and~\ref{apx:greedy-1/e}.
Section 4 discusses their fast versions with
pseudocodes and comprehensive analysis in
Appendix~\ref{apx:tg}.
Subsequently, Section~\ref{sec:ptg} delves into our sublinear adaptive algorithms, 
with further technical details and proofs available in Appendix~\ref{apx:ptg}.
Finally, we provide the empirical evaluation in Section~\ref{sec:exp},
with its detailed setups and additional results
in Appendix~\ref{apx:exp}.

\section{A Parallel-Friendly Greedy Variant via Blending Technique}\label{sec:gd}
In this section, we present a simplified and practical variant
of \itg~\citep{DBLP:conf/kdd/ChenK23} that retains its theoretical guarantees 
while improving its success probability from $(\ell+1)^{-\ell}$ to $1$.
This simplification serves as a key step toward developing 
efficient parallel submodular maximization algorithms (Section~\ref{sec:ptg}). 
As a byproduct, we also show that \ig~\citep{DBLP:conf/nips/Kuhnle19} 
can be simplified (Appendix~\ref{apx:greedy-1/4})
and parallelized (Section~\ref{sec:ptg}).

\subsection{An Overview of Prior Work and Its Limitation}
\rev
\ig (Alg.~\ref{alg:ig}) employs a two-branch initialization strategy conditioned on
differently based on the guessing of whether the maximum singleton 
$a_{0} = \argmax_{x\in \uni} \ff{x}$ is contained within the optimal solution $O$.
The algorithm initializes two solution pools either as empty sets if $a_{0} \not \in O$,
or both containing $a_{0}$ otherwise.
It then alternates between two greedy procedures, 
growing solutions pools one by one until they reach size $k$.
\color{black}
The algorithm returns the best solution, achieving a $1/4-\epsi$ approximation.
This structure is necessary to handle the uncertainty about 
whether the maximum singleton $a_0$ coincides with $o_{\max} = \argmax_{o\in O}\ff{o}$.

Later on, \itg \rev (Alg.~\ref{alg:itg}, Appendix~\ref{apx:pseudocode}) \color{black} 
generalizes this idea by interlacing $\ell$ greedy procedures,
each adding $k/\ell$ elements to intermediate solutions.
Critically, the algorithm must account for uncertainty in the position of $o_{\max}$.
Since $o_{\max}$ could be one of the top-$\ell$
elements in marginal gain (or none of them), the algorithm must consider
$\ell+1$ possible cases. To handle this,
it maintains $\ell+1$ \rev \textit{solution families} (a set of multiple potential solutions), \color{black} 
each corresponding to a distinct guess about where $o_{\max}$ appears.
The algorithm then returns a random solution from the \rev solution families\color{black}, achieving a
$1/e-\epsi$ approximation with success probability $(\ell+1)^{-\ell}$.
To guarantee success probability $1$, 
all possible solution branches must be checked, 
increasing the query complexity by a prohibitive $\oh{(\ell+1)^{\ell}}$ factor.
While asymptotically efficient ($\oh{nk}$ queries),
\itg suffers from either low success probability or impractical computational overhead, 
making it unsuitable for parallelization.
% With more solutions created, \itg makes more guesses on where $o_{\max}$ is.
% As a result, $(\ell+1)$ solution pools are created, 
% with the solutions in the right pool being initialized by $\{o_{\max}\}$ if its marginal
% gain is top $\ell$, or initialized by the elements with top $\ell$ marginal gains respectively.
% Then, a random solution is returned among all the pools.
% By repeating the above steps $\ell$ times, a solution with size $k$ is returned.
% Thus, the success probability of achieving $1/e-\epsi$ approximation ratio is $(\ell+1)^{-\ell}$.
% Or, to achieve a success probability of 1, the algorithm needs to check all the possible branches of the solution, which increase the query complexity by a large factor of $\oh{(\ell+1)^{\ell}}$.
% Although the query complexity of \itg is asymptotically $\oh{nk}$,
% it either succeed with low probability or is impractical.

\subsection{Motivation for Simplification}
% \itg is not good enough to be parallelized due to its low success probability.
% To propose a more promising candidate of parallelizable algorithm, 
% we introduce a novel analysis technique 
% called the blended marginal gains strategy to analyze \itg.
% We show that the guessing step can be eliminated and 
% only one pool of pairwise disjoint solutions is needed.
% We provide the pseudocode of the simplified \itg in Alg.~\ref{alg:gdtwo}
% with its theoretical guarantee in Theorem~\ref{thm:gdtwo}.
The primary challenge in parallelizing \itg 
stems from its inherent branching structure - specifically, 
\rev
the need to make $\ell+1$ guesses,
each resulting in an independent construction path for solution pools.
\color{black}
This complex architecture naturally raises a fundamental question: 
\textit{Can we develop an alternative analysis framework that 
eliminates the need to explicitly guess the position of $o_{\max}$?}

To address this, we introduce a novel blended marginal gains strategy.
Our key insight demonstrates that by carefully tracking combined marginal gains 
across iterations, we can effectively work with just 
a single \rev solution family \color{black} of pairwise disjoint solutions 
while preserving all theoretical guarantees.

This simplification not only boosts the success probability to $1$ 
but also enables a more efficient parallel implementation. 
We present the pseudocode of the simplified \itg in Alg.~\ref{alg:gdtwo},
with theoretical guarantees in Theorem~\ref{thm:gdtwo}.
\begin{algorithm}[ht]
    \KwIn{evaluation oracle $f:2^{\uni} \to \reals$, constraint $k$, size of solution $\ell$, error $\epsi$}
    \Init{$G\gets \emptyset$, $V\gets \uni$, $m \gets \left\lfloor\frac{k}{\ell}\right\rfloor$, add $2k$ dummy elements to the ground set.}
    \For{$i\gets 1$ to $\ell$}{
        $A_{l}\gets G, \forall l \in [\ell]$\;
        \For{$j\gets 1$ to $m$}{\label{line:gdtwo-for-2-start}
            \For{$l\gets 1$ to $\ell$}{
                $a \gets \argmax_{x\in V}\marge{x}{A_{l}}$\;
                % $\delta_{l, j} \gets \marge{a_{l, j}}{A_l}$\;
                $A_{l}\gets A_{l}+a$, $V\gets V-a$ \hspace*{-0.7em}\;
                % \tcc*[r]{If $\marge{a_{l, j}}{A_{l}} < 0$, a dummy element is added instead.}
            }
        }\label{line:gdtwo-for-2-end}
        % $A_l'\gets \left\{a_{l, j}: j\in \argmax\limits_{I\subseteq [m], |I| = m-1} \sum\limits_{j\in I}\delta_{l, j}\right\}$\;
        $G\gets$ a random set in $\{A_l\}_{l\in [\ell]}$\;
    }
    \Return{$G$}
    \caption{A simplified \itg with a randomized $1/e$ approximation ratio and $\oh{ nk\ell }$ query complexity. }
    \label{alg:gdtwo}
\end{algorithm}
\begin{restatable}{theorem}{thmgdtwo}\label{thm:gdtwo}
With input instance $(f, k, \ell, \epsi)$
such that $\ell =\oh{\epsi^{-1}} \ge \frac{2}{e\epsi}$ and $k \ge \frac{2(e\ell-2)}{e\epsi-\frac{2}{\ell}}$,
Alg.~\ref{alg:gdtwo} returns a set $G$ with $\oh{kn/\epsi}$ queries
such that $\ex{\ff{G}} \ge \left(1/e-\epsi\right) \ff{O}$.
\end{restatable}

\subsection{Technical Overview of Blended Marginal Gains Strategy for \itg} \label{sec:greedy-blend}
In \itg, each iteration maintains $\ell+1$ \rev solution families, \color{black} 
each containing $\ell$ nearly pairwise disjoint sets.
Among these \rev solution families\color{black}, only one is \textit{right} which
satisfies the following key inequality:
\begin{align}
    \marge{O}{A_{u, i}} \le \ell\marge{A_{u, i}}{G}, \forall i \in [\ell],\label{inq:gdtwo-itg}
\end{align}
where $G$ is the intermediate solution at the start of this iteration,
and $A_{u,i}$ is the $i$-th solution in the $u$-th 
\rev solution families at the end of this iteration\color{black}.

The right \rev solution families \color{black} guarantees that elements added 
in the first round do not belong to $O\setminus A_{u, i}$ for any $i \in [\ell]$.
This property enables a partition of the optimal solution $O$
into $k/\ell$ subsets of size $\ell$,
where each subset’s marginal gain is dominated by a corresponding element in $A_{u, i}$.
Consequently, $\marge{O}{A_{u, i}}$ depends solely on $\marge{A_{u, i}}{G}$.

In what follows,
we introduce a novel blended marginal gains approach to analyze \itg using only
a single interlaced greedy step (Alg.~\ref{alg:gdtwo}).
This approach leverages a mixture of marginal gains across solutions 
to derive tighter bounds for each  $\marge{O}{A_l}$. 
The analysis proceeds in four steps:

\textbf{Step 1: Partitioning the Optimal Solution $O$.}
Our analysis begins by establishing a correspondence 
between the algorithm's solutions and 
partitions of the optimal set $O$.
Claim~\ref{claim:par-A} provides the foundation for this pairing:
\begin{restatable}{claim}{claimParA}\label{claim:par-A}
At an iteration $i$ of the outer for loop in Alg.~\ref{alg:gdtwo},
let $G_{i-1}$ be $G$ at the start of this iteration,
and $A_{l}$ be the set at the end of this iteration,
for each $l\in [\ell]$.
% Add dummy elements to $O\setminus G_{i-1}$ until its size equals $k$.
The set $O\setminus G_{i-1}$ can then be split into $\ell$ pairwise disjoint sets $\{O_1, \ldots, O_\ell\}$
such that $|O_l| \le\frac{k}{\ell}$ and $\left(O\setminus G_{i-1}\right) \cap \left(A_{l}\setminus G_{i-1}\right) \subseteq O_l$, for all $l \in [\ell]$.
\end{restatable}
This partition enables us to decompose the marginal gains of $O$ 
with respect to each solution $A_l$. 
Specifically, we express the total marginal gain as:
\vspace*{-1em}
{\small\begin{align*}
&\sum_{l\in [\ell]}\marge{O}{A_{l}} \le \sum_{l\in [\ell]} \sum_{i\in [\ell]} \marge{O_i}{A_{l}} \tag{Proposition~\ref{prop:sum-marge}}\\
&= \sum_{1\le l_1 < l_2 \le \ell} \left(\marge{O_{l_1}}{A_{l_2}}+\marge{O_{l_2}}{A_{l_1}}\right)\\
&\hspace*{2em}+\sum_{l\in [\ell]} \marge{O_l}{A_{l}}. \numberthis \label{inq:gdtwo-par}
\end{align*}}
The decomposition consists of two types of terms:
1) self-interaction term $\marge{O_l}{A_{l}}$ for each $l\in [\ell]$,
and 2) cross-interaction terms $\marge{O_{l_1}}{A_{l_2}}+\marge{O_{l_2}}{A_{l_1}}$
for each $1\le l_1 < l_2 \le \ell$.
Below we establish upper bounds for each term type, 
with detailed analysis to follow.
\begin{restatable}{lemma}{lemmaparA}\label{lemma:par-A}
Fix on $G_{i-1}$ for an iteration $i$ of the outer for loop in Alg.~\ref{alg:gdtwo}.
Following the definition in Claim~\ref{claim:par-A}, it holds that
{\small\begin{align*}
\text{1) }&\marge{A_{l}}{G_{i-1}} \ge \marge{O_{l}}{A_{l}}, \forall 1\le l \le \ell,\\
\text{2) }&\left(1+\frac{1}{m}\right)\left(\marge{A_{l_1}}{G_{i-1}} + \marge{A_{l_2}}{G_{i-1}}\right)\\
&\ge \marge{O_{l_2}}{A_{l_1}} + \marge{O_{l_1}}{A_{l_2}}, \forall 1\le l_1< l_2 \le \ell.
\end{align*}}
\end{restatable}

\textbf{Step 2: Self-Interaction Term Bounding.}
The partition from Claim~\ref{claim:par-A} immediately yields our first bound.
For any $l\in [\ell]$, elements in $O_l\setminus A_l$
were available but not selected by the greedy procedure, 
directly implying the first required bound.

\textbf{Step 3: Cross-Interaction Term Bounding.}
The primary technical challenge lies in effectively 
bounding the cross-interaction terms
$\marge{A_{l_1}}{G_{i-1}} + \marge{A_{l_2}}{G_{i-1}}$.
Rather than relying on original greedy selection bounds, 
we develop a more sophisticated approach through 
our blended marginal gains technique, 
formalized in the following proposition:
\begin{proposition}[Blended Marginal Gains]\label{prop:blend}
For any submodular function $f:2^{\uni}\to \reals$ and $S, T, O \in \uni$,
Let $S_i$ be a prefix of $S$ with size $i$
such that $S_i \subseteq O$.
Similarly, define $T_j$.
It satisfies that,
{\small\begin{align}
    \marge{O}{T} &\le \marge{S_i}{T} + \marge{O\setminus S_i}{T},\\
    \marge{O}{S} &\le \marge{T_j}{S} + \marge{O\setminus T_j}{S},
\end{align}}
\end{proposition}
The key insight involves strategically partitioning $O$ into two components:
1) a subset keeps the original greedy selection bound,
and 2) a residual subset where we apply submodularity only. 

Applying this proposition to each solution pair ($A_{l_1}$, $A_{l_2}$)
with carefully chosen prefixes yields
the second inequality in Lemma~\ref{lemma:par-A}.

\textbf{Step 4: Final Composition.}
By applying Inequality~\eqref{inq:gdtwo-par} and Lemma~\ref{lemma:par-A}, 
we derive a result analogous to Inequality~\eqref{inq:gdtwo-itg} achieved by \itg,
\rev
\begin{align}
    \sum_{l\in [\ell]}\marge{O}{A_{l}} \le \ell\left(1+\frac{1}{m}\right)\sum_{l\in [\ell]}\marge{A_{l}}{G_{i-1}}.
\end{align} 
\color{black}
This forms the key property necessary to establish the $1/e-\epsi$ approximation ratio
while requiring only a single pool of solutions.
The detailed analysis of the approximation ratio is provided in Appendix~\ref{apx:greedy-1/e}.

\section{Sublinear Adaptive Algorithms}\label{sec:ptg}

In this section, we present the main subroutine for our parallel algorithms, 
\ptgone (\ptgoneshort, Alg.~\ref{alg:ptgone}).
A single execution of \ptgoneshort achieves an 
approximation ratio of $1/4-\epsi$ with high probability,
while repeatedly running \ptgoneshort, as in \ptgtwo (\ptgtwoshort, Alg.~\ref{alg:ptgtwo} in 
Appendix~\ref{apx:ptgtwo}), 
guarantees a randomized approximation ratio of $1/e-\epsi$.
Below, we outline the theoretical guarantees,
with the detailed analysis provided in Appendix~\ref{apx:ptg}.
\begin{restatable}{theorem}{thmptgone}\label{thm:ptgone}
With input $(f, k, 2, \frac{\epsi M}{k}, \epsi)$,
where $M = \max_{x\in \uni} \ff{x}$,
\ptgoneshort (Alg.~\ref{alg:ptgone}) returns $\{A_1', A_2'\}$
with $\oh{\epsi^{-4}\log(n)\log(k)}$ adaptive rounds and $\oh{\epsi^{-5}n\log(n)\log(k)}$ queries with a probability of $1-1/n$.
It satisfies that $\max\{\ff{A_1'}, \ff{A_2'}\}\ge (1/4-\epsi)\ff{O}$.
\end{restatable}

\begin{restatable}{theorem}{thmptgtwo}\label{thm:ptgtwo}
With input $(f, k, \epsi)$ such that
$\ell = \oh{\epsi^{-1}}\ge \frac{4}{e\epsi}$
and $k \ge \frac{(2-\epsi)^2\ell}{e\epsi\ell-4}$,
\ptgtwoshort (Alg.~\ref{alg:ptgtwo}) returns $G$
such that $\ex{\ff{G}} \ge (1/e-\epsi)\ff{O}$
with $\oh{\epsi^{-5}\log(n)\log(k)}$ adaptive rounds and $\oh{\epsi^{-6}n\log(n)\log(k)}$ queries with a probability of $1-\oh{1/(\epsi n)}$.
\end{restatable}
The remainder of this section is organized as follows.
In Section~\ref{sec:subroutine}, we present the key subroutines 
that form the building blocks of \ptgoneshort.
Section~\ref{sec:ptg-overview} then provides a high-level overview of the algorithm,
% In Section~\ref{sec:ptg-overview}, we provide a high-level overview of the algorithm.
% Section~\ref{sec:subroutine} then presents the key subroutines 
% that form the building blocks of \ptgoneshort.
and introduces several critical properties that must be preserved throughout the process.

\subsection{Subroutines for \ptgoneshort}\label{sec:subroutine}
\textbf{\dist} (Alg.~\ref{alg:dist}) constructs pairwise disjoint subsets 
$\left\{\mathcal V_l: l\in [\ell]\right\}$ 
from candidate pools $\{V_l: l\in [\ell]\}$,
preparing for the subsequent threshold sampling phase.
This crucial preprocessing step ensures that 
all elements added to the solution maintain 
disjointness across different solution sets.
We provide its theoretical guarantee in Lemma~\ref{lemma:dist}.
\begin{lemma}\label{lemma:dist}
With input $\{V_l\}_{l\in [\ell]}$, where $|V_l|\ge 2\ell$ for each $l\in [\ell]$,
\dist returns $\ell$ pairwise disjoint sets $\{\mathcal V_l\}_{l\in [\ell]}$ \st
$\mathcal V_l\subseteq V_l$ and $|\mathcal V_j| \ge \frac{|V_j|}{2\ell}$.
\end{lemma}
\begin{algorithm}[ht]
\Fn{\dist($\{V_l: l\in [\ell]\}$)}{
    \KwIn{$V_1, V_2, \ldots, V_\ell \subseteq \uni$}
    \Init{$\mathcal V_1, \mathcal V_2, \ldots, \mathcal V_\ell \gets \emptyset$, $I\gets [\ell]$}
    \For{$i\gets 1$ to $\ell$}{
        $j\gets\argmin_{j\in I}|V_j|$\label{line:dis-index}\;
        $\mathcal V_j \gets$ randomly select $\left\lfloor\frac{|V_j|}{\ell}\right\rfloor$ elements in $V_j\setminus \left(\bigcup_{l \in [\ell]}\mathcal V_j\right)$ \label{line:dis-select}\;
        $I\gets I-j$\;
    }
    \Return{$\left\{\mathcal V_l: l\in [\ell]\right\}$}
    }
\caption{Return $\ell$ pairwise disjoint subsets where $|\mathcal V_j| \ge \frac{|V_j|}{2\ell}$ for any $j \in [\ell]$ if $|V_j|\ge 2\ell$}
\label{alg:dist}
\end{algorithm}
\begin{algorithm}[h]
\Fn{\prefix($f, \mathcal V, s, \tau, \epsi$)}{
    \KwIn{evaluation oracle $f:2^{\uni} \to \reals$, maximum size $s$, threshold $\tau$, error $\epsi$, candidate pool $\mathcal V$ where $\marge{x}{\emptyset} \ge \tau $ for any $ x\in \mathcal V$}
    \Init{$B[1:s]\gets [\textbf{none}, \ldots, \textbf{none}]$}
    $\mathcal V \gets\left\{v_1, v_2, \ldots\right\} \gets \textbf{random-permutation}(\mathcal V)$\label{line:prefix-permute}\;
    \For{$i\gets 1$ to $s$ in parallel}{
        $T_{i-1} \gets \left\{v_1, \ldots, v_{i-1}\right\}$\;
        \lIf{$\marge{v_i}{T_{i-1}} \ge \tau$}{$B[i]\gets \textbf{true}$}\label{line:prefix-B-true}
        \lElseIf{$\marge{v_i}{T_{i-1}} < 0$}{$B[i]\gets \textbf{false}$}\label{line:prefix-B-false}
    }
    $i^*\gets \max \{i: \#\text{\textbf{true}s in }B[1:i] \ge (1-\epsi) i\}$\label{line:prefix-istar}\;
    \Return{$i^*$, $B$}
}
\caption{Select a prefix of $\mathcal V$ \st its average marginal gain is greater than $(1-\epsi)\tau$, and with a probability of $1/2$, more than an $\epsi/2$-fraction of $\mathcal V$ has a marginal gain less than $\tau$ relative to the prefix.}
\label{alg:prefix}
\end{algorithm}
\textbf{\prefix} (Alg.~\ref{alg:prefix}) serves as the fundamental building block 
for threshold sampling procedure.
It achieves two critical objectives simultaneously:
1) it identifies a prefix $T_{i^*}$ that provides 
sufficient marginal gain to the solution, 
and 2) with probability at least $1/2$, it guarantees that 
a constant fraction of candidate elements have low marginal gain ($<\tau$)
relative to the augmented solution.

Our implementation of \prefix follows Lines 8-15 of \ts~\citep{Chen2024},
and consequently inherits similar theoretical guarantees (Lemma 4 and 5) as follows.
\begin{lemma}\label{lemma:prefix-filter}
In \prefix, given $\mathcal V$ after \textbf{random-permutation} in Line~\ref{line:prefix-permute},
let $D_i = \left\{x\in \mathcal V: \marge{x}{T_i} < \tau\right\}$.
It holds that $|D_0|=0$, $|D_{|\mathcal V|}| = |\mathcal V|$, and $|D_{i-1}|\le |D_i|$.
\end{lemma}
\begin{lemma}\label{lemma:prefix-prob}
In \prefix, following the definition of $D_i$ in Lemma~\ref{lemma:prefix-filter},
let $t = \min\{i: |D_i| \ge \epsi |\mathcal V|/2\}$.
It holds that $\prob{i^* < \min\{s,t\}} \le 1/2$.
\end{lemma}
% As formalized in Lemma~\ref{lemma:prefix-filter},
% $D_i$ contains the elements in $\mathcal V$
% which can be filtered out by threshold value $\tau$
% regarding the prefix $T_i$.
% Therefore, Lemma~\ref{lemma:prefix-prob} indicates that,
% with a probability of at least $1/2$,
% $i^* = s$ or
% more than $\epsi/2$-fraction of $\mathcal V$
% can be filtered out if prefix $T_{i^*}$ is added to the solution.

\begin{algorithm}[ht]
\Fn{\update($f, V, \tau, \epsi$)}{
\KwIn{evaluation oracle $f:2^{\uni} \to \reals$, candidate set $V$, threshold value $\tau$, error $\epsi$}
    \For{$j\gets 1$ to $\ell$ in parallel}{
        $V \gets \left\{x\in V : \marge{x}{\emptyset} \ge \tau\right\}$ \label{line:update-filter}\;
        \While{$|V| = 0$}{
            $\tau \gets (1-\epsi)\tau$\;
            $V \gets \left\{x\in \uni : \marge{x}{\emptyset} \ge \tau\right\}$\;
        }
    }
    \Return{$V, \tau$}
}
\caption{Update candidate set $V$ with threshold value $\tau$}
\label{alg:update}
\end{algorithm}
\textbf{\update} maintains the candidate set $V$ for a solution using threshold $\tau$.
If $V$ becomes empty, it decrease the threshold and regenerate $V$ until
it is not empty.
This is a common component of threshold sampling algorithms.

\subsection{Algorithm Overview}\label{sec:ptg-overview}
\ptgoneshort synthesizes threshold sampling techniques 
with interlaced solution construction 
to achieve both parallel efficiency and strong approximation guarantees. 
It begins by initializing $\ell$ empty solutions $\{A_j: j\in [\ell]\}$,
with corresponding threshold values set to the maximum singleton marginal gain
$M = \max_{x\in \uni}\marge{x}{\emptyset}$. 

At each iteration, \ptgoneshort dynamically switches 
between two distinct operational modes based on candidate set sizes.
When any candidate set $V_j$ (maintained by \update, Alg.~\ref{alg:update})
contains fewer than $2\ell$ elements,
the algorithm enters an alternating addition phase 
where single elements are sequentially distributed across solutions.
This preserves the alternating selection property crucial for 
maintaining approximation guarantees 
while ensuring progress when candidate pools are limited.

For cases where all candidate sets contain sufficient number of elements 
($|V_j| \ge 2\ell$ for each $j \in [\ell]$),
\ptgoneshort employs an efficient batch processing approach.
First, \dist constructs pairwise disjoint candidate subsets, 
enabling parallel processing while maintaining solution quality.
\prefix then simultaneously looks for blocks with high-marginal-gain elements.
At last, the block size and the elements added to each solution are 
carefully chosen to maintain both solution quality parity 
and the equivalent alternating selection effect.

This dual-mode architecture combines the strengths of threshold sampling 
and interlaced greedy methods. 
The threshold sampling components (\update and \prefix) 
ensure efficient element filtering and geometric threshold reduction, 
while the interlaced construction maintains 
the crucial alternating marginal gain properties.
This allows \ptgoneshort to achieve sublinear adaptivity
without compromising its approximation guarantee.

Below, we further analyze three fundamental properties 
that the algorithm must maintain throughout execution.

\begin{algorithm*}[h]
\Fn{\ptgone($f, m, \ell, \tau_{\min}, \epsi$)}{
    \KwIn{evaluation oracle $f:2^{\uni} \to \reals$, 
    constraint $m$, constant $\ell$,
    minimum threshold value $\tau_{\min}$, error $\epsi$}
    \Init{$M\gets \max_{x\in \uni} \marge{\{x\}}{\emptyset}$, $I = [\ell]$, $m_0 \gets m$, $A_j\gets A_j'\gets \emptyset$, 
    $\tau_j\gets M$, $V_j \gets \uni, \forall j \in [\ell]$ }
    % \tcp*[h]{$V_j$ contains all good elements outside solutions}
    \While{$I \neq \emptyset$ and $m_0 > 0$}{
        \For(\textcolor{blue}{\tcc*[h]{Update candidate sets with high-quality elements}}){$j\in I$ in parallel\label{line:tgone-update-for-begin}}{
            $\{V_j, \tau_j\} \gets \update(f_{A_j}\restriction_{\uni\setminus\left(\bigcup_{l\in [\ell]}A_l\right)}, V_j\setminus\left(\bigcup_{l\in [\ell]}\right)A_l, \tau_j, \epsi)$\label{line:tgone-update}\;
            \lIf{$\tau_j < \tau_{\min}$}{ $I\gets I-j$}\label{line:tgone-update-for-end}
        }
        \If(\textcolor{blue}{\tcc*[h]{Add 1 element to each solution alternately}}){$\exists i\in I$ \st $|V_i|< 2\ell$}{
            \For{$j\in I$ in sequence\label{line:tgone-for-begin}\label{line:pig-if-start}}{
                \lIf{$|V_j| = 0$}{
                    $\{V_j, \tau_j\} \gets \update(f_{A_j}\restriction_{\uni\setminus\left(\bigcup_{l\in [\ell]}A_l\right)}, V_j, \tau_j, \epsi)$\label{line:tgone-update-2}
                }
                \lIf{$\tau_j < \tau_{\min}$}{ $I\gets I-j$}
                \Else{
                    $x_j\gets $ randomly select one element from $V_j$ \label{line:tgone-select}\;
                    $A_j\gets A_j+x_j, A_j'\gets A_j'+x_j$\label{line:tgone-update-A}\;
                    $V_l\gets V_l-x_j, \forall l\in [\ell]$\;
                }
            }\label{line:tgone-for-end}
            $m_0\gets m_0-1$\;\label{line:pig-if-end}
        }
        \Else(\textcolor{blue}{\tcc*[h]{Add an equal number of elements to each solution}}){
            $\{\mathcal V_l: l\in I\} \gets \dist(\{V_l: l\in I\})$\label{line:tgone-dist}\label{line:pig-else-start} \tcp*[h]{Create pairwise disjoint candidate sets}\;
            $s \gets \min \{m_0, \min\{|\mathcal V_l|: l\in I\}\}$\;
            \lFor{$j\in I$ in parallel}{
                $i^*_j, B_j \gets \prefix(f_{A_j}, \mathcal V_j, s, \tau_j, \epsi)$\label{line:tgone-prefix}
            }
            $i^*\gets \min\{i^*_1, \ldots, i^*_\ell\}$ \label{line:tgone-index}\;
            \For(\textcolor{blue}{\tcc*[h]{Add $i^*$ high-quality elements to each set}}){$j\gets 1$ to $\ell$ in parallel \label{line:tgone-add-begin}}{
                $S_j\gets$ select $i^*$ elements from $\mathcal V_j[1:i^*_l]$ in three passes, prioritizing $B_j[i] = \textbf{true}$, then $B_j[i] = \textbf{none}$, and finally $B_j[i] = \textbf{false}$
                \label{line:tgone-subset}\;
                $S_j'\gets S_j\cap \left\{v_i \in \mathcal V_j: B_j[i]\neq \textbf{false}\right\}$\label{line:tgone-subset-2}\;
                $A_j\gets A_j\cup S_j, A_j'\gets A_j'\cup S_j'$\label{line:tgone-update-A-2}\;
            }
            $m_0\gets m_0-i*$ \label{line:tgone-update-size}\;\label{line:pig-else-end}
        }
    }
    \Return{$\{A_l': l\in [\ell]\}$}
}
\caption{A highly parallelized algorithm with $\oh{\ell^2\epsi^{-2} \log(n)\log\left(\frac{M}{\tau_{\min}}\right)}$ adaptivity
and $\oh{\ell^3\epsi^{-2} n \log(n)\log\left(\frac{M}{\tau_{\min}}\right)}$ query complexity.}
\label{alg:ptgone}
\end{algorithm*}

\subsubsection{Maintaining Alternating Additions during Parallel Algorithms}
\label{sec:ptg-alter}
This property is crucial to interlaced greedy variants 
introduced in prior sections. 
Below, we demonstrate that \ptgoneshort preserves this property.

During an iteration of the while loop in Alg.~\ref{alg:ptgone},
after updating the candidate sets in Lines~\ref{line:tgone-update-for-begin}-\ref{line:tgone-update-for-end},
two scenarios arise. 
In the first scenario, there exists a candidate set satisfies $|V_j| < 2\ell$,
Lines~\ref{line:pig-if-start}-\ref{line:pig-if-end} are executed,
and elements are appended to solutions one at a time in turn.
In this case, the alternating property is maintained immediately.

In the second scenario, Lines~\ref{line:pig-else-start}-\ref{line:pig-else-end} are executed.
Here, a block of elements with average marginal gain approximately exceeding $\tau_j$
is added to each solution $A_j$.
These blocks $S_j$ are of the same size $i^*$ (Line~\ref{line:tgone-subset})
selected from $\mathcal V_j$,
and guaranteed to be pairwise disjoint by 
Lemma~\ref{lemma:dist} (for \dist, Alg.~\ref{alg:dist}).
Crucially, threshold values $\tau_j$ 
remain unchanged during this step.
While a small fraction of elements in the blocks may have marginal gains below $\tau_j$,
the process retains the alternating property at a structural level: 
the uniform block sizes, and disjoint selection mimic the alternating addition of elements, even when processing multiple elements in parallel.

\subsubsection{Ensuring Sublinear Adaptivity Through Threshold Sampling}
The core mechanism for achieving sublinear adaptivity 
lies in iteratively reducing the pool of high-quality candidate elements 
(those with marginal gains above the threshold) 
by a constant factor within a constant number of adaptive rounds. 
This progressive reduction ensures efficient convergence.

At every iteration of the while loop in Alg.~\ref{alg:ptgone}, 
after updating the candidate sets,
if there exists $V_j$ such that $|V_j| < 2\ell$,
the following occurs after the for loop
(Lines~\ref{line:tgone-for-begin}-\ref{line:tgone-for-end}):
If the threshold $\tau_j$ remains unchanged,
one element from $V_j$ is added to the solution.
If $\tau_j$ is reduced,
$V_j$ is repopulated with high-quality elements.
This implies that a $1/(2\ell)$-fraction of $V_j$ is filtered out
after per iteration,
or even further, it becomes empty and
the threshold value is updated.

In the second case, where Lines~\ref{line:pig-else-start}-\ref{line:pig-else-end} are executed,
the algorithm employs \prefix (Alg.~\ref{alg:prefix}) in Line~\ref{line:tgone-prefix},
inspired by \ts~\citep{Chen2024}.
% At each iteration of \ts, a \textit{good prefix} is selected
% from the candidate set,
% consisting of elements outside the current solution
% that have marginal gains greater than
% the threshold value.
% After each addition to the solution, a constant fraction of elements in the candidate set
% is filtered out with constant probability, 
% based on the given threshold value,
% thus ensuring sublinear adaptivity.
% 
% In Alg.~\ref{alg:ptgone},
% we apply the same prefix selection step from \ts as
% \prefix (Alg.~\ref{alg:prefix} in Appendix~\ref{apx:subroutine})
% in Line~\ref{line:tgone-prefix}.
% If the prefix sizes returned for each solution are different,
% their minimum value $i^*$ are chosen and
% subsets of size $i^*$ are added to each solution 
% (Line~\ref{line:tgone-index}-\ref{line:tgone-update-A-2}).
Then, the smallest prefix size $i^*$ is selected in Line~\ref{line:tgone-index}.
For the solution where its corresponding call to \prefix returns $i^*$,
the entire prefix with size $i^*$ is added to it.
This ensures that a constant fraction of elements in $\mathcal V_j$
can be filtered out by Lemma~\ref{lemma:prefix-prob} with probability at least $1/2$.
Moreover, Lemma~\ref{lemma:dist}
guarantees that $|\mathcal V_j| \ge \frac{1}{\ell}|V_j|$
for each candidate set.
As a result, with constant probability,
at least one candidate set will filter out a constant fraction
of the elements.

\subsubsection{Ensuring Most Added Elements Significantly Contribute to the Solutions}
\label{sec:ptg-gain}
In \ts, the selection of a \textit{good prefix} inherently ensures this property immediately.
However, when interlacing $\ell$ threshold sampling processes,
prefix sizes selected in Line~\ref{line:tgone-prefix} by each solution may vary.
To preserve the alternating addition property introduced in Section~\ref{sec:ptg-alter},
subsets of equal size are selected instead of variable-length good prefixes.
This raises the question: \textit{How can a good subset be derived from a good prefix?}
The solution lies in Line~\ref{line:tgone-subset} of Alg.~\ref{alg:ptgone}.

For any $j \in I$, if $i_j^* = i^*$, $S_j$ is directly the good prefix $\mathcal V_j[1:i_j^*]$.
Otherwise, if $i_j^* \ge i^*$,
$i^*$ elements are selected from $\mathcal V_j[1:i_j^*]$ in three sequential passes until the size limit is reached:

\textbf{First pass}: Iterate through the prefix,
selecting those with marginal gains strictly greater than
$\tau_j$ (marked as \text{true} in $B_j$).

\textbf{Second pass}: 
From the remaining elements in the prefix, select those 
with marginal gains between $0$ and $\tau_j$
(marked as \text{none} in $B_j$).

\textbf{Third pass}: Fill any remaining slots with remaining elements from the prefix (marked as \text{false} in $B_j$).

This approach, combined with submodularity, 
ensures that any element marked as \textbf{true} in the selected subset has a marginal gain greater than $\tau_j$.
By prioritizing the addition of these \textbf{true} elements,
the selected subset remain high-quality while adhering to the alternating
addition framework.

%%% Local Variables:
%%% mode: latex
%%% TeX-master: "main.tex"
%%% End:

%% file: 3_exp.tex
\section{Empirical Evaluation}\label{sec:exp}
\begin{figure*}[ht]
    \centering
    \subfigure[er, solution value]{\label{fig:er-val}
    \includegraphics[width=0.31\linewidth]{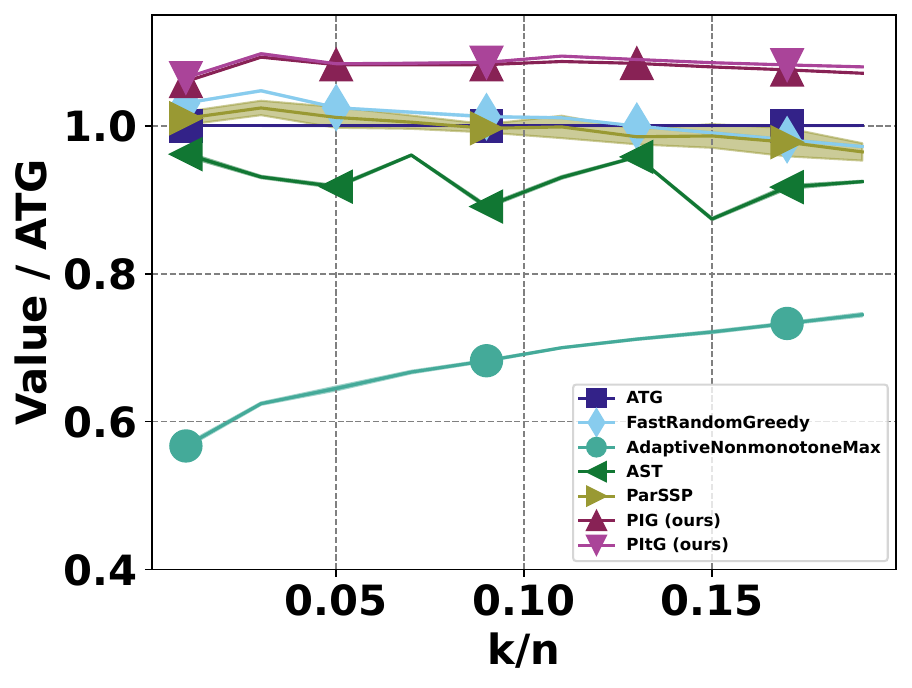}}
    \subfigure[er, query]{\label{fig:er-query}
    \includegraphics[width=0.31\linewidth]{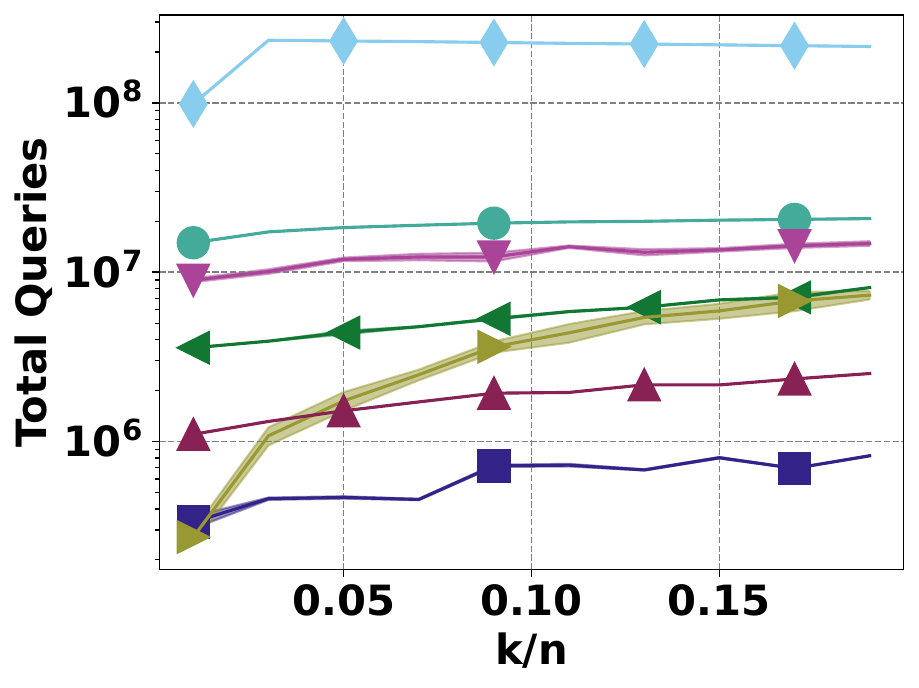}}
    \subfigure[er, round]{\label{fig:er-round}
    \includegraphics[width=0.31\linewidth]{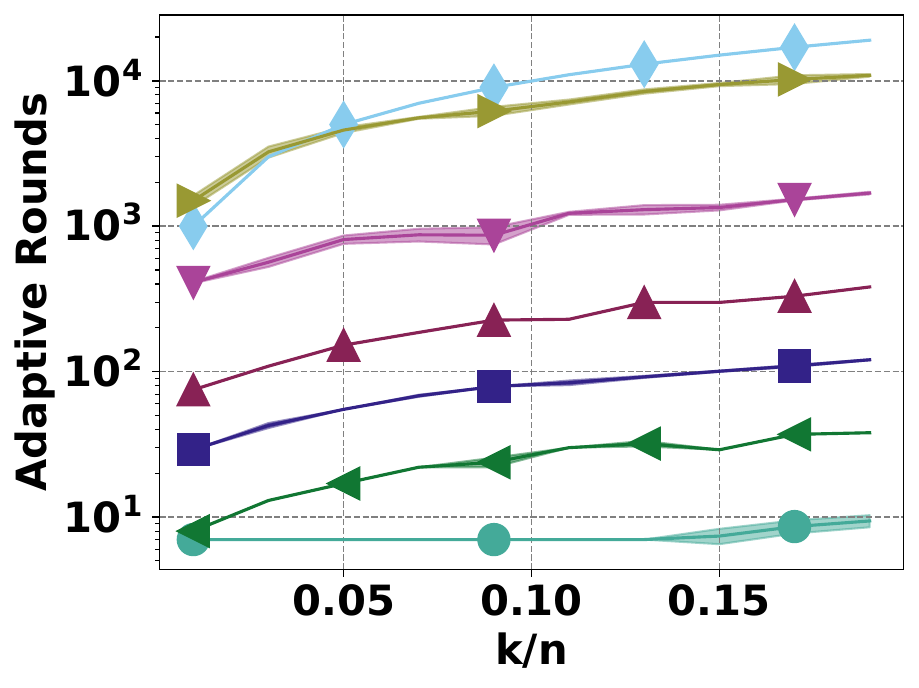}}
    \subfigure[twitch-gamers, solution value]{\label{fig:twitch-val}
    \includegraphics[width=0.31\linewidth]{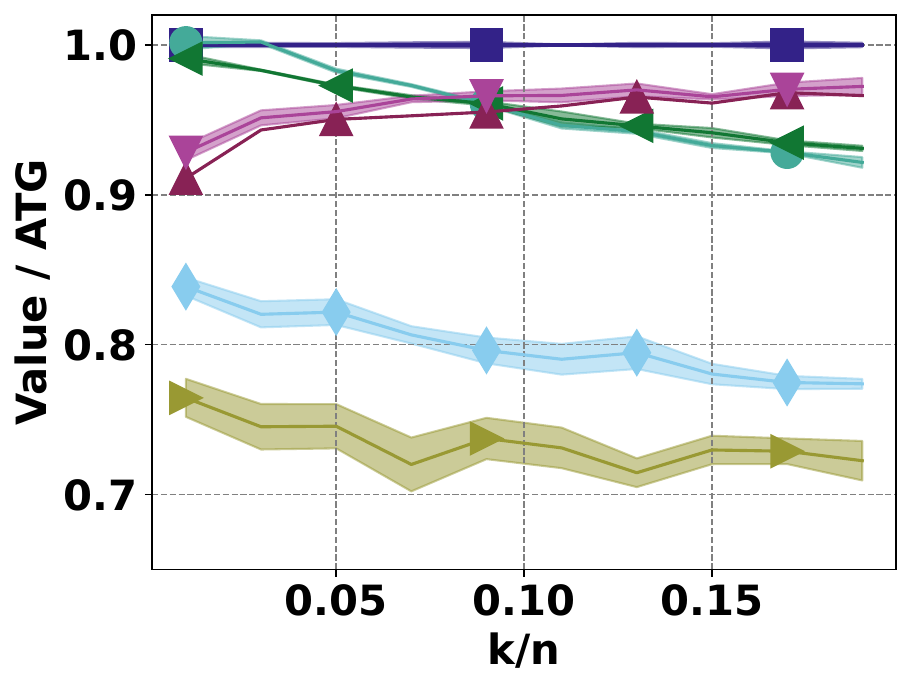}}
    \subfigure[twitch-gamers, query]{\label{fig:twitch-query}
    \includegraphics[width=0.31\linewidth]{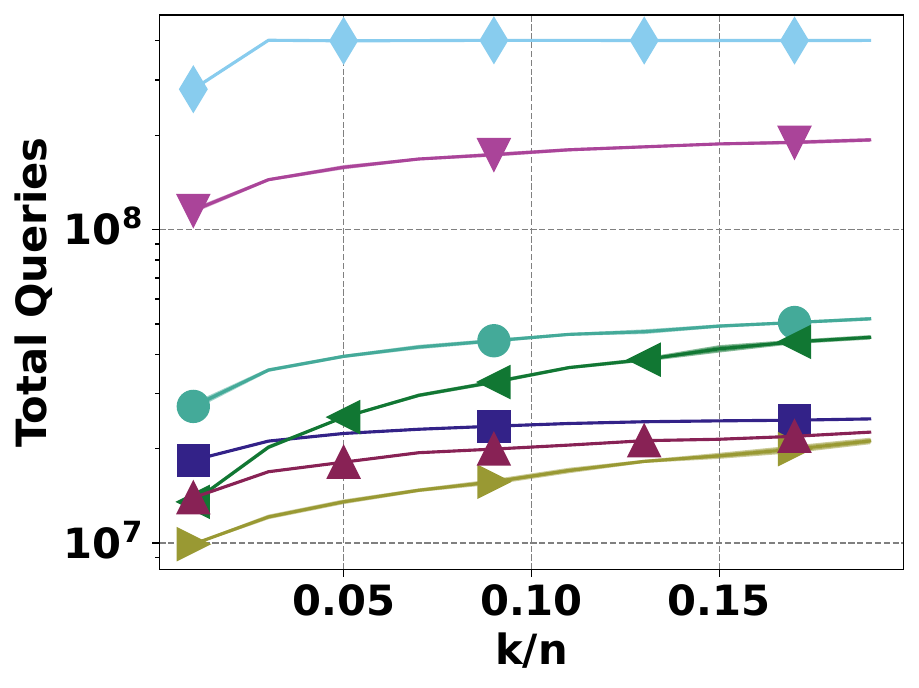}}
    \subfigure[twitch-gamers, round]{\label{fig:twitch-round}
    \includegraphics[width=0.29\linewidth]{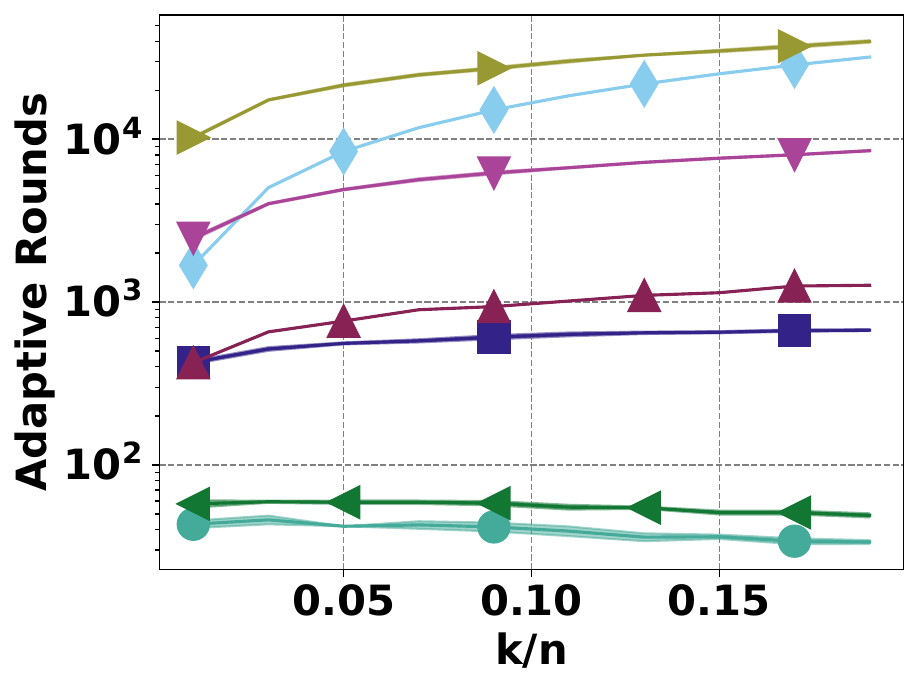}}
    \caption{Results for \maxcut on er with $n=99,997$,
    and \revmax on twitch-gamers with $n=168,114$.}
    \label{fig:er}
\end{figure*}
To evaluate the effectiveness of our algorithms,
we conducted experiments on $2$ applications of \nmon,
comparing its performance to $5$ baseline algorithms.
We measured the objective value (normalized by ATG~\citep{Chen2024}) achieved by each algorithm,
the number of queries made, and the number of adaptive rounds required.
The results showed that our algorithm achieved competitive objective value,
number of queries and adaptive rounds 
compared to nearly linear time algorithms.
\footnote{Our code is available at https://gitlab.com/luciacyx/size-constraints-parallel-algorithms.git.}

\textbf{Applications and Datasets.}
The algorithms were evaluated on $2$ applications: 
Maximum Cut (\maxcut) and Revenue Maximization (\revmax),
with er ($n=99,997$), a synthetic random graph,
web-Google ($n=875,713$),
musae-github ($n=37,700$) and twitch-gamers ($n=168,114$) datasets, the rest of which are
real-world social network datasets from Stanford Large Network Dataset Collection~\citep{snapnets}.
% and Image Summarization (\imgsum) with CIFAR-10~\citep{krizhevsky2009learning} dataset. 
See Appendix~\ref{apx:app} and~\ref{apx:data} for more details.
We provide the result of er and twitch-gamers in the main paper,
while the other results can be found in Appendix~\ref{apx:nmon}.

\textbf{Baselines and their Setups.}
We compare our algorithms with \frg~\citep{DBLP:journals/mor/BuchbinderFS17}, 
\anm~\citep{fahrbach2018non},
\textsc{AST}~\citep{Chen2024}, \textsc{ATG}~\citep{Chen2024},
\rev
and \parssp~\citep{Cui2023} using binary search.
\color{black}
For those algorithms who required an \unc algorithm,
a random subset was employed instead.
For all algorithms, the accuracy parameter $\epsi$ and
the failure probability parameter $\delta$ were both set to $0.1$.
Whenever smaller $\epsi$ and $\delta$ values were specified by the algorithm, 
we substituted them with the input values.
\frg samples each element in the ground set $\uni$ with
a probability $p = 8k^{-1}\epsi^{-2}\log(2\epsi^{-1})$.
If $p > 1$, \randomgreedy in \citet{Buchbinder2014a} was implemented instead.
As for \textsc{Threshold-Sampling} in \anm, $100$ samples were used 
to estimate an indicator.
In the implementation of \ptgtwoshort, $\ell$ was set to $5$.
All randomized algorithms were repeated with $5$ runs,
and the mean is reported. 
The standard deviation is represented by a shaded region in the plots.

\textbf{Overview of Results.}
On the er dataset (Fig.~\ref{fig:er-val}), \ptgoneshort and \ptgtwoshort achieve highest objective values, 
followed by \textsc{ATG}, \frg \rev and \parssp. \color{black}
For the twitch-gamers dataset (Fig.~\ref{fig:twitch-val}), 
\ptgoneshort and \ptgtwoshort 
demonstrate greater robustness, particularly for larger values of $k$.

In terms of query complexity, 
\rev
\parssp exhibits the highest query complexity ($\oh{n\log^2(n) \log(k)}$),
followed by our proposed algorithms ($\oh{n\log(n) \log(k)}$), 
\color{black}
compared to other algorithms ($\oh{n\log(k)}$).
\rev
Nevertheless, \ptgoneshort achieves second-best performance across both datasets. 
The top three performing algorithms consistently include \ptgoneshort, \textsc{ATG} and \parssp.
\color{black}
\ptgtwoshort is slightly better than \anm on er,
and both are more efficient than \frg on the two datasets.
\frg's high query count arises from its requirement for a large number of samples (specifically, $8nk^{-1}\epsi^{-1}\log(2\epsi^{-1})$) at every iteration.
When $8k^{-1}\epsi^{-1}\log(2\epsi^{-1}) \ge 1$, it defaults to executing \rg instead,
which incurs $\oh{nk}$ queries.

Regarding adaptive rounds (Fig.~\ref{fig:er-round} and~\ref{fig:twitch-query}),
the results align with theoretical guarantees.
\anm and \textsc{AST} operate with
$\oh{\log(n)}$ adaptivity
and achieve the best performance.
They are followed by \textsc{ATG}, \ptgoneshort, and \ptgtwoshort
which all achieve $\oh{\log(n)\log(k)}$ adaptivity.

\rev
\parssp exhibits an interesting trade-off: 
while its binary search procedure benefits query complexity, 
it significantly increases adaptive rounds, 
performing worse than even \frg (which requires $k$ adaptive rounds) on the twitch-gamers dataset. 
\color{black}
%%% Local Variables:
%%% mode: latex
%%% TeX-master: "main"
%%% End:

%% file: apx_tech.tex
\section{Technical Lemmata}\label{apx:tech}
% \begin{lemma}[\citep{feige2011maximizing}]\label{lemma:OneRandomSet}
% Let $f:2^{\mathcal{N}} \to \reals$ be submodular. 
% Denote by $A(p)$ a random subset of $A$ where each element
% appears with probability $p$ (not necessarily independently).
% Then
% \[\ex{f(A(p))} \ge (1-p)\cdot f(\emptyset) + p\cdot f(A).\]
% \end{lemma}

\begin{lemma}\label{lemma:val-inq}
    \begin{align*}
        & 1-\frac{1}{x}\le \log(x) \le x-1, &\forall x>0\\
        & 1-\frac{1}{x+1}\ge e^{-\frac{1}{x}} , &\forall x\in \mathbb{R}\\
        & (1-x)^{y-1}\ge e^{-xy}, &\forall xy \le 1
    \end{align*}
\end{lemma}

\begin{lemma}[Chernoff bounds \citep{mitzenmacher2017probability}]\label{lemma:chernoff}
    Suppose $X_1$, ... , $X_n$ are independent binary random variables such that 
    $\prob{X_i = 1} = p_i$. Let $\mu = \sum_{i=1}^n p_i$, and 
    $X = \sum_{i=1}^n X_i$. Then for any $\delta \geq 0$, we have
    \begin{align}
        \prob{X \ge (1+\delta)\mu} \le e^{-\frac{\delta^2 \mu}{2+\delta}}.
    \end{align}
    Moreover, for any $0 \leq \delta \leq 1$, we have
    \begin{align}
        \prob{X \le (1-\delta)\mu} \le e^{-\frac{\delta^2 \mu}{2}}.
    \end{align}
\end{lemma}
\begin{lemma}[\citet{Chen2021}] \label{lemma:indep}
    Suppose there is a sequence of $n$ Bernoulli trials:
    $X_1, X_2, \ldots, X_n,$
    where the success probability of $X_i$
    depends on the results of
    the preceding trials $X_1, \ldots, X_{i-1}$.
    Suppose it holds that $$\prob{X_i = 1 | X_1 = x_1, X_2 = x_2, \ldots, X_{i-1} = x_{i-1} } \ge \eta,$$ where $\eta > 0$ is a constant and $x_1,\ldots,x_{i-1}$ are arbitrary.
  
    Then, if $Y_1,\ldots, Y_n$ are independent Bernoulli trials, each with probability $\eta$ of
    success, then $$\prob {\sum_{i = 1}^n X_i \le b } \le \prob{\sum_{i = 1}^n Y_i \le b }, $$
    where $b$ is an arbitrary integer.
  
    Moreover, let $A$ be the first occurrence of success in sequence $X_i$.
    Then, $$\ex{A} \le 1/\eta.$$
\end{lemma}

\section{Propositions on Submodularity} \label{apx:prop}

\begin{proposition}\label{prop:sum-marge}
Let $\{A_1, A_2, \ldots, A_m\}$ be $m$ pairwise disjoint subsets of $\uni$,
and $B\in \uni$.
For any submodular function $f: 2^\uni \to \reals$,
it holds that
\begin{align*}
\text{1) }&\sum\limits_{i\in [m]} \marge{A_i}{B}\ge \marge{\bigcup\limits_{i\in [m]}A_i}{B},\\
\text{2) }&\sum_{i\in [m]}\ff{B\cup A_i}\ge (m-1)\ff{B}.
\end{align*}
\end{proposition}

\begin{proposition}\label{prop:subset}
Let $A=\{a_1, \ldots, a_m\}$ and $A_i = \{a_1, \ldots, a_i\}$ for all $i\in [m]$.
For any submodular function $f: 2^\uni \to \reals$,
let $B = \argmax\limits_{B\subseteq A, |B| = m-1}\sum\limits_{a_i \in B}\marge{a_i}{A_{i-1}}$.
It holds that
$\ff{B} \ge \left(1-\frac{1}{m}\right)\ff{A}$.
\end{proposition}

\begin{proposition}\label{prop:dif-opt}
For any submodular function $f: 2^\uni \to \reals$,
let $O_1 = \argmax_{S\subseteq \uni, |S|\le k_1} \ff{S}$ and 
$O_2 = \argmax_{S\subseteq \uni, |S|\le k_2} \ff{S}$.
It holds that 
\[\ff{O_1} \ge \frac{k_1}{k_2}\ff{O_2}.\]
\end{proposition}

%% file: apx.tex
\section{Pseudocode and Theoretical Guarantees of \itg~\citep{DBLP:conf/kdd/ChenK23}}\label{apx:pseudocode}
In this section, we provide the original greedy version of 
\itg~\citep{DBLP:conf/kdd/ChenK23}
with its theoretical guarantees.

\begin{algorithm}[H]
	\KwIn{oracle $f:2^{\uni} \to \reals$, constraint $k$, error $\epsi$}
    \Init{$\ell \gets \frac{2e}{\epsi}+1$, $G_0\gets \emptyset$}
    \For{$m\gets 1$ to $\ell$}{
    	$\{a_1, \ldots, a_\ell\}\gets$ top $\ell$ elements in $\uni\setminus G_{m-1}$
		with respect to marginal gains on $G_{m-1}$\;
		\For{$u\gets 0$ to $\ell$ in parallel}{
			\lIf{$u = 0$}{$A_{u, l}\gets G\cup \{a_l\}$, for all $1\le l\le \ell$}
			\lElse{$A_{u, l}\gets G\cup \{a_u\}$, for all $1\le l\le \ell$}
			\For{$j\gets 1$ to $k/\ell-1$}{
				\For{$i\gets 1$ to $\ell$}{
					$x_{j, i} \gets \argmax_{x\in \uni\setminus\left(\bigcup_{l=1}^{\ell}A_{u, l}\right)}\marge{x}{A_{u, i}}$\;
					$A_{u, i}\gets A_{u, i}\cup \{x_{j, i}\}$\;
				}
			}
		}
		$G_m\gets$ a random set in $\{A_{u, i}:1\le i\le \ell, 0\le u\le \ell\}$
    }
    \Return{$G_\ell$}\;
    \caption{$\itg(f,k,\epsi)$: An $1/(e+\epsi)$-approximation algorithm for \sm}
    \label{alg:itg}
\end{algorithm}
\begin{theorem}
Let $\epsi \ge 0$, and $(f, k)$ be an instance of \sm, 
with optimal solution value \opt.
Algorithm \itg outputs a set $G_\ell$ with $\oh{\epsi^{-2}kn}$ queries
such that $\opt \le (e+\epsi)\ex{\ff{G_\ell}}$ with 
probability $(\ell+1)^{-\ell}$, where $\ell = \frac{2e}{\epsi}+1$.
\end{theorem}

\section{Analysis of Simplified \ig(Alg.~\ref{alg:gdone})}
\label{apx:greedy-1/4}
\rev
In this section, we present a detailed approximation analysis of the simplified
\ig (Alg.~\ref{alg:gdone}), 
demonstrating that while the algorithm removes the initial guessing step 
from its original formulation, it maintains the same theoretical approximation guarantees.
\color{black}
\begin{algorithm}[ht]
    \KwIn{evaluation oracle $f:2^{\uni} 
    \to \reals$, constraint $k$}
    \Init{$A\gets B\gets \emptyset$, add $2k$ dummy elements to the ground set}
    \For{$i\gets 1$ to $k$ \label{line:gdone-for-begin}}{
        $a\gets \argmax_{x\in \uni\setminus \left(A\cup B\right)} \marge{x}{A}$\; \label{line:gdone-greedy-A}
        $A\gets A+a$\;
        % \tcc*[r]{If $\marge{a_i}{A_{i-1}} < 0$, a dummy element is added instead.}
        $b\gets \argmax_{x\in \uni\setminus \left(A\cup B\right)} \marge{x}{B}$\; \label{line:gdone-greedy-B}
        $B\gets B+b$\;\label{line:gdone-for-end}
        % \tcc*[r]{If $\marge{b_i}{B_{i-1}} < 0$, a dummy element is added instead.}
    }
    \Return{$S\gets \argmax\{\ff{A}, \ff{B}\}$}
    \caption{A deterministic $1/4$-approximation algorithm with $\oh{ nk }$ queries.}
    \label{alg:gdone}
\end{algorithm}
\begin{restatable}{theorem}{thmgdone}\label{thm:gdone}
With input instance $(f, k)$, Alg.~\ref{alg:gdone} returns a set $S$ with $\oh{kn}$ queries
such that $\ff{S} \ge 1/4 \ff{O}$.
\end{restatable}
\begin{proof}[Proof of Theorem~\ref{thm:gdone}]
% Consider adding dummy elements to the ground set.
% If $|O|$ is less than $k$, add dummy elements to $O$ until $|O| = k$.
% During each iteration of the for-loop,
% if no element in $\uni\setminus (A\cup B)$ has marginal gain greater than 0,
% a dummy element will be added to $A$,
% and similarly for $B$.
% Thus, after the for loop ends, we ensure that $|A| = |B| = k$.
% Moreover, let $A_0 = B_0 = \emptyset$,
% and $A_i$, $B_i$ be $A$, $B$ after $i$-th element is added, respectively.

\textbf{Notation.}
Let $a_i$ be the $i$-th element added to $A$,
and $A_i$ be the set containing the first $i$ elements of $A$.
Similarly, define $b_i$ and $B_i$ for the solution $B$.

Since the two solutions $A_k$ and $B_k$ are disjoint, by submodularity and non-negativity,
\begin{equation*}
\ff{O} \le \ff{O\cup A_k} + \ff{O\cup B_k}.
\end{equation*}
Let $i^* = \max\{i \in [k]: A_i\subseteq O\}$
and $j^* = \max\{j \in [k]: B_j\subseteq O\}$.
If either $i^* = k$ or $j^* = k$,
then $\ff{S} = \ff{O}$.
In the following, we consider $i^* < k$ and $j^* < k$
and discuss two cases of the relationship between $i^*$ and $j^*$ (Fig.~\ref{fig:gdone}).

\textbf{Case 1: $0\le i^*\le j^* < k$; Fig.~\ref{fig:gdone-1}.}
First, we bound $\ff{O\cup A_k}$. 
Consider the set $\tilde{O} = O\setminus \left(A_k \cup B_{i^*}\right)$.
Obviously, it holds that $|\tilde{O}|\le k-i^*$.
Then, order $\tilde{O}$ as $\{o_1, o_2, \ldots\}$ such that $o_i \not \in B_{i+i^*-1}$,
for all $1\le i \le |\tilde{O}|$.
Thus, by the greedy selection step in Line~\ref{line:gdone-greedy-A},
it holds that $\marge{a_{i+i^*}}{A_{i+i^*-1}}\ge \marge{o_i}{A_{i+i^*-1}}$
for all $1\le i \le |\tilde{O}|$.
Then,
\begin{align*}
\ff{O\cup A_k} - \ff{A_k} &\le \marge{B_{i^*}}{A_k} + \marge{\tilde{O}}{A_k}\\
&\le \ff{B_{i^*}} + \sum\limits_{i=1}^{|\tilde{O}|} \marge{o_i}{A_k}\\
&\le \ff{B_{i^*}} + \sum\limits_{i=1}^{|\tilde{O}|} \marge{o_i}{A_{i+i^*-1}}\\
&\le \ff{B_{i^*}} + \sum\limits_{i=i^*+1}^{k} \marge{a_i}{A_{i-1}}
= \ff{B_{i^*}} + \ff{A_k} - \ff{A_{i^*}},
\end{align*}
where the first three inequalities follow from submodularity;
and the last inequality follows from 
$\marge{a_{i+i^*}}{A_{i+i^*-1}}\ge \marge{o_i}{A_{i+i^*-1}}$
for all $1\le i \le |\tilde{O}|$,
and $\marge{a_i}{A_{i-1}}\ge 0$ for all $i \in [k]$.

\begin{figure}[ht]
\centering
\subfigure[$i^* \le j^*$]{\label{fig:gdone-1}\includegraphics[width=0.2\textwidth]{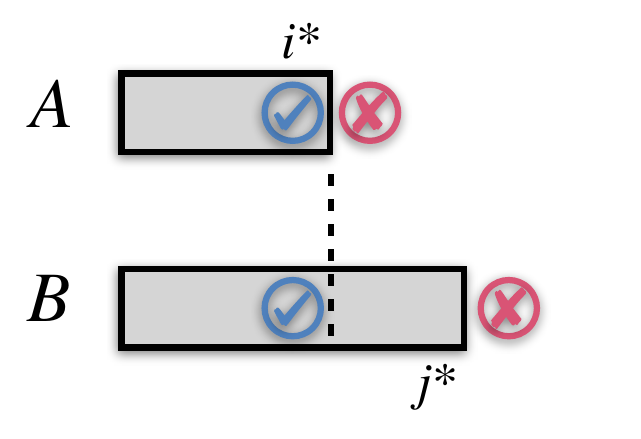}}
\subfigure[$i^* > j^*$]{\label{fig:gdone-2}\includegraphics[width=0.2\textwidth]{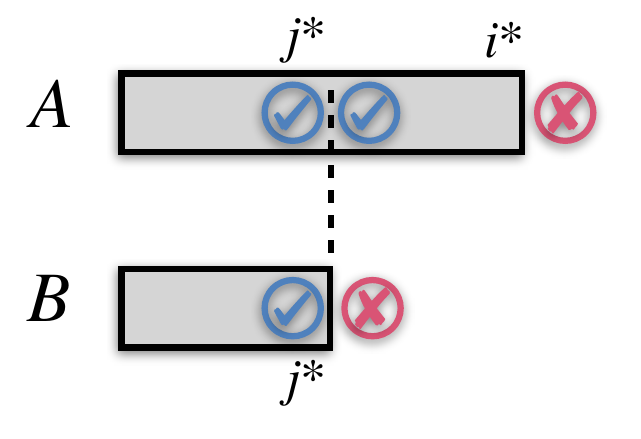}}
    \caption{This figure depicts the components of solution sets $A$ and $B$ in Alg.~\ref{alg:gdone}.
    The black rectangle highlights a sequence of consecutive elements from $O$
    that were added to the solution at the initial.
    Red circles with a cross mark signifies the first element in $A$ or $B$ that is outside $O$.
    }
\label{fig:gdone}
\end{figure}

Next, we bound $\ff{O\cup B_k}$.
Consider the set $\tilde{O} = O\setminus \left(A_{i^*} \cup B_{k}\right)$.
Obviously, it holds that $|\tilde{O}| \le k-i^*$.
Since $i^* = \max\{i \in [k]: A_i\subseteq O\}$,
we know that $a_{i^*+1} \not \in O$.
Thus, we can order $|\tilde{O}|$ as $\{o_1, o_2, \ldots\}$
such that $o_i \not \in A_{i+i^*}$ for all $1\le i \le |\tilde{O}|$.
Then, by the greedy selection step in Line~\ref{line:gdone-greedy-B},
it holds that $\marge{b_{i+i^*}}{B_{i+i^*-1}}\ge \marge{o_i}{B_{i+i^*-1}}$
for all $1\le i \le |\tilde{O}|$.
Following the analysis for $\ff{O\cup A}$,
we get
\begin{align*}
\ff{O\cup B_k} - \ff{B_k} &\le \marge{A_{i^*}}{B_k} + \marge{\tilde{O}}{B_k}\\
&\le \ff{A_{i^*}} + \sum\limits_{i=1}^{|\tilde{O}|} \marge{o_i}{B_k}\\
&\le \ff{A_{i^*}} + \sum\limits_{i=1}^{|\tilde{O}|} \marge{o_i}{B_{i+i^*-1}}\\
&\le \ff{A_{i^*}} + \sum\limits_{i=i^*+1}^{k} \marge{b_i}{B_{i-1}}
= \ff{A_{i^*}} + \ff{B_k} - \ff{B_{i^*}}.
\end{align*}

\textbf{Case 2: $0\le j^* < i^* < k$; Fig.~\ref{fig:gdone-2}.}
First, we bound $\ff{O\cup A_k}$.
Consider the set $\tilde{O} = O\setminus \left(A_{k}\cup B_{j^*}\right)$,
where $|\tilde{O}| \le k-j^*-1$.
By the definition of $j^*$,
we know that $b_{j^*+1}\not\in O$.
Thus, we can order $\tilde{O}$ as $\{o_1, o_2, \ldots\}$
such that $o_i\not \in B_{i+j^*}$ for all $1\le i\le |\tilde{O}|$.
Then, by the greedy selection step in Line~\ref{line:gdone-greedy-A},
it holds that $\marge{a_{i+j^*+1}}{A_{i+j^*}}\ge \marge{o_i}{A_{i+j^*}}$
for all $1\le i\le |\tilde{O}|$.
Following the above analysis, we get
\begin{align*}
\ff{O\cup A_k}-\ff{A_k} &\le \marge{B_{j^*}}{A_k}+\marge{\tilde{O}}{A_k}\\
&\le \ff{B_{j^*}} + \sum\limits_{i=1}^{|\tilde{O}|} \marge{o_i}{A_k}\\
&\le \ff{B_{j^*}} + \sum\limits_{i=1}^{|\tilde{O}|} \marge{o_i}{A_{i+j^*}}\\
&\le \ff{B_{j^*}} + \sum\limits_{i=j^*+2}^{k} \marge{a_i}{A_{i-1}}
= \ff{B_{j^*}} + \ff{A_k}-\ff{A_{j^*+1}}.
\end{align*}

Next, we bound $\ff{O\cup B_k}$.
Consider the set $\tilde{O} = O\setminus \left(A_{j^*+1} \cup B_k\right)$,
where $|\tilde{O}| \le k-j^*-1$.
Then, order $\tilde{O}$ as $\{o_1, o_2, \ldots\}$
such that $o_i\not \in A_{i+j^*}$ for all $1\le i\le |\tilde{O}|$.
By the greedy selection step in Line~\ref{line:gdone-greedy-B},
it holds that $\marge{b_{i+j^*}}{B_{i+j^*-1}}\ge \marge{o_i}{B_{i+j^*-1}}$.
Then,
\begin{align*}
\ff{O\cup B_k} - \ff{B_k}&\le \marge{A_{j^*+1}}{B_k} + \marge{\tilde{O}}{B_k}\\
&\le \ff{A_{j^*+1}}+\sum\limits_{i=1}^{|\tilde{O}|} \marge{o_i}{B_k}\\
&\le \ff{A_{j^*+1}}+\sum\limits_{i=1}^{|\tilde{O}|} \marge{o_i}{B_{i+j^*-1}}\\
&\le \ff{A_{j^*+1}}+\sum\limits_{i=j^*+1}^{k} \marge{b_i}{B_{i-1}} 
= \ff{A_{j^*+1}} + \ff{B_k}-\ff{B_{j^*}}.
\end{align*}

Therefore, in both cases, it holds that
\[\ff{O}\le \ff{O\cup A_k}+\ff{O\cup B_k}\le 2\left(\ff{A_k}+\ff{B_k}\right)\le 4\ff{S}.\]

\end{proof}

\section{\rev Analysis of Simplified \itg (Alg.~\ref{alg:gdtwo}, Section~\ref{sec:gd})}
\label{apx:greedy-1/e}
\rev
In the section, we first provide key Lemmata and their analysis in Appendix~\ref{apx:gdtwo-lemma}
for the case when $k\, \text{mod}\,\ell = 0$.
Then, we conclude with an analysis of approximation ratio in Appendix~\ref{apx:gdtwo-approx}.
\color{black}
% In what follows, we address the scenario where $k\, \text{mod}\,\ell > 0$
% and Alg.~\ref{alg:gdtwo} returns a solution with size smaller than $k$
% in Appendix~\ref{apx:gdtwo-k}.
% We then provide proofs for the relevant Lemmata in Appendix~\ref{apx:gdtwo-lemma},
% and conclude with an analysis of approximation ratio in Appendix~\ref{apx:gdtwo-approx}.
% \subsection{Scenario where $k\, \text{mod}\,\ell > 0$.}\label{apx:gdtwo-k}
% If $k\, \text{mod}\,\ell > 0$, 
% the algorithm returns an approximation solution for a size constraint of 
% $\ell\cdot\left\lfloor\frac{k}{\ell}\right\rfloor$.
% By Proposition~\ref{prop:dif-opt}, it holds that 
% \begin{equation}\label{inq:dif-opt}
% \ff{O'}\ge \ell\cdot\left\lfloor\frac{k}{\ell}\right\rfloor / k \ff{O}
% \ge \left(1-\frac{\ell}{k}\right)\ff{O},
% O' = \argmax\limits_{S\subseteq \uni, |S|\le \ell\cdot\left\lfloor\frac{k}{\ell}\right\rfloor} \ff{S}.
% \end{equation}

\subsection{Proofs of Lemmata for Theorem~\ref{thm:gdtwo}} \label{apx:gdtwo-lemma}
\rev
In what follows, we address the scenario where $k\, \text{mod}\,\ell = 0$
and Alg.~\ref{alg:gdtwo} returns a solution with size exactly $k$.
\color{black}

\textbf{Notation.} 
Let $G_{i-1}$ be $G$ at the start of $i$-th iteration in Alg.~\ref{alg:gdtwo},
$A_l$ be the set at the end of this iteration,
and $a_{l, j}$ be the $j$-th element added to $A_l$ during this iteration.

\lemmaparA*
\begin{proof}[Proof of Lemma~\ref{lemma:par-A}]
\begin{figure}
\centering
\includegraphics[width=0.45\linewidth]{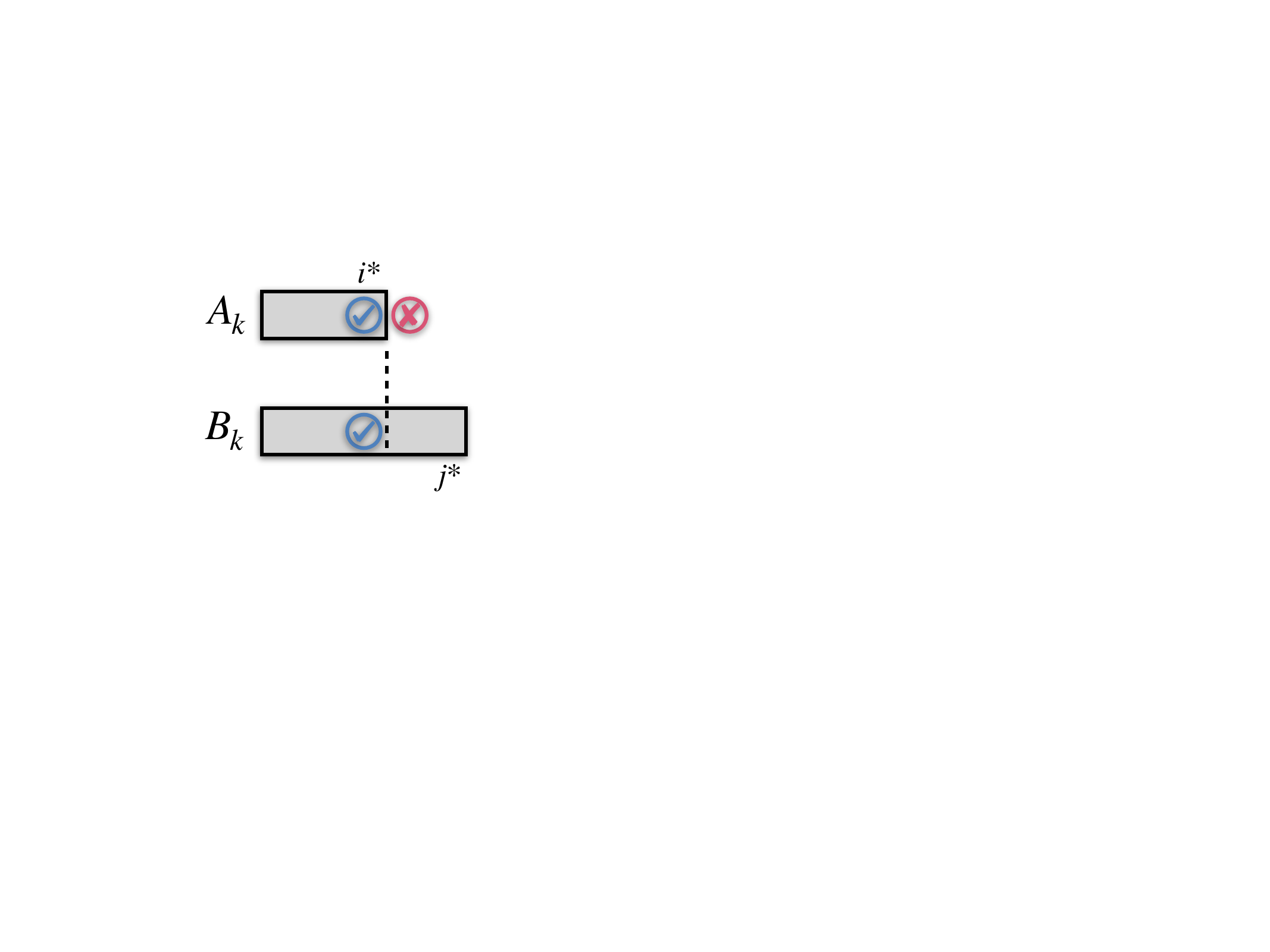}
\caption{This figure depicts the components of the solution sets $A_{l_1}$ and $A_{l_2}$.
A blue circle with a check mark represents an element in $O$,
while a red circle with a cross mark represents an element outside of $O$.
The grey rectangles indicate a sequence of consecutive elements in $O$.
The pink rectangles indicate the corresponding elements used to bound $\marge{O_{l_2}}{A_{l_1}}$ or $\marge{O_{l_1}}{A_{l_2}}$.
It is illustrated that $\marge{O_{l_1}}{A_{l_2}} + \marge{O_{l_2}}{A_{l_1}}\le \marge{A_{l_1}}{G_{i-1}} +  \marge{A_{l_2}}{G_{i-1}}$ under both cases.}
\label{fig:gdtwo}
\end{figure}
Recall that $A_{l, j}$ is $A_l$ after $j$-th element is added to $A_l$ at iteration $i$ of the outer for loop,
and $c_l^* = \max\left\{c\in [m]: A_{l, c}\setminus G_{i-1}\subseteq O_l \right\}$.

First, we prove that the first inequality holds.
For each $l\in [\ell]$, order the elements in $O_l$ as $\{o_1, o_2, \ldots\}$
such that $o_j \not \in A_{l, j-1}$ for any $1\le j \le |O_l|$.
Since each $o_j$ is either in $A_l$ or not in any solution set,
it remains in the candidate pool when $a_{l, j}$ is considered to be added to the solution.
Therefore, it holds that
\begin{equation}\label{inq:itg-1}
\marge{a_{l, j}}{A_{l, j-1}}\ge \marge{o_j}{A_{l, j-1}}.
\end{equation}
Then,
\begin{align*}
\marge{O_{l}}{A_{l}} \le \sum_{o_j\in O_l} \marge{o_j}{A_{l}} \le \sum_{o_j\in O_l} \marge{o_j}{A_{l, j-1}}\le \sum_{j=1}^{m}\marge{a_{l, j}}{A_{l, j-1}} = \marge{A_{l}}{G_{j-1}},
\end{align*}
where the first inequality follows from Proposition~\ref{prop:sum-marge},
the second inequality follows from submodularity,
and the last inequality follows from Inequality~\eqref{inq:itg-1}.

In the following, we prove that the second inequality holds.
For any $1\le l_1\le l_2\le \ell$,
we analyze two cases of the relationship between $c_{l_1}^* $ and $ c_{l_2}^*$ in the following.

% \textbf{Case 1: $c_{l_1}^* = c_{l_2}^* = m$.}
% Then, $O_{l_1} = A_{l_1}\setminus G_{i-1}$ and $O_{l_2} = A_{l_2}\setminus G_{i-1}$.
% By submodularity,
% \[\marge{O_{l_2}}{A_{l_1}} + \marge{O_{l_1}}{A_{l_2}} \le \marge{O_{l_2}}{G_{i-1}} + \marge{O_{l_1}}{G_{i-1}} = \marge{A_{l_1}}{G_{i-1}} + \marge{A_{l_2}}{G_{i-1}}.\]
% Therefore, the lemma holds in this case.

\textbf{Case 1: $c_{l_1}^* \le c_{l_2}^*$; left half part in Fig.~\ref{fig:gdtwo}.}

First, we bound $\marge{O_{l_1}}{A_{l_2}}$.
Since $c_{l_1}^* \le m$, we know that the $(c_{l_1}^*+1)$-th element in $A_{l_1}\setminus G_{i-1}$ is not in $O$.
So, we can order the elements in $O_{l_1}\setminus A_{l_1, c_{l_1}^*}$ as $\{o_1, o_2, \ldots\}$ such that $o_j \not \in A_{l_1, c_{l_1}^*+j+1}$.
(Refer to the gray block with a dotted edge in the top left corner of Fig.~\ref{fig:gdtwo} for $O_{l_1}$.)
Since each $o_j$ is either added to $A_{l_1}$ or not in any solution set,
it remains in the candidate pool when $a_{l_2, c_{l_1}^*+j}$ is considered to be added to $A_{l_2}$.
Therefore, it holds that 
\begin{equation}\label{inq:itg-case2-1}
\marge{a_{l_2, c_{l_1}^*+j}}{A_{l_2, c_{l_1}^*+j-1}} \ge \marge{o_j}{A_{l_2, c_{l_1}^*+j-1}}, \forall 1\le j\le m-c_{l_1}^*.
\end{equation}
Then,
\begin{align*}
\marge{O_{l_1}}{A_{l_2}} &\le \marge{A_{l_1, c_{l_1}^*}}{A_{l_2}}  + \sum_{o_j \in O_{l_1}\setminus A_{l_1, c_{l_1}^*}}\marge{o_j}{A_{l_2}} \tag{Proposition~\ref{prop:sum-marge}}\\
&\le \marge{A_{l_1, c_{l_1}^*}}{G_{i-1}} + \sum_{o_j \in O_{l_1}\setminus A_{l_1, c_{l_1}^*}}\marge{o_j}{A_{l_2, , c_{l_1}^*+j-1}} \tag{submodularity}\\
&\le \ff{A_{l_1, c_{l_1}^*}}-\ff{G_{i-1}} + \sum_{j = 1}^{m-c_{l_1}^*}\marge{a_{l_2, c_{l_1}^*+j}}{A_{l_2, , c_{l_1}^*+j-1}} \tag{Inequality~\eqref{inq:itg-case2-1}}\\
& \le \ff{A_{l_1, c_{l_1}^*}}-\ff{G_{i-1}} + \ff{A_{l_2}} - \ff{A_{l_2, c_{l_1}^*}} 
\end{align*}

Similarly, we bound $\marge{O_{l_2}}{A_{l_1}}$ below.
Order the elements in $O_{l_2}\setminus A_{l_2, c_{l_1}^*}$ as $\{o_1, o_2, \ldots\}$
such that $o_j \not \in A_{l_2, c_{l_1}^* + j}$.
(See the gray block with a dotted edge in the bottom left corner of Fig.~\ref{fig:gdtwo} for $O_{l_2}$.)
Since each $o_j$ is either added to $A_{l_2}$ or not in any solution set,
it remains in the candidate pool when $a_{l_1, c_{l_1}^*+j}$ is considered to be added to $A_{l_2}$.
Therefore, it holds that
\begin{equation}\label{inq:itg-case2-2}
\marge{a_{l_1, c_{l_1}^*+j}}{A_{l_1, c_{l_1}^*+j-1}} \ge \marge{o_j}{A_{l_1, c_{l_1}^*+j-1}}, \forall 1\le j\le m-c_{l_1}^*.
\end{equation}
Then,
\begin{align*}
\marge{O_{l_2}}{A_{l_1}} &\le \marge{A_{l_2, c_{l_1}^*}}{A_{l_1}}  + \sum_{o_j \in O_{l_2}\setminus A_{l_2, c_{l_1}^*}}\marge{o_j}{A_{l_1}} \tag{Proposition~\ref{prop:sum-marge}}\\
&\le \marge{A_{l_2, c_{l_1}^*}}{G_{i-1}} + \sum_{o_j \in O_{l_2}\setminus A_{l_2, c_{l_1}^*}}\marge{o_j}{A_{l_1, , c_{l_1}^*+j-1}} \tag{submodularity}\\
&\le \ff{A_{l_2, c_{l_1}^*}}-\ff{G_{i-1}} + \sum_{j = 1}^{m-c_{l_1}^*}\marge{a_{l_1, c_{l_1}^*+j}}{A_{l_1, , c_{l_1}^*+j-1}} \tag{Inequality~\eqref{inq:itg-case2-1}}\\
& \le \ff{A_{l_1, c_{l_1}^*}}-\ff{G_{i-1}} + \ff{A_{l_1}} - \ff{A_{l_1, c_{l_1}^*}} 
\end{align*}
Thus, the lemma holds in this case.

\textbf{Case 2: $c_{l_1}^* > c_{l_2}^*$; right half part in Fig.~\ref{fig:gdtwo}.}
First, we bound $\marge{O_{l_1}}{A_{l_2}}$.
Order the elements in $O_{l_1}\setminus A_{l_1, c_{l_2}^*+1}$ as $\{o_1, o_2, \ldots\}$ such that $o_j \not \in A_{l_1, c_{l_2}^*+j}$.
(Refer to the gray block with a dotted edge in the top right corner of Fig.~\ref{fig:gdtwo} for $O_{l_1}$.)
Since each $o_j$ is either in $A_{l_1}$ or not in any solution set,
it remains in the candidate pool when $a_{l_2, c_{l_2}^*+j}$ is considered to be added to $A_{l_2}$.
Therefore, it holds that 
\begin{equation}\label{inq:itg-case3-1}
\marge{a_{l_2, c_{l_2}^*+j}}{A_{l_2, c_{l_2}^*+j-1}} \ge \marge{o_j}{A_{l_2, c_{l_2}^*+j-1}}, \forall 1\le j\le m-c_{l_2}^*-1.
\end{equation}
Then,
\begin{align*}
\marge{O_{l_1}}{A_{l_2}} &\le \marge{A_{l_1, c_{l_2}^*+1}}{A_{l_2}}  + \sum_{o_j \in O_{l_1}\setminus A_{l_1, c_{l_2}^*+1}}\marge{o_j}{A_{l_2}} \tag{Proposition~\ref{prop:sum-marge}}\\
&\le \marge{A_{l_1, c_{l_2}^*+1}}{G_{i-1}} + \sum_{o_j \in O_{l_1}\setminus A_{l_1, c_{l_2}^*+1}}\marge{o_j}{A_{l_2, , c_{l_2}^*+j-1}} \tag{submodularity}\\
&\le \ff{A_{l_1, c_{l_2}^*+1}}-\ff{G_{i-1}} + \sum_{j = 1}^{m-c_{l_2}^*-1}\marge{a_{l_2, c_{l_2}^*+j}}{A_{l_2, , c_{l_2}^*+j-1}} \tag{Inequality~\eqref{inq:itg-case3-1}}\\
& \le \ff{A_{l_1, c_{l_2}^*+1}}-\ff{G_{i-1}} + \ff{A_{l_2}} - \ff{A_{l_2, c_{l_2}^*}} 
\end{align*}

Similarly, we bound $\marge{O_{l_2}}{A_{l_1}}$ below.
Since $c_{l_2}^* < c_{l_2}^*$,
we know that the $(c_{l_2}^*+1)$-th element in $A_{l_2}\setminus G_{i-1}$
is not in $O$, which implies that $|O_{l_2}| \le m$.
So, we can order the elements in $O_{l_2}\setminus A_{l_2, c_{l_2}^*}$ as $\{o_1, o_2, \ldots\}$
such that $o_j \not \in A_{l_2, c_{l_2}^* + j}$ for each $1\le j\le m-c_{l_2}^*$.
(See the gray block with a dotted edge in the bottom right corner of Fig.~\ref{fig:gdtwo} for $O_{l_2}$.)

When $1\le j< m-c_{l_2}^*$,
since each $o_j$ is either in $A_{l_2}$ or not in any solution set,
it remains in the candidate pool when $a_{l_1, c_{l_2}^*+i+1}$ is considered to be added to $A_{l_1}$.
Therefore, it holds that
\begin{equation}
\marge{a_{l_1, c_{l_2}^*+j+1}}{A_{l_1, c_{l_2}^*+j}} \ge \marge{o_j}{A_{l_1, c_{l_2}^*+j}}, \forall 1\le j< m-c_{l_2}^*.
\end{equation}
As for the last element $o_{m-c_{l_2}^*}$ in $O_{l_2}\setminus A_{l_2, c_{l_2}^*}$,
we know that $o_{m-c_{l_2}^*}$ is not added to any solution set.
So,
\begin{equation}
\marge{o_{m-c_{l_2}^*}}{A_{l_1}} \le \frac{1}{m}\sum_{j = 1}^m \marge{a_{l_1, j}}{A_{l_1,j-1}}
 = \frac{1}{m} \marge{A_{l_1}}{G_{i-1}}
\end{equation}

Then,
\begin{align*}
\marge{O_{l_2}}{A_{l_1}} &\le \marge{A_{l_2, c_{l_2}^*}}{A_{l_1}}  + \sum_{o_j \in O_{l_2}\setminus A_{l_2, c_{l_2}^*}}\marge{o_j}{A_{l_1}} \tag{Proposition~\ref{prop:sum-marge}}\\
&\le \marge{A_{l_2, c_{l_2}^*}}{G_{i-1}} + \sum_{o_j \in O_{l_2}\setminus A_{l_2, c_{l_2}^*}}\marge{o_j}{A_{l_1, , c_{l_2}^*+j}} \tag{submodularity}\\
&\le \ff{A_{l_2, c_{l_2}^*}}-\ff{G_{i-1}} + \sum_{j = 1}^{m-c_{l_2}^*-1}\marge{a_{l_1, c_{l_2}^*+j+1}}{A_{l_1, , c_{l_2}^*+j}} + \frac{1}{m} \marge{A_{l_1}}{G_{i-1}} \tag{Inequality~\eqref{inq:itg-case3-1}}\\
& \le \ff{A_{l_1, c_{l_2}^*}}-\ff{G_{i-1}} + \ff{A_{l_1}} - \ff{A_{l_1, c_{l_2}^*+1}}+ \frac{1}{m} \marge{A_{l_1}}{G_{i-1}} \tag{$|O_{l_2}|\le m$}
\end{align*}
Thus, the lemma holds in this case.
\end{proof}

\begin{restatable}{lemma}{lemmagdtworec}\label{lemma:gdtwo-rec}
For any iteration $i$ of the outer for loop in Alg.~\ref{alg:gdtwo},
it holds that 

\vspace*{-1em}
\begin{align*}
&\ex{\ff{G_i} - \ff{G_{i-1}}}\ge \frac{1}{\ell+1}\left(1-\frac{1}{m+1}\right) \\
&\cdot \left(\left(1-\frac{1}{\ell}\right)\ex{\ff{O\cup G_{i-1}}} - \ex{\ff{G_{i-1}}}\right)
% \left(1-\frac{1}{m+1}\right)\left((\ell-1)\ff{O\cup G_{i-1}}-\ell\ff{G_{i-1}}\right) \le (\ell+1)\sum_{l\in [\ell]} \marge{A_l}{G_{i-1}}.
\end{align*}
\end{restatable}
\begin{proof}[Proof of Lemma~\ref{lemma:gdtwo-rec}]
Fix on $G_{i-1}$ for an iteration $i$ of the outer for loop in Alg.~\ref{alg:gdtwo}.
Let $A_l$ be the set after for loop in Lines~\ref{line:gdtwo-for-2-start}-\ref{line:gdtwo-for-2-end} ends (with $m$ iterations).
Then,
\begin{align*}
&\sum_{l\in [\ell]}\marge{O}{A_{l}} 
\le \sum_{l\in [\ell]} \marge{O_l}{A_{l}} + \sum_{1\le l_1 < l_2 \le \ell} \left(\marge{O_{l_1}}{A_{l_2}}+\marge{O_{l_2}}{A_{l_1}}\right) \tag{Inequality~\ref{inq:gdtwo-par}}\\
&\le \sum_{l\in [\ell]} \marge{A_{l}}{G_{i-1}} + \sum_{1\le l_1 < l_2 \le \ell}\left(1+\frac{1}{m}\right) \left(\marge{A_{l_1}}{G_{i-1}}+\marge{A_{l_2}}{G_{i-1}}\right)\tag{Lemma~\ref{lemma:par-A}}\\
&\le \ell\left(1+\frac{1}{m}\right) \sum_{l\in [\ell]} \marge{A_{l}}{G_{i-1}}\\
\Rightarrow& \left(\ell+1\right)\left(1+\frac{1}{m}\right)\sum_{l\in [\ell]}\marge{A_l}{G_{i-1}} \ge \sum_{l\in [\ell]}\ff{O\cup A_l} - \ell \ff{G_{i-1}}\\
& \hspace*{15em} \ge \left(\ell-1\right)\ff{O\cup G_{i-1}} - \ell \ff{G_{i-1}},\numberthis \label{inq:itg-rec-1}
\end{align*}
where the last inequality follows from Proposition~\ref{prop:sum-marge}.
Then, it holds that
\begin{align*}
&\exc{\ff{G_i} - \ff{G_{i-1}}}{G_{i-1}}  = \frac{1}{\ell}\sum_{l \in [\ell]}\marge{A_{l}}{G_{i-1}}\\
&\ge \frac{1}{\ell+1}\cdot\frac{m}{m+1}\cdot\left(\left(1-\frac{1}{\ell}\right)\ff{O\cup G_{i-1}} - \ff{G_{i-1}}\right) \tag{Inequality~\eqref{inq:itg-rec-1}}
\end{align*}
By unfixing $G_{i-1}$, the lemma holds.
\end{proof}

\begin{restatable}{lemma}{lemmagdtwodeg}\label{lemma:gdtwo-deg}
For any iteration $i$ of the outer for loop in Alg.~\ref{alg:gdtwo},
it holds that

\vspace*{-1em}
\begin{align*}
\ex{\ff{O\cup G_i}} \ge \left(1-\frac{1}{\ell}\right) \ex{\ff{O\cup G_{i-1}}}.
\end{align*}
\end{restatable}
\begin{proof}[Proof of Lemma~\ref{lemma:gdtwo-deg}]
Fix on $G_{i-1}$ at the beginning of this iteration.
Since $\left\{A_l\setminus G_{i-1}\right\}_{l\in [\ell]}$ 
are pairwise disjoint sets at the end of this iteration,
by Proposition~\ref{prop:sum-marge},
it holds that
\[\exc{\ff{O\cup G_i}}{G_{i-1}} = \frac{1}{\ell}\sum_{l\in [\ell]}\ff{O\cup A_l} \ge \left(1-\frac{1}{\ell}\right)\ff{O\cup G_{i-1}}.\]
Then, by unfixing $G_{i-1}$, the lemma holds.
\end{proof}

\subsection{Proof of Theorem~\ref{thm:gdtwo}}\label{apx:gdtwo-approx}
\thmgdtwo*
\begin{proof}
By Lemma~\ref{lemma:gdtwo-rec} and~\ref{lemma:gdtwo-deg},
the recurrence of $\ex{\ff{G_i}}$ can be expressed as follows,
\begin{align*}
\ex{\ff{G_i}} &\ge \left(1-\frac{1}{\ell+1}\left(1-\frac{1}{m+1}\right)\right)\ex{\ff{G_{i-1}}} + \frac{1}{\ell+1}\left(1-\frac{1}{m+1}\right)\left(1-\frac{1}{\ell}\right)^i\ff{O}\\
&\ge \left(1-\frac{1}{\ell}\right)\ex{\ff{G_{i-1}}} + \frac{1}{\ell+1}\left(1-\frac{1}{m+1}\right)\left(1-\frac{1}{\ell}\right)^i\ff{O}.
\end{align*}
By solving the above recurrence,
\begin{align*}
\ex{\ff{G_{\ell}}} &\ge \frac{\ell}{\ell+1}\left(1-\frac{1}{m+1}\right)\left(1-\frac{1}{\ell}\right)^\ell\ff{O}\\
&\ge \frac{\ell-1}{\ell+1}\left(1-\frac{1}{m+1}\right)e^{-1}\ff{O}\tag{Lemma~\ref{lemma:val-inq}}\\
&\ge \left(1-\frac{2}{\ell}\right)\left(1-\frac{\ell}{k}\right)e^{-1}\ff{O} \tag{$m = \left\lfloor \frac{k}{\ell} \right\rfloor$}\\
&\ge \frac{1}{1-\frac{\ell}{k}}\left(1-\frac{2}{\ell}-\frac{2\ell}{k}+\frac{4}{k}\right)e^{-1}\ff{O}\\
&\ge \frac{1}{1-\frac{\ell}{k}}\left(e^{-1}-\epsi\right)\ff{O}. \tag{$\ell \ge \frac{2}{e\epsi}, k\ge \frac{2(\ell-2)}{e\epsi-\frac{2}{\ell}}$}
\end{align*}
\rev
If $k\, \text{mod}\,\ell = 0$, the approximation ratio holds immediately.

Otherwise, when $k\, \text{mod}\,\ell > 0$, 
the algorithm returns an approximation solution for a size constraint of 
$\ell\cdot\left\lfloor\frac{k}{\ell}\right\rfloor$.
By Proposition~\ref{prop:dif-opt}, it holds that 
\begin{equation}\label{inq:dif-opt}
\ff{O'}\ge \ell\cdot\left\lfloor\frac{k}{\ell}\right\rfloor / k \ff{O}
\ge \left(1-\frac{\ell}{k}\right)\ff{O},
O' = \argmax\limits_{S\subseteq \uni, |S|\le \ell\cdot\left\lfloor\frac{k}{\ell}\right\rfloor} \ff{S}.
\end{equation}
In this case, the approximation ratio still holds.
\color{black}
\end{proof}

\section{\rev Preliminary Warm-Up of Parallel Approaches: Nearly-Linear Time Algorithms}
\label{apx:tg}
\rev
\citet{DBLP:conf/nips/Kuhnle19} introduced a fast version of \ig, replacing the 
greedy procedure with a descending threshold greedy procedure~\citep{DBLP:conf/soda/BadanidiyuruV14}
to achieve a query complexity of $\oh{n\log(k)}$.
This same technique was subsequently employed in \itg~\citep{DBLP:conf/kdd/ChenK23}.
In this section, we present simplified versions of both algorithms,
incorporating the blended marginal gain analysis strategy introduced in Section~\ref{sec:gd}.
These fast algorithms serve as building blocks for the parallel algorithms introduced in this work.
\color{black}

% In this section, we provide the pseudocodes and analysis of
% the simplified fast \ig~\citep{DBLP:conf/nips/Kuhnle19} 
% and \itg~\citep{DBLP:conf/kdd/ChenK23},
% implemented as Alg.~\ref{alg:tgone} and~\ref{alg:tgtwo}, respectively.
% The analysis of these algorithms employ a blending technique
% to eliminate the guessing step in their original version.
\subsection{Simplified Fast \ig with $1/4-\epsi$ Approximation Ratio (Alg.~\ref{alg:tgone})}
\begin{algorithm}[ht]
    \KwIn{evaluation oracle $f:2^{\uni} 
    \to \reals$, constraint $k$, error $\epsi$}
    \Init{$A\gets \emptyset$, $B\gets \emptyset$, $M\gets \max_{x \in \uni}\ff{\{x\}}$,
    $\tau_1\gets M$, $\tau_2\gets M$}
    \For{$i\gets 1$ to $k$}{
        \While{$\tau_1 \ge \frac{\epsi M}{k}$ and $|A| < k$}{
            \If{$\exists a \in \uni\setminus\left(A\cup B\right)$ \st $\marge{a}{A} \ge \tau_1$}{
            $A\gets A+ a$\;
            \textbf{break}\;}
            \lElse{$\tau_1 \gets (1-\epsi)\tau_1$}
        }
        \While{$\tau_2 \ge \frac{\epsi M}{k}$ and $|B| < k$}{
            \If{$\exists b \in \uni\setminus\left(A\cup B\right)$ \st $\marge{b}{B} \ge \tau_2$}{
            $B\gets B+ b$\;
            \textbf{break}\;}
            \lElse{$\tau_2 \gets (1-\epsi)\tau_2$}
        }
    }
    \Return{$S\gets \argmax\{\ff{A}, \ff{B}\}$}
    \caption{A nearly-linear time, $(1/4-\epsi)$-approximation algorithm.}
    \label{alg:tgone}
\end{algorithm}
\begin{restatable}{theorem}{thmtgone}\label{thm:tgone}
With input instance $(f, k, \epsi)$, Alg.~\ref{alg:tgone} returns a set $S$ with $\oh{n\log (k)/\epsi}$ queries
such that $\ff{S} \ge \left(\frac{1}{4}-\epsi\right) \ff{O}$.
\end{restatable}
\begin{proof}
\textbf{Query Complexity.}
Without loss of generality, we analyze the number queries related to set $A$.
For each threshold value $\tau_1$, at most $n$ queries are made to the value oracle.
Since $\tau_1$ is initialized with value $M$, decreases by a factor of $1-\epsi$,
and cannot exceed $\frac{\epsi M}{k}$,
there are at most $\log_{1-\epsi}\left(\frac{\epsi}{k}\right)+1$ possible values of $\tau_1$.
Therefore, the total number of queries is bounded as follows,
\begin{align*}
\#\text{Queries} \le 2\cdot n\cdot \left(\log_{1-\epsi}\left(\frac{\epsi}{k}\right)+1\right)
\le \oh{n\log(k)/\epsi},
\end{align*}
where the last inequality follows from the first inequality in Lemma~\ref{lemma:val-inq}.

\textbf{Approximation Ratio.}
Since $A$ and $B$ are disjoint, by submodularity and non-negativity,
\begin{equation}\label{inq:tgone-1}
\ff{O} \le \ff{O\cup A} + \ff{O\cup B}.
\end{equation}

Let $a_i$ be the $i$-th element added to $A$,
$A_i$ be the first $i$ elements added to $A$,
and $\tau_1^{a_i}$ be the threshold value when $a_i$ is added to $A$.
Similarly, define $b_i$, $B_i$, and $\tau_2^{b_i}$.
Let $i^* = \max\{i \le |A|: A_i \subseteq O\}$
and $j^* = \max\{i \le |B|: B_i \subseteq O\}$.
If either $i^*= k$ or $j^* = k$,
then $\ff{S}= \ff{O}$.
Next, we follow the analysis of Alg.~\ref{alg:gdone} in Section~\ref{apx:greedy-1/4}
to analyze the approximation ratio of Alg.~\ref{alg:tgone}.

\textbf{Case 1: $0\le i^*\le j^* < k$; Fig.~\ref{fig:gdone-1}.}
First, we bound $\ff{O\cup A}$. 
Since $B_{i^*} \subseteq O$, by submodularity
\begin{equation}\label{inq:tgone-3}
\ff{O\cup A} - \ff{A} \le \marge{B_{i^*}}{A} + \marge{O\setminus B_{i^*}}{A}
\le \ff{B_{i^*}}+ \sum_{o\in O\setminus \left(A\cup B_{i^*}\right)}\marge{o}{A}.
\end{equation}
Next, we bound $\marge{o}{A}$ for each $o\in O\setminus \left(A\cup B_{i^*}\right)$.

Let $\tilde{O} = O\setminus \left(A \cup B_{i^*}\right)$.
Obviously, it holds that $|\tilde{O}|\le k-i^*$.
Then, order $\tilde{O}$ as $\{o_1, o_2, \ldots\}$ such that $o_i \not \in B_{i+i^*-1}$,
for all $1\le i \le |\tilde{O}|$.
If $|A| < k$, the algorithm terminates with $\tau_1 < \frac{\epsi M}{k}$.
Thus, it follows that
\begin{equation}\label{inq:tgone-4}
\marge{o_i}{A} < \frac{\epsi M}{k(1-\epsi)}, \forall |A|-i^* < i \le |\tilde{O}|.
\end{equation}

Next, consider tuple $(o_i, a_{i + i^*}, A_{i+i^*-1})$,
for any $1\le i \le \min\{|\tilde{O}|, |A|-i^*\}$.
Since $\tau_1^{a_{i + i^*}}$ is the threshold value when $a_{i + i^*}$ is added,
it holds that 
\begin{equation}\label{inq:tgone-2}
\marge{a_{i + i^*}}{A_{i+i^*-1}} \ge \tau_1^{a_{i + i^*}},
\forall 1\le i\le |A| - i^*.
\end{equation}
Then, we show that $\marge{o_i}{A_{i+i^*-1}} < \tau_1^{a_{i + i^*}}/(1-\epsi)$ always holds
for any $1\le i \le \min\{|\tilde{O}|, |A|-i^*\}$.

Since $M = \max_{x\in \uni}\ff{\{x\}}$,
if $\tau_1^{a_{i+i^*}} \ge M$,
it always holds that $\marge{o_i}{A_{i+i^*-1}} < M/(1-\epsi)\le \tau_1^{a_{i+i^*}}/(1-\epsi)$.
If $\tau_1^{a_{i+i^*}} < M$, 
since $o_i \not \in B_{i+i^*-1}$,
$o_i$ is not considered to be added to $A$ with threshold value $\tau_1^{a_{i+i^*}}/(1-\epsi)$.
Then, by submodularity,
$\marge{o_i}{A_{i+i^*-1}} < \tau_1^{a_{i+i^*}}/(1-\epsi)$.
Therefore, by submodularity and Inequality~\eqref{inq:tgone-2},
it holds that 
\begin{equation}\label{inq:tgone-5}
\marge{o_i}{A} \le \marge{o_i}{A_{i+i^*-1}} < \marge{a_{i + i^*}}{A_{i+i^*-1}}/(1-\epsi), \forall 1\le i \le \min\{|\tilde{O}|, |A|-i^*\}.
\end{equation}

Then,
\begin{align*}
\ff{O\cup A} - \ff{A} &\le \ff{B_{i^*}} + \sum_{o\in O\setminus \left(A\cup B_{i^*}\right)}\marge{o}{A}\\
&\le \ff{B_{i^*}} + \sum_{i = 1}^{\min\{|\tilde{O}|, |A|\}-i^*}\marge{a_{i + i^*}}{A_{i+i^*-1}}/(1-\epsi) + \frac{\epsi M}{1-\epsi}\\
&\le \frac{1}{1-\epsi}\left(\ff{B_{i^*}} + \ff{A}-\ff{A_{i^*}} + \epsi \ff{O}\right),\numberthis \label{inq:tgone-10}
\end{align*}
where the first inequality follows from Inequality~\eqref{inq:tgone-3};
the second inequality follows from Inequalities~\eqref{inq:tgone-4} and~\eqref{inq:tgone-5};
and the last inequality follows from $M\le \ff{O}$.

Second, we bound $\ff{O\cup B}$.
Since $A_{i^*} \subseteq O$, by submodularity
\begin{equation}\label{inq:tgone-6}
\ff{O\cup B} - \ff{B} \le \marge{A_{i^*}}{B} + \marge{O\setminus A_{i^*}}{B}
\le \ff{A_{i^*}}+ \sum_{o\in O\setminus \left(B\cup A_{i^*}\right)}\marge{o}{B}.
\end{equation}
Next, we bound $\marge{o}{B}$ for each $o\in O\setminus \left(B\cup A_{i^*}\right)$.

Let $\tilde{O} = O\setminus \left(B\cup A_{i^*}\right)$.
Obviously, it holds that $|\tilde{O}|\le k-i^*$.
Then, since $a_{i^*+1} \not\in O$,
we can order $\tilde{O}$ as $\{o_1, o_2, \ldots\}$ such that $o_i \not \in A_{i+i^*}$,
for all $1\le i \le |\tilde{O}|$.
If $|B| < k$, the algorithm terminates with $\tau_2 < \frac{\epsi M}{k}$.
Thus, it follows that
\begin{equation}\label{inq:tgone-7}
\marge{o_i}{B} < \frac{\epsi M}{k(1-\epsi)}, \forall |B|-i^* < i \le |\tilde{O}|.
\end{equation}

Next, consider tuple $(o_i, b_{i + i^*}, B_{i+i^*-1})$,
for any $1\le i \le \min\{|\tilde{O}|, |B|-i^*\}$.
Since $\tau_2^{b_{i + i^*}}$ is the threshold value when $b_{i + i^*}$ is added,
it holds that 
\begin{equation}\label{inq:tgone-8}
\marge{b_{i + i^*}}{B_{i+i^*-1}} \ge \tau_2^{b_{i + i^*}},
\forall 1\le i\le |B| - i^*.
\end{equation}
Then, we show that $\marge{o_i}{B_{i+i^*-1}} < \tau_2^{b_{i + i^*}}/(1-\epsi)$ always holds
for any $1\le i \le \min\{|\tilde{O}|, |B|-i^*\}$.

Since $M = \max_{x\in \uni}\ff{\{x\}}$,
if $\tau_2^{b_{i+i^*}} \ge M$,
it always holds that $\marge{o_i}{B_{i+i^*-1}} < M/(1-\epsi)\le \tau_2^{b_{i+i^*}}/(1-\epsi)$.
If $\tau_2^{b_{i+i^*}} < M$, 
since $o_i \not \in A_{i+i^*}$,
$o_i$ is not considered to be added to $B$ with threshold value $\tau_2^{b_{i+i^*}}/(1-\epsi)$.
Then, by submodularity,
$\marge{o_i}{B_{i+i^*-1}} < \tau_2^{b_{i+i^*}}/(1-\epsi)$.
Therefore, by submodularity and Inequality~\eqref{inq:tgone-8},
it holds that 
\begin{equation}\label{inq:tgone-9}
\marge{o_i}{B} \le \marge{o_i}{B_{i+i^*-1}} < \marge{b_{i + i^*}}{B_{i+i^*-1}}/(1-\epsi), \forall 1\le i \le \min\{|\tilde{O}|, |B|-i^*\}.
\end{equation}

Then,
\begin{align*}
\ff{O\cup B} - \ff{B} &\le \ff{A_{i^*}} + \sum_{o\in O\setminus \left(B\cup A_{i^*}\right)}\marge{o}{B}\\
&\le \ff{A_{i^*}} + \sum_{i = 1}^{\min\{|\tilde{O}|, |B|\}-i^*}\marge{b_{i + i^*}}{B_{i+i^*-1}}/(1-\epsi) + \frac{\epsi M}{1-\epsi}\\
&\le \frac{1}{1-\epsi}\left(\ff{A_{i^*}} + \ff{B}-\ff{B_{i^*}} + \epsi \ff{O}\right),\numberthis \label{inq:tgone-11}
\end{align*}
where the first inequality follows from Inequality~\eqref{inq:tgone-6};
the second inequality follows from Inequalities~\eqref{inq:tgone-7} and~\eqref{inq:tgone-9};
and the last inequality follows from $M\le \ff{O}$.

By Inequalities~\eqref{inq:tgone-1},~\eqref{inq:tgone-10} and~\eqref{inq:tgone-11},
it holds that
\begin{align*}
&\ff{O} \le \frac{2-\epsi}{1-\epsi}\left(\ff{A} + \ff{B}\right) + \frac{2\epsi}{1-\epsi}\ff{O}\\
\Rightarrow &\ff{S} \ge \left(\frac{1}{4} - \frac{5}{2(4-2\epsi)}\epsi\right)\ff{O}
\ge \left(\frac{1}{4}-\epsi\right)\ff{O}\tag{$\epsi < 1/2$}
\end{align*}

\textbf{Case 2: $0\le j^* < i^* < k$; Fig.~\ref{fig:gdone-2}.}

First, we bound $\ff{O\cup A}$. 
Since $B_{j^*} \subseteq O$, by submodularity
\begin{equation}\label{inq:tgone-20}
\ff{O\cup A} - \ff{A} \le \marge{B_{j^*}}{A} + \marge{O\setminus B_{j^*}}{A}
\le \ff{B_{j^*}}+ \sum_{o\in O\setminus \left(A\cup B_{j^*}\right)}\marge{o}{A}.
\end{equation}
Next, we bound $\marge{o}{A}$ for each $o\in O\setminus \left(A\cup B_{j^*}\right)$.

Let $\tilde{O} = O\setminus \left(A \cup B_{j^*}\right)$.
Since $i^* > j^*\ge 0$, 
it holds that $|\tilde{O}|\le k-j^*-1$.
Since $b_{j^*+1}\not \in O$,
we can order $\tilde{O}$ as $\{o_1, o_2, \ldots\}$ such that $o_i \not \in B_{i+j^*}$,
for all $1\le i \le |\tilde{O}|$.
If $|A| < k$, the algorithm terminates with $\tau_1 < \frac{\epsi M}{k}$.
Thus, it follows that
\begin{equation}\label{inq:tgone-21}
\marge{o_i}{A} < \frac{\epsi M}{k(1-\epsi)}, \forall |A|-j^*-1 < i \le |\tilde{O}|.
\end{equation}

Next, consider tuple $(o_i, a_{i + j^*+1}, A_{i+j^*})$,
for any $1\le i \le \min\{|\tilde{O}|, |A|-j^*-1\}$.
Since $\tau_1^{a_{i + j^*+1}}$ is the threshold value when $a_{i + j^*+1}$ is added,
it holds that 
\begin{equation}\label{inq:tgone-22}
\marge{a_{i + j^*+1}}{A_{i+j^*}} \ge \tau_1^{a_{i + j^*+1}},
\forall 1\le i\le |A| - j^*-1.
\end{equation}
Then, we show that $\marge{o_i}{A_{i+j^*}} < \tau_1^{a_{i + j^*+1}}/(1-\epsi)$ always holds
for any $1\le i \le \min\{|\tilde{O}|, |A|-j^*-1\}$.

Since $M = \max_{x\in \uni}\ff{\{x\}}$,
if $\tau_1^{a_{i + j^*+1}} \ge M$,
it always holds that $\marge{o_i}{A_{i+j^*}} < M/(1-\epsi)\le \tau_1^{a_{i + j^*+1}}/(1-\epsi)$.
If $\tau_1^{a_{i + j^*+1}} < M$, 
since $o_i \not \in B_{i+j^*}$,
$o_i$ is not considered to be added to $A$ with threshold value $\tau_1^{a_{i + j^*+1}}/(1-\epsi)$.
Then, by submodularity,
$\marge{o_i}{A_{i+j^*}} < \tau_1^{a_{i + j^*+1}}/(1-\epsi)$.
Therefore, by submodularity and Inequality~\eqref{inq:tgone-22},
it holds that 
\begin{equation}\label{inq:tgone-23}
\marge{o_i}{A} \le \marge{o_i}{A_{i+j^*}} < \marge{a_{i + j^*+1}}{A_{i+j^*}}/(1-\epsi), \forall 1\le i \le \min\{|\tilde{O}|, |A|-j^*-1\}.
\end{equation}

Then,
\begin{align*}
\ff{O\cup A} - \ff{A} &\le \ff{B_{j^*}} + \sum_{o\in O\setminus \left(A\cup B_{j^*}\right)}\marge{o}{A}\\
&\le \ff{B_{j^*}} + \sum_{i = 1}^{\min\{|\tilde{O}|, |A|-j^*-1\}}\marge{a_{i + j^*+1}}{A_{i+j^*}}/(1-\epsi) + \frac{\epsi M}{1-\epsi}\\
&\le \frac{1}{1-\epsi}\left(\ff{B_{j^*}} + \ff{A}-\ff{A_{j^*+1}} + \epsi \ff{O}\right),\numberthis \label{inq:tgone-24}
\end{align*}
where the first inequality follows from Inequality~\eqref{inq:tgone-20};
the second inequality follows from Inequalities~\eqref{inq:tgone-21} and~\eqref{inq:tgone-23};
and the last inequality follows from $M\le \ff{O}$.

Second, we bound $\ff{O\cup B}$.
Since $A_{j^*+1} \subseteq O$, by submodularity
\begin{equation}\label{inq:tgone-25}
\ff{O\cup B} - \ff{B} \le \marge{A_{j^*+1}}{B} + \marge{O\setminus A_{j^*+1}}{B}
\le \ff{A_{j^*+1}}+ \sum_{o\in O\setminus \left(B\cup A_{j^*+1}\right)}\marge{o}{B}.
\end{equation}
Next, we bound $\marge{o}{B}$ for each $o\in O\setminus \left(B\cup A_{j^*+1}\right)$.

Let $\tilde{O} = O\setminus \left(B\cup A_{j^*+1}\right)$.
Obviously, it holds that $|\tilde{O}|\le k-j^*-1$.
Then, order $\tilde{O}$ as $\{o_1, o_2, \ldots\}$ such that $o_i \not \in A_{i+j^*}$,
for all $1\le i \le |\tilde{O}|$.
If $|B| < k$, the algorithm terminates with $\tau_2 < \frac{\epsi M}{k}$.
Thus, it follows that
\begin{equation}\label{inq:tgone-26}
\marge{o_i}{B} < \frac{\epsi M}{k(1-\epsi)}, \forall |B|-j^*-1 < i \le |\tilde{O}|.
\end{equation}

Next, consider tuple $(o_i, b_{i + j^*}, B_{i+j^*-1})$,
for any $1\le i \le \min\{|\tilde{O}|, |B|-j^*-1\}$.
Since $\tau_2^{b_{i + j^*}}$ is the threshold value when $b_{i + j^*}$ is added,
it holds that 
\begin{equation}\label{inq:tgone-27}
\marge{b_{i + j^*}}{B_{i+j^*-1}} \ge \tau_2^{b_{i + j^*}},
\forall 1\le i\le |B| - j^*-1.
\end{equation}
Then, we show that $\marge{o_i}{B_{i+j^*-1}} < \tau_2^{b_{i + j^*}}/(1-\epsi)$ always holds
for any $1\le i \le \min\{|\tilde{O}|, |B|-j^*-1\}$.

Since $M = \max_{x\in \uni}\ff{\{x\}}$,
if $\tau_2^{b_{i+j^*}} \ge M$,
it always holds that $\marge{o_i}{B_{i+j^*-1}} < M/(1-\epsi)\le \tau_2^{b_{i+j^*}}/(1-\epsi)$.
If $\tau_2^{b_{i+j^*}} < M$, 
since $o_i \not \in A_{i+j^*}$,
$o_i$ is not considered to be added to $B$ with threshold value $\tau_2^{b_{i+j^*}}/(1-\epsi)$.
Then, by submodularity,
$\marge{o_i}{B_{i+j^*-1}} < \tau_2^{b_{i+j^*}}/(1-\epsi)$.
Therefore, by submodularity and Inequality~\eqref{inq:tgone-27},
it holds that 
\begin{equation}\label{inq:tgone-28}
\marge{o_i}{B} \le \marge{o_i}{B_{i+j^*-1}} < \marge{b_{i + j^*}}{B_{i+j^*-1}}/(1-\epsi), \forall 1\le i \le \min\{|\tilde{O}|, |B|-j^*-1\}.
\end{equation}

Then,
\begin{align*}
\ff{O\cup B} - \ff{B} &\le \ff{A_{j^*+1}} + \sum_{o\in O\setminus \left(B\cup A_{j^*+1}\right)}\marge{o}{B}\\
&\le \ff{A_{j^*+1}} + \sum_{i = 1}^{\min\{|\tilde{O}|, |B|-j^*-1\}}\marge{b_{i + j^*}}{B_{i+j^*-1}}/(1-\epsi) + \frac{\epsi M}{1-\epsi}\\
&\le \frac{1}{1-\epsi}\left(\ff{A_{j^*+1}} + \ff{B}-\ff{B_{i^*}} + \epsi \ff{O}\right),\numberthis \label{inq:tgone-29}
\end{align*}
where the first inequality follows from Inequality~\eqref{inq:tgone-25};
the second inequality follows from Inequalities~\eqref{inq:tgone-26} and~\eqref{inq:tgone-28};
and the last inequality follows from $M\le \ff{O}$.

By Inequalities~\eqref{inq:tgone-1},~\eqref{inq:tgone-24} and~\eqref{inq:tgone-29},
it holds that
\begin{align*}
&\ff{O} \le \frac{2-\epsi}{1-\epsi}\left(\ff{A} + \ff{B}\right) + \frac{2\epsi}{1-\epsi}\ff{O}\\
\Rightarrow &\ff{S} \ge \left(\frac{1}{4} - \frac{5}{2(4-2\epsi)}\epsi\right)\ff{O}
\ge \left(\frac{1}{4}-\epsi\right)\ff{O}\tag{$\epsi < 1/2$}
\end{align*}

Therefore, in both cases, it holds that
\[\ff{S} \ge \left(\frac{1}{4}-\epsi\right)\ff{O} .\]
\end{proof}

\subsection{Simplified Fast \itg with $1/e-\epsi$ Approximation Ratio (Alg.~\ref{alg:tgtwo})}
\begin{algorithm}[ht]
    \KwIn{evaluation oracle $f:2^{\uni} \to \reals$, constraint $k$, error $\epsi$}
    \Init{$G_0\gets \emptyset$, $\epsi'\gets \frac{\epsi}{2}$, $m \gets \left\lfloor\frac{k}{\ell}\right\rfloor$, $\ell\gets \left \lceil\frac{4}{e\epsi'}\right \rceil$,
    $M\gets \max_{x\in \uni} \ff{\{x\}}$}
    \For{$i\gets 1$ to $\ell$}{
        $\tau_l \gets M, \forall l\in [\ell]$\;
        $A_{l}\gets G_{i-1}, \forall l \in [\ell]$\;
        \For{$j\gets 1$ to $m$}{
            \For{$l\gets 1$ to $\ell$}{
                \While{$\tau_l \ge \frac{\epsi' M}{k}$ and $|A_l\setminus G_{i-1}| < m$}{
                    \If{$\exists x \in \uni\setminus\left(\bigcup_{r\in [\ell]} A_r\right)$ \st $\marge{a}{A_l} \ge \tau_l$}{
                    $A_l\gets A_l+ x$\;
                    \textbf{break}\;}
                    \lElse{$\tau_l \gets (1-\epsi')\tau_l$}
                }
            }
        }
        $G_i\gets$ a random set in $\{A_l\}_{l\in [\ell]}$\;
    }
    \Return{$G_\ell$}
    \caption{A nearly-linear time, $(1/e-\epsi)$-approximation algorithm.}
    \label{alg:tgtwo}
\end{algorithm}
\begin{restatable}{theorem}{thmtgtwo}\label{thm:tgtwo}
With input instance $(f, k, \epsi)$
such that $\ell = \oh{\epsi^{-1}}\ge \frac{4}{e\epsi}$
and $k \ge \frac{2(2-\epsi)\ell^2}{e\epsi\ell-4}$,
Alg.~\ref{alg:gdtwo} (Alg.~\ref{alg:tgtwo}) returns a set $G_\ell$ with $\oh{n\log(k)/\epsi^2}$ queries
such that $\ff{G_\ell} \ge \left(1/e-\epsi\right) \ff{O}$.
\end{restatable}
\begin{proof}
When $k\,\text{mod}\,\ell > 0$, the algorithm returns an approximation with a size constraint of 
$\ell\cdot\left\lfloor \frac{k}{\ell}\right\rfloor$, where by Proposition~\ref{prop:dif-opt},
\begin{equation}\label{inq:tgtwo-dif-opt}
\ff{O'} \ge \left(1-\frac{\ell}{k}\right)\ff{O}, 
O' = \argmax\limits_{S\subseteq \uni, |S|\le \ell\cdot \left\lfloor \frac{k}{\ell} \right\rfloor}\ff{S}.
\end{equation}
In the following, we only consider the case where $k\,\text{mod}\,\ell = 0$.

At every iteration of the outer for loop,
$\ell$ solutions are constructed, with each solution being augmented
by at most $k/\ell$ elements.
To bound the marginal gain of the optimal set $O$ on each solution set $A_l$,
we consider partitioning $O$ into $\ell$ subsets.
We formalize this partition in the following claim,
which yields a result analogous to Claim~\ref{claim:par-A} presented in 
Section~\ref{sec:greedy-blend}.
Specifically, the claim states that the optimal set $O$ can be evenly 
divided into $\ell$ subsets,
where each subset only overlaps with only one solution set.
\begin{claim}
At an iteration $i$ of the outer for loop in Alg.~\ref{alg:tgtwo},
let $G_{i-1}$ be $G$ at the start of this iteration,
and $A_{l}$ be the set at the end of this iteration,
for each $l\in [\ell]$.
% Add dummy elements to $O\setminus G_{i-1}$ until its size equals $k$.
The set $O\setminus G_{i-1}$ can then be split into $\ell$ pairwise disjoint sets $\{O_1, \ldots, O_\ell\}$
such that $|O_l| \le\frac{k}{\ell}$ and $\left(O\setminus G_{i-1}\right) \cap \left(A_{l}\setminus G_{i-1}\right) \subseteq O_l$, for all $l \in [\ell]$.
\end{claim}
Next, based on such partition, we introduce the following lemma, 
which provides a bound on the marginal gain of any subset $O_{l_1}$ 
with respect to any solution set $A_{l_2}$,
where $1\le l_1, l_2 \le \ell$.
\begin{lemma}\label{lemma:tg-par-A}
Fix on $G_{i-1}$ for an iteration $i$ of the outer for loop in Alg.~\ref{alg:tgtwo}.
Following the definition in Claim~\ref{claim:par-A}, it holds that
\begin{align*}
\text{1) }&\marge{O_{l}}{A_{l}}\le \frac{\marge{A_{l}}{G_{j-1}}}{1-\epsi'}+\frac{\epsi' M}{(1-\epsi')\ell}, \forall 1\le l \le \ell,\\
\text{2) }&\marge{O_{l_2}}{A_{l_1}} + \marge{O_{l_1}}{A_{l_2}} \le \frac{1}{1-\epsi'}\left(1+\frac{1}{m}\right)\left(\marge{A_{l_1}}{G_{i-1}}+\marge{A_{l_2}}{G_{i-1}}\right)
+\frac{2\epsi' M}{(1-\epsi')\ell}, \forall 1\le l_1 < l_2 \le \ell.
\end{align*}
\end{lemma}
Followed by the above lemma, 
we provide the recurrence of $\ex{\ff{G_i}}$ and $\ex{\ff{O\cup G_i}}$.
\begin{lemma}\label{lemma:tg-recur}
For any iteration $i$ of the outer for loop in Alg.~\ref{alg:tgtwo},
it holds that
\begin{align*}
\text{1) } & \ex{\ff{O\cup G_i}}\ge \left(1-\frac{1}{\ell}\right) \ex{\ff{O\cup G_{i-1}}}\\
\text{2) } & \ex{\ff{G_i} - \ff{G_{i-1}}}
\ge\frac{1}{1+\frac{\ell}{1-\epsi'}}\left(1-\frac{1}{m+1}\right)\left(\left(1-\frac{1}{\ell}\right)  \ex{\ff{O\cup G_{i-1}}} - \ex{\ff{G_{i-1}}} - \frac{\epsi'}{1-\epsi'}\ff{O}\right).
\end{align*}
\end{lemma}
By solving the recurrence in Lemma~\ref{lemma:tg-recur},
we calculate the approximation ratio of the algorithm as follows,
\begin{align*}
&\ex{\ff{G_{i}}}  \ge \left(1-\frac{1}{\ell}\right) \ex{\ff{G_{i-1}}}
+ \frac{1}{1+\frac{\ell}{1-\epsi'}}\left(1-\frac{1}{m+1}\right)\left(\left(1-\frac{1}{\ell}\right)^i - \frac{\epsi'}{1-\epsi'}\right)\ff{O}\\
\Rightarrow& \ex{\ff{G_\ell}} \ge \frac{\ell}{1+\frac{\ell}{1-\epsi'}}\left(1-\frac{1}{m+1}\right)\left(\left(1-\frac{1}{\ell}\right)^\ell - \frac{\epsi'}{1-\epsi'}\left(1-\left(1-\frac{1}{\ell}\right)^\ell\right)\right)\ff{O}\\
&\hspace*{4em} \ge \frac{\ell-1}{1+\frac{\ell}{1-\epsi'}}\left(1-\frac{1}{m+1}\right)\left(e^{-1} - \frac{\epsi'}{1-\epsi'}\left(1-e^{-1}\right)\right)\ff{O}\\
&\hspace*{4em} \ge \frac{1}{1-\frac{\ell}{k}}\left(1-\epsi' - \frac{2}{\ell}\right)\left(1-\frac{\ell}{k}\right)^2\left(e^{-1} - \frac{\epsi'}{1-\epsi'}\left(1-e^{-1}\right)\right) \ff{O}\\
% &\hspace*{4em} \ge \frac{1}{1-\frac{\ell}{k}}\left(1-\epsi' - \frac{2}{\ell}\right)\left(1-\frac{2\ell}{k}\right)\left(e^{-1} - \frac{\epsi'}{1-\epsi'}\left(1-e^{-1}\right)\right) \ff{O}\\
&\hspace*{4em} \ge \frac{1}{1-\frac{\ell}{k}}\left(1-\epsi' - \frac{2}{\ell}-\frac{2(1-\epsi')\ell}{k}\right)\left(e^{-1} - \frac{\epsi'}{1-\epsi'}\left(1-e^{-1}\right)\right) \ff{O}\\
&\hspace*{4em} \ge \frac{1}{1-\frac{\ell}{k}} \left(1-(e+1)\epsi'\right)\left(e^{-1} - \frac{\epsi'}{1-\epsi'}\left(1-e^{-1}\right)\right) \ff{O}\tag{$\ell \ge \frac{2}{e\epsi'}, k \ge \frac{2(1-\epsi')\ell}{e\epsi'-\frac{2}{\ell}}$}\\
&\hspace*{4em} \ge \frac{1}{1-\frac{\ell}{k}} \left(e^{-1}-\epsi\right)\ff{O}\tag{$\epsi' = \frac{\epsi}{2}$}.
\end{align*}
By Inequality~\ref{inq:tgtwo-dif-opt},
the approximation ratio of Alg.~\ref{alg:tgtwo} is $e^{-1}-\epsi$.
\end{proof}

In the rest of this section, we provide the proofs for 
Lemma~\ref{lemma:tg-par-A} and~\ref{lemma:tg-recur}.
\begin{proof}[Proof of Lemma~\ref{lemma:tg-par-A}]
At iteration $i$ of the outer for loop,
let $A_l$ be the set at the end of iteration $i$,
$a_{l, j}$ be the $j$-th element added to $A_l$,
$\tau_l^j$ be the threshold value of $\tau_l$ when $a_{l, j}$ is added to $A_l$,
and $A_{l, j}$ be $A_l$ after $a_{l, j}$ is added to $A_l$.
Let $c_l^* = \max\{c\in [m]:A_{l, c}\setminus G_{i-1}\subseteq O_l\}$.

First, we prove that the first inequality holds.
For each $l\in [\ell]$, order the elements in $O_l$ as $\{o_1, o_2, \ldots\}$
such that $o_j \not \in A_{l, j-1}$ for any $1\le j \le |A_l\setminus G_{i-1}|$,
and $o_j\not \in A_l$ for any $|A_l\setminus G_{i-1}| < j \le m$.

When $1\le j \le |A_l\setminus G_{i-1}|$, by Claim~\ref{claim:par-A},
each $o_j$ is either added to $A_l$ or not in any solution set.
Since $\tau_l$ is initialized with the maximum marginal gain $M$,
$o_j$ is not considered to be added to $A_l$ with threshold value 
$\tau_l^j/(1-\epsi')$.
Therefore, by submodularity it holds that
\begin{equation}\label{inq:tgtwo-1}
\marge{o_j}{A_{l, j-1}} < \tau_l^j/(1-\epsi')\le \marge{a_{l,j}}{A_{l,j-1}}/(1-\epsi'),
\forall 1\le j \le |A_l\setminus G_{i-1}|.
\end{equation}

When $|A_l\setminus G_{i-1}| <  m$,
the minimum value of $\tau_l$ is less than $\frac{\epsi' M}{k}$.
Then, for any $|A_l\setminus G_{i-1}| < j \le m$,
$o_j$ is not considered to be added to $A_l$ with threshold value less than $\frac{\epsi' M}{(1-\epsi')k}$.
It follows that 
\begin{equation}\label{inq:tgtwo-2}
\marge{o_j}{A_l} \le \frac{\epsi' M}{(1-\epsi')k},
\forall |A_l\setminus G_{i-1}| < j \le m.
\end{equation}

Then,
\begin{align*}
\marge{O_{l}}{A_{l}} &\le \sum_{o_j\in O_l} \marge{o_j}{A_{l}} \tag{Proposition~\ref{prop:sum-marge}}\\
&\le \sum_{j=1}^{|A_l\setminus G_{i-1}|} \marge{o_j}{A_{l, j-1}} + 
\sum_{j=|A_l\setminus G_{i-1}|+1}^m \marge{o_j}{A_{l}}\tag{Submodularity}\\
&\le \sum_{j=1}^{|A_l\setminus G_{i-1}|}\frac{\marge{a_{l,j}}{A_{l,j-1}}}{1-\epsi'}+\frac{\epsi' M}{(1-\epsi')\ell}
\tag{Inequalities~\eqref{inq:tgtwo-1} and~\eqref{inq:tgtwo-2}}\\
&= \frac{\marge{A_{l}}{G_{j-1}}}{1-\epsi'}+\frac{\epsi' M}{(1-\epsi')\ell}.
\end{align*}
The first inequality holds.

In the following, we prove that the second inequality holds.
For any $1\le l_1\le l_2\le \ell$,
we analyze two cases of the relationship between $c_{l_1}^* $ and $ c_{l_2}^*$ in the following.

% \textbf{Case 1: $c_{l_1}^* = c_{l_2}^* = m$.}
% Then, $O_{l_1} = A_{l_1}\setminus G_{i-1}$ and $O_{l_2} = A_{l_2}\setminus G_{i-1}$.
% By submodularity,
% \[\marge{O_{l_2}}{A_{l_1}} + \marge{O_{l_1}}{A_{l_2}} \le \marge{O_{l_2}}{G_{i-1}} + \marge{O_{l_1}}{G_{i-1}} = \marge{A_{l_1}}{G_{i-1}} + \marge{A_{l_2}}{G_{i-1}}.\]
% Therefore, the lemma holds in this case.

\textbf{Case 1: $c_{l_1}^* \le c_{l_2}^*$; left half part in Fig.~\ref{fig:gdtwo}.}

First, we bound $\marge{O_{l_1}}{A_{l_2}}$.
Order the elements in $O_{l_1}\setminus A_{l_1, c_{l_1}^*}$ as $\{o_1, o_2, \ldots\}$ such that $o_j \not \in A_{l_1, c_{l_1}^*+j}$.
(Refer to the gray block with a dotted edge in the top left corner of Fig.~\ref{fig:gdtwo} for $O_{l_1}$.
If $c_{l_1}^*+j$ is greater than the number of elements added to $A_{l_1}$,
$A_{l_1, c_{l_1}^*+j}$ refers to $A_{l_1}$.)
Note that, since $A_{l_1, c_{l_1}^*} \subseteq O_{l_1}$,
it follows that $|O_{l_1}\setminus A_{l_1, c_{l_1}^*}| \le m - c_{l_1}^*$.

When $1 \le j \le |A_{l_2}\setminus G_{i-1}| - c_{l_1}^*$,
since each $o_j$ is either added to $A_{l_1}$ or not in any solution set by Claim~\ref{claim:par-A}
and $\tau_{l_2}$ is initialized with the maximum marginal gain $M$,
$o_j$ is not considered to be added to $A_{l_2}$ with threshold value $\tau_{l_2}^{c_{l_1}^* + j}/(1-\epsi')$.
Therefore, it holds that 
\begin{equation}\label{inq:tgtwo-case2-1}
\marge{o_j}{A_{l_2, c_{l_1}^*+j-1}} < \frac{\tau_{l_2}^{c_{l_1}^* + j}}{1-\epsi'} \le \frac{\marge{a_{l_2, c_{l_1}^* + j}}{A_{l_2, c_{l_1}^* + j-1}}}{1-\epsi'}, \forall 1\le j\le |A_{l_2}\setminus G_{i-1}|-c_{l_1}^*.
\end{equation}

When $|A_{l_2}\setminus G_{i-1}| < m$ and $|A_{l_2}\setminus G_{i-1}|-c_{l_1}^* < j\le m-c_{l_1}^*$,
this iteration ends with $\tau_{l_2} < \frac{\epsi' M}{k}$ and
$o_j$ is never considered to be added to $A_{l_2}$.
Thus, it holds that
\begin{equation}\label{inq:tgtwo-case2-3}
\marge{o_j}{A_{l_2}} < \frac{\epsi' M}{(1-\epsi')k}, 
\forall |A_{l_2}\setminus G_{i-1}|-c_{l_1}^* < j \le m-c_{l_1}^*.
\end{equation}

Then,
\begin{align*}
\marge{O_{l_1}}{A_{l_2}} &\le \marge{A_{l_1, c_{l_1}^*}}{A_{l_2}}  + \sum_{o_j \in O_{l_1}\setminus A_{l_1, c_{l_1}^*}}\marge{o_j}{A_{l_2}} \tag{Proposition~\ref{prop:sum-marge}}\\
&\le \marge{A_{l_1, c_{l_1}^*}}{G_{i-1}} + \sum_{j = 1}^{|A_{l_2}\setminus G_{i-1}|-c_{l_1}^*}\marge{o_j}{A_{l_2, , c_{l_1}^*+j-1}} + \sum_{j=|A_{l_2}\setminus G_{i-1}|-c_{l_1}^*+1}^{m-c_{l_1}^*} \marge{o_j}{A_{l_2}} \tag{submodularity}\\
&\le \ff{A_{l_1, c_{l_1}^*}}-\ff{G_{i-1}} + \sum_{j = 1}^{|A_{l_2}\setminus G_{i-1}|-c_{l_1}^*}\frac{\marge{a_{l_2, c_{l_1}^*+j}}{A_{l_2, , c_{l_1}^*+j-1}}}{1-\epsi'} + \frac{\epsi' M}{(1-\epsi')\ell} \tag{Inequality~\eqref{inq:tgtwo-case2-1} and~\eqref{inq:tgtwo-case2-3}}\\
& \le \ff{A_{l_1, c_{l_1}^*}}-\ff{G_{i-1}} + \frac{\ff{A_{l_2}} - \ff{A_{l_2, c_{l_1}^*}}}{1-\epsi'} + \frac{\epsi' M}{(1-\epsi')\ell} \numberthis \label{inq:tgtwo-case2-6}
\end{align*}

Similarly, we bound $\marge{O_{l_2}}{A_{l_1}}$ below.
Order the elements in $O_{l_2}\setminus A_{l_2, c_{l_1}^*}$ as $\{o_1, o_2, \ldots\}$ such that $o_j \not \in A_{l_2, c_{l_1}^*+j-1}$.
(See the gray block with a dotted edge in the bottom left corner of Fig.~\ref{fig:gdtwo} for $O_{l_2}$.
If $c_{l_1}^*+j-1$ is greater than the number of elements added to $A_{l_2}$,
$A_{l_2, c_{l_1}^*+j-1}$ refers to $A_{l_2}$.)
Note that, since $A_{l_2, c_{l_1}^*} \subseteq O_{l_2}$,
it follows that $|O_{l_2}\setminus A_{l_2, c_{l_1}^*}| \le m - c_{l_1}^*$.

When $1 \le j \le |A_{l_1}\setminus G_{i-1}|-c_{l_1}^*$,
since each $o_j$ is either added to $A_{l_2}$ or not in any solution set by Claim~\ref{claim:par-A}
and $\tau_{l_1}$ is initialized with the maximum marginal gain $M$,
$o_j$ is not considered to be added to $A_{l_1}$ with threshold value $\tau_{l_1}^{c_{l_1}^* + j}/(1-\epsi')$.
Therefore, it holds that 
\begin{equation}\label{inq:tgtwo-case2-4}
\marge{o_j}{A_{l_1, c_{l_1}^*+j-1}} < \frac{\tau_{l_1}^{c_{l_1}^* + j}}{1-\epsi'} \le \frac{\marge{a_{l_1, c_{l_1}^* + j}}{A_{l_1, c_{l_1}^* + j-1}}}{1-\epsi'}, \forall 1\le j\le |A_{l_2}\setminus G_{i-1}|-c_{l_1}^*.
\end{equation}

When $|A_{l_1}\setminus G_{i-1}| < m$ and $|A_{l_1}\setminus G_{i-1}|-c_{l_1}^* < j\le m-c_{l_1}^*$,
this iteration ends with $\tau_{l_1} < \frac{\epsi' M}{k}$
and $o_j$ is never considered to be added to $A_{l_1}$.
Thus, it holds that
\begin{equation}\label{inq:tgtwo-case2-5}
\marge{o_j}{A_{l_1}} < \frac{\epsi' M}{(1-\epsi')k}, \forall |A_{l_1}\setminus G_{i-1}|-c_{l_1}^* < j \le m-c_{l_1}^*.
\end{equation}

Then,
\begin{align*}
\marge{O_{l_2}}{A_{l_1}} &\le \marge{A_{l_2, c_{l_1}^*}}{A_{l_1}}  + \sum_{o_j \in O_{l_2}\setminus A_{l_2, c_{l_1}^*}}\marge{o_j}{A_{l_1}} \tag{Proposition~\ref{prop:sum-marge}}\\
&\le \marge{A_{l_2, c_{l_1}^*}}{G_{i-1}} + \sum_{j = 1}^{|A_{l_1}\setminus G_{i-1}|-c_{l_1}^*}\marge{o_j}{A_{l_1, c_{l_1}^*+j-1}} + \sum_{j=|A_{l_1}\setminus G_{i-1}|-c_{l_1}^*+1}^{m-c_{l_1}^*} \marge{o_j}{A_{l_1}} \tag{submodularity}\\
&\le \ff{A_{l_2, c_{l_1}^*}}-\ff{G_{i-1}} + \sum_{j = 1}^{|A_{l_1}\setminus G_{i-1}|-c_{l_1}^*}\frac{\marge{a_{l_1, c_{l_1}^*+j}}{A_{l_1, , c_{l_1}^*+j-1}}}{1-\epsi'} + \frac{\epsi' M}{(1-\epsi')\ell} \tag{Inequality~\eqref{inq:tgtwo-case2-4} and~\eqref{inq:tgtwo-case2-5}}\\
& \le \ff{A_{l_2, c_{l_1}^*}}-\ff{G_{i-1}} + \frac{\ff{A_{l_1}} - \ff{A_{l_1, c_{l_1}^*}}}{1-\epsi'} + \frac{\epsi' M}{(1-\epsi')\ell} \numberthis \label{inq:tgtwo-case2-7}
\end{align*}

By Inequalities~\eqref{inq:tgtwo-case2-6} and~\eqref{inq:tgtwo-case2-7},
\begin{align*}
\marge{O_{l_1}}{A_{l_2}}+\marge{O_{l_2}}{A_{l_1}}
\le \frac{1}{1-\epsi'}\left(\marge{A_{l_1}}{G_{i-1}}+\marge{A_{l_2}}{G_{i-1}}\right)
+\frac{2\epsi' M}{(1-\epsi')\ell}
\end{align*}

Thus, the lemma holds in this case.

\textbf{Case 2: $c_{l_1}^* > c_{l_2}^*$; right half part in Fig.~\ref{fig:gdtwo}.}

First, we bound $\marge{O_{l_1}}{A_{l_2}}$.
Order the elements in $O_{l_1}\setminus A_{l_1, c_{l_2}^*+1}$ as $\{o_1, o_2, \ldots\}$ such that $o_j \not \in A_{l_1, c_{l_1}^*+j}$.
(Refer to the gray block with a dotted edge in the top right corner of Fig.~\ref{fig:gdtwo} for $O_{l_1}$.
If $c_{l_1}^*+j$ is greater than the number of elements added to $A_{l_1}$,
$A_{l_1, c_{l_1}^*+j}$ refers to $A_{l_1}$.)
Note that, since $A_{l_1, c_{l_2}^*+1} \subseteq O_{l_1}$,
it follows that $|O_{l_1}\setminus A_{l_1, c_{l_2}^*+1}| \le m - c_{l_2}^*-1$.

When $1 \le j \le |A_{l_2}\setminus G_{i-1}| - c_{l_2}^* - 1$,
since each $o_j$ is either added to $A_{l_1}$ or not in any solution set by Claim~\ref{claim:par-A}
and $\tau_{l_2}$ is initialized with the maximum marginal gain $M$,
$o_j$ is not considered to be added to $A_{l_2}$ with threshold value $\tau_{l_2}^{c_{l_2}^* + j}/(1-\epsi')$.
Therefore, it holds that 
\begin{equation}\label{inq:tgtwo-case3-1}
\marge{o_j}{A_{l_2, c_{l_2}^*+j-1}} < \frac{\tau_{l_2}^{c_{l_2}^* + j}}{1-\epsi'} \le \frac{\marge{a_{l_2, c_{l_2}^* + j}}{A_{l_2, c_{l_2}^* + j-1}}}{1-\epsi'}, \forall 1\le j\le |A_{l_2}\setminus G_{i-1}| - c_{l_2}^* - 1.
\end{equation}

When $|A_{l_2}\setminus G_{i-1}| < m$ and $|A_{l_2}\setminus G_{i-1}|- c_{l_2}^* - 1 < j\le m- c_{l_2}^* - 1$,
this iteration ends with $\tau_{l_2} < \frac{\epsi' M}{k}$ and
$o_j$ is never considered to be added to $A_{l_2}$.
Thus, it holds that
\begin{equation}\label{inq:tgtwo-case3-3}
\marge{o_j}{A_{l_2}} < \frac{\epsi' M}{(1-\epsi')k}, 
\forall |A_{l_2}\setminus G_{i-1}|- c_{l_2}^* - 1 < j \le m- c_{l_2}^* - 1.
\end{equation}

Then,
\begin{align*}
\marge{O_{l_1}}{A_{l_2}} &\le \marge{A_{l_1, c_{l_2}^*}}{A_{l_2}}  + \sum_{o_j \in O_{l_1}\setminus A_{l_1, c_{l_2}^*+1}}\marge{o_j}{A_{l_2}} \tag{Proposition~\ref{prop:sum-marge}}\\
&\le \marge{A_{l_1, c_{l_2}^*}}{G_{i-1}} + \sum_{j = 1}^{|A_{l_2}\setminus G_{i-1}|- c_{l_2}^* - 1}\marge{o_j}{A_{l_2, , c_{l_2}^*+j-1}} + \sum_{j=|A_{l_2}\setminus G_{i-1}|- c_{l_2}^*}^{m- c_{l_2}^* - 1} \marge{o_j}{A_{l_2}} \tag{submodularity}\\
&\le \ff{A_{l_1, c_{l_2}^*}}-\ff{G_{i-1}} + \sum_{j = 1}^{|A_{l_2}\setminus G_{i-1}|- c_{l_2}^* - 1}\frac{\marge{a_{l_2, c_{l_2}^*+j}}{A_{l_2, , c_{l_2}^*+j-1}}}{1-\epsi'} + \frac{\epsi' M}{(1-\epsi')\ell} \tag{Inequality~\eqref{inq:tgtwo-case3-1} and~\eqref{inq:tgtwo-case3-3}}\\
& \le \ff{A_{l_1, c_{l_2}^*}}-\ff{G_{i-1}} + \frac{\ff{A_{l_2}} - \ff{A_{l_2, c_{l_2}^*}}}{1-\epsi'} + \frac{\epsi' M}{(1-\epsi')\ell} \numberthis \label{inq:tgtwo-case3-6}
\end{align*}

Similarly, we bound $\marge{O_{l_2}}{A_{l_1}}$ below.
Order the elements in $O_{l_2}\setminus A_{l_2, c_{l_2}^*}$ as $\{o_1, o_2, \ldots\}$ such that $o_j \not \in A_{l_2, c_{l_2}^*+j}$.
(See the gray block with a dotted edge in the bottom right corner of Fig.~\ref{fig:gdtwo} for $O_{l_2}$.
If $c_{l_2}^*+j$ is greater than the number of elements added to $A_{l_2}$,
$A_{l_2, c_{l_2}^*+j}$ refers to $A_{l_2}$.)
Note that, since $A_{l_2, c_{l_2}^*} \subseteq O_{l_2}$,
it follows that $|O_{l_2}\setminus A_{l_2, c_{l_2}^*}| \le m - c_{l_2}^*$.

When $1 \le j \le |A_{l_1}\setminus G_{i-1}|- c_{l_2}^* - 1$, 
since each $o_j$ is either added to $A_{l_2}$ or not in any solution set by Claim~\ref{claim:par-A}
and $\tau_{l_1}$ is initialized with the maximum marginal gain $M$,
$o_j$ is not considered to be added to $A_{l_1}$ with threshold value $\tau_{l_1}^{c_{l_2}^* + j+1}/(1-\epsi')$.
Therefore, it holds that 
\begin{equation}\label{inq:tgtwo-case3-4}
\marge{o_j}{A_{l_1, c_{l_2}^*+j}} < \frac{\tau_{l_1}^{c_{l_2}^* + j+1}}{1-\epsi'} \le \frac{\marge{a_{l_1, c_{l_2}^* + j+1}}{A_{l_1, c_{l_2}^* + j}}}{1-\epsi'}, \forall 1\le j\le |A_{l_2}\setminus G_{i-1}|- c_{l_2}^* - 1.
\end{equation}
If $|A_{l_1}\setminus G_{i-1}| = m$,
consider the last element $o_{m-c_{l_2}^*}$ in $O_{l_2}\setminus A_{l_2, c_{l_2}^*}$.
Since $o_{m-c_{l_2}^*} \not\in A_{l_2}$ and $o_{m-c_{l_2}^*} \not\in A_{l_1}$, $o_{m-c_{l_2}^*}$ is not considered to be added to 
$A_{l_1}$ with threshold value $\tau_{l_1}^j/(1-\epsi')$ for any $j \in [m]$.
Then,
\begin{equation}\label{inq:tgtwo-case3-2}
\marge{o_{m-c_{l_2}^*}}{A_{l_1}} < \frac{\sum_{j=1}^m \tau_{l_1}^j}{(1-\epsi')m}
\le \frac{\sum_{j=1}^m \marge{a_{l_1, j}}{A_{l_1, j-1}}}{(1-\epsi')m}
 = \frac{\marge{A_{l_1}}{G_{i-1}}}{(1-\epsi')m}.
\end{equation}

When $|A_{l_1}\setminus G_{i-1}| < m$ and $|A_{l_1}\setminus G_{i-1}|- c_{l_2}^* - 1 < j\le m- c_{l_2}^*$,
this iteration ends with $\tau_{l_1} < \frac{\epsi' M}{k}$
and $o_j$ is never considered to be added to $A_{l_1}$.
Thus, it holds that
\begin{equation}\label{inq:tgtwo-case3-5}
\marge{o_j}{A_{l_1}} < \frac{\epsi' M}{(1-\epsi')k}, 
\forall |A_{l_1}\setminus G_{i-1}|- c_{l_2}^* - 1 < j \le m- c_{l_2}^*.
\end{equation}

Then,
\begin{align*}
\marge{O_{l_2}}{A_{l_1}} &\le \marge{A_{l_2, c_{l_2}^*}}{A_{l_1}}  + \sum_{o_j \in O_{l_2}\setminus A_{l_2, c_{l_2}^*}}\marge{o_j}{A_{l_1}} \tag{Proposition~\ref{prop:sum-marge}}\\
&\le \marge{A_{l_2, c_{l_2}^*}}{G_{i-1}} + \sum_{j = 1}^{|A_{l_1}\setminus G_{i-1}|- c_{l_2}^* - 1}\marge{o_j}{A_{l_1, c_{l_2}^*+j-1}} + \sum_{j=|A_{l_1}\setminus G_{i-1}|- c_{l_2}^*}^{m} \marge{o_j}{A_{l_1}} \tag{submodularity}\\
&\le \ff{A_{l_2, c_{l_2}^*}}-\ff{G_{i-1}} + \sum_{j = 1}^{|A_{l_1}\setminus G_{i-1}|- c_{l_2}^* - 1}\frac{\marge{a_{l_1, c_{l_2}^*+j}}{A_{l_1, , c_{l_2}^*+j-1}}}{1-\epsi'}
+ \frac{\marge{A_{l_1}}{G_{i-1}}}{(1-\epsi')m}
+\frac{\epsi' M}{(1-\epsi')\ell} \tag{Inequalities~\eqref{inq:tgtwo-case3-4}-\eqref{inq:tgtwo-case3-5}}\\
& \le \ff{A_{l_2, c_{l_2}^*}}-\ff{G_{i-1}} + \frac{\ff{A_{l_1}} - \ff{A_{l_1, c_{l_2}^*}}}{1-\epsi'} + \frac{\marge{A_{l_1}}{G_{i-1}}}{(1-\epsi')m} + \frac{\epsi' M}{(1-\epsi')\ell} \numberthis \label{inq:tgtwo-case3-7}
\end{align*}

By Inequalities~\eqref{inq:tgtwo-case3-6} and~\eqref{inq:tgtwo-case3-7},
\begin{align*}
\marge{O_{l_1}}{A_{l_2}}+\marge{O_{l_2}}{A_{l_1}}
\le \frac{1}{1-\epsi'}\left(1+\frac{1}{m}\right)\left(\marge{A_{l_1}}{G_{i-1}}+\marge{A_{l_2}}{G_{i-1}}\right)
+\frac{2\epsi' M}{(1-\epsi')\ell}
\end{align*}

Thus, the lemma holds in this case.
\end{proof}

\begin{proof}[Proof of Lemma~\ref{lemma:tg-recur}]
Fix on $G_{i-1}$ at the beginning of this iteration,
Since $\{A_l\setminus G_{i-1}: l\in [\ell]\}$ are pairwise disjoint sets,
by Proposition~\ref{prop:sum-marge}, it holds that
\[\exc{\ff{O\cup G_i}}{G_{i-1}} = \frac{1}{\ell}\sum_{l\in [\ell]}\ff{O\cup A_l} \ge \left(1-\frac{1}{\ell}\right)\ff{O\cup G_{i-1}}.\]
Then, by unfixing $G_{i-1}$, the first inequality holds.

To prove the second inequality, also consider fix on $G_{i-1}$ at the beginning of iteration $i$.
Then,
\begin{align*}
\sum_{l\in [\ell]}\marge{O}{A_l} &\le \sum_{l_1\in [\ell]}\sum_{l_2\in [\ell]}\marge{O_{l_1}}{A_{l_2}}\tag{Proposition~\ref{prop:sum-marge}}\\
& = \sum_{l \in [\ell]}\marge{O_{l}}{A_{l}} + \sum_{1\le l_1< l_2 \le \ell} \left(\marge{O_{l_1}}{A_{l_2}} +\marge{O_{l_2}}{A_{l_1}}\right) \tag{Lemma~\ref{lemma:tg-par-A}}\\
& \le \sum_{l \in [\ell]}\left(\frac{\marge{A_{l}}{G_{i-1}}}{1-\epsi'}+\frac{\epsi' M}{(1-\epsi')\ell}\right)\\
&\hspace*{2em}+\sum_{1\le l_1< l_2 \le \ell} \left(\frac{1}{1-\epsi'}\left(1+\frac{1}{m}\right)\left(\marge{A_{l_1}}{G_{i-1}}
+\marge{A_{l_2}}{G_{i-1}}\right)
+\frac{2\epsi' M}{(1-\epsi')\ell}\right)\tag{Lemma~\ref{lemma:tg-par-A}}\\
&\le \frac{\ell}{1-\epsi'}\left(1+\frac{1}{m}\right)\sum_{l \in [\ell]}\marge{A_{l}}{G_{i-1}} + \frac{\epsi' \ell}{1-\epsi'}\ff{O}\tag{$M \le \ff{O}$}
\end{align*}
\begin{align*}
\Rightarrow \left(1+\frac{\ell}{1-\epsi'}\right)\left(1+\frac{1}{m}\right) \sum_{l\in [\ell]}\marge{A_l}{G_{i-1}} &\ge \sum_{l\in [\ell]}\ff{O\cup A_l} -\ell\ff{G_{i-1}} - \frac{\epsi' \ell}{1-\epsi'}\ff{O}\\
&\ge \left(\ell-1\right)\ff{O\cup G_{i-1}}-\ell\ff{G_{i-1}} - \frac{\epsi' \ell}{1-\epsi'}\ff{O}
\end{align*}
Thus,
\begin{align*}
&\exc{\ff{G_i} - \ff{G_{i-1}}}{G_{i-1}}  = \frac{1}{\ell}\sum_{l \in [\ell]}\marge{A_{l}}{G_{i-1}}\\
&\ge \frac{1}{1+\frac{\ell}{1-\epsi'}} \frac{m}{m+1}\left(\left(1-\frac{1}{\ell}\right)  \ff{O\cup G_{i-1}} - \ff{G_{i-1}} - \frac{\epsi'}{1-\epsi'}\ff{O}\right)\tag{Proposition~\ref{prop:sum-marge}}
\end{align*}
By unfixing $G_{i-1}$, the second inequality holds.
\end{proof}

\section{Analysis of \rev Parallel Algorithms Introduced in Section~\ref{sec:ptg}} % (fold)
\label{apx:ptg}
\rev
This section presents the formal analysis of our parallel algorithms introduced in Section~\ref{sec:ptg}.
We first prove fundamental lemmata for \ptgoneshort in Appendix~\ref{apx:ptgone},
then establish its approximation guarantees in Appendix~\ref{apx:ptgone-guarantee},
and finally analyze \ptgtwoshort in Appendix~\ref{apx:ptgtwo}.
\color{black}

\subsection{\rev Key Lemmata for \ptgoneshort (Alg.~\ref{alg:ptgone}, Section~\ref{sec:ptg})}\label{apx:ptgone}
\rev
We provide the key lemmata achieved by \ptgoneshort as follows,
\color{black}
\begin{lemma}\label{lemma:ptgone}
With input $(f, m, \ell, \tau_{\min}, \epsi)$, \ptgone (Alg.~\ref{alg:ptgone})
runs in $\oh{\ell^2\epsi^{-2}\log(n)\log\left(\frac{M}{\tau_{\min}}\right)}$ adaptive rounds and $\oh{\ell^3 \epsi^{-2}n\log(n)\log\left(\frac{M}{\tau_{\min}}\right)}$ queries with a probability of $1-1/n$,
and terminates with $\{(A_l, A_l'): l\in [\ell]\}$ \st
\begin{enumerate}
\item $A_l'\subseteq A_l$, $\marge{A_l'}{\emptyset} \ge \marge{A_l}{\emptyset}, \forall 1\le l \le \ell$, and $\{A_l: l\in [\ell]\}$ are pairwise disjoint sets,
\item $\marge{O_{l}}{A_{l}}\le \frac{\marge{A_{l}'}{\emptyset}}{(1-\epsi)^2}+\frac{m\cdot\tau_{\min}}{1-\epsi}, \forall 1\le l \le \ell$,
\item $\marge{O_{l_2}}{A_{l_1}} + \marge{O_{l_1}}{A_{l_2}} \le 
\frac{1+\frac{1}{m}}{(1-\epsi)^2}\left(\marge{A_{l_1}'}{\emptyset}+\marge{A_{l_2}'}{\emptyset}\right) + \frac{2m\cdot \tau_{\min}}{1-\epsi}, \text{ if } O_{l_1} = O_{l_2}, \forall 1\le l_1 < l_2 \le \ell$,
\end{enumerate}
where $O_l\subseteq \uni$, $|O_l| \le m$, and $O_l \cap A_j = \emptyset$ for each $j \neq l$.

Especially, when $\ell = 2$,
\begin{itemize}
	\item[4.] $\marge{S}{A_{1}} + \marge{S}{A_{2}} \le \frac{1}{(1-\epsi)^2}\left(\marge{A_{l_1}'}{\emptyset}+\marge{A_{l_2}'}{\emptyset}\right) + \frac{2m\cdot \tau_{\min}}{1-\epsi}, \forall S\subseteq \uni, |S| \le m$. 
\end{itemize}
\end{lemma}
Before proving Lemma~\ref{lemma:ptgone}, we provide the following lemma regarding each iteration of \ptgone.
\begin{lemma}\label{lemma:tgone-iteration}
For any iteration of the while loop in \ptgone (Alg.~\ref{alg:ptgone}),
let $A_{l, 0}$, $A_{l, 0}'$, $V_{l, 0}$, $\tau_{l, 0}$ be the set and threshold value at the beginning,
and $A_l$, $A_l'$, $V_l$, $\tau_l$ be those at the end.
The following properties hold. 
\begin{enumerate}
\item With a probability of at least $1/2$,
there exists $l \in [\ell]$ \st $\tau_l < \tau_{l, 0}$
or $m_0 = 0$ or
$|V_l| \le \left(1-\frac{\epsi}{4\ell}\right)|V_{l, 0}|$.
\item $\{A_l: l\in [\ell]\}$ have the same size and are pairwise disjoint.  
% \item $V_l = \left\{x\in V\setminus\left(\bigcup_{i\in [\ell]} A_l\right) : \marge{x}{A_l} \ge \tau_l \right\}$  for all $l\in [\ell]$.
% \item For each $l \in [\ell]$ \st $\tau_l < \tau_{l,0}$,
% it holds that $\marge{x}{A_l} < \frac{\tau_l}{1-\epsi}$ for all $x\in V\setminus \left(\bigcup_{i\in [\ell]} A_l\right)$.
\item For each $x\in A_l\setminus A_{l,0}$, let $\tau_l^{(x)}$ be the threshold value when $x$ is added to the solution,
$A_{l, (x)}$ be the largest prefix of $A_l$ that do not include $x$,
and for any $j\in [\ell]$ and $j\neq l$,
$A_{j, (x)}$ be the prefix of $A_j$ with $|A_{l, (x)}|$ elements if $j < l$,
or with $|A_{l, (x)}|-1$ elements if $j > l$.
Then, for any $l\in [\ell]$, $x\in A_l\setminus A_{l,0}$,
and $y\in \uni\setminus \left(\bigcup_{j\in [\ell]} A_{j, (x)}\right)$,
it holds that $\marge{y}{A_{l, (x)}} < \frac{\tau_l^{(x)}}{1-\epsi}$.
\item $A_l'\subseteq A_l$, $\marge{A_l'}{A_{l, 0}'} \ge \marge{A_l}{A_{l, 0}}$,
and $\marge{A_l'}{A_{l, 0}'}\ge (1-\epsi)\sum_{x \in A_l\setminus A_{l, 0}}\tau_l^{(x)}$ for all $l\in [\ell]$.
\end{enumerate}
\end{lemma}
\begin{proof}[Proof of Lemma~\ref{lemma:tgone-iteration}]
\textbf{Proof of Property 1.}
At the beginning of the iteration, if there exists $l\in I$ \st $|V_{l, 0}| < 2\ell$,
then either $\tau_{l, 0}$ is decreased to $\tau_l$ and $V_l$ is updated accordingly, 
or an element $x_l$ from $V_{l, 0}$ is added to $A_j$ and $A_j'$
and subsequently removed from $V_{l, 0}$. 
This implies that
\[|V_l| \le |V_{l,0}| -1 < \left(1-\frac{1}{2\ell}\right)|V_{l,0}|.\]
Property 1 holds in this case.

Otherwise, for all $l\in I$, it holds that $|V_{l, 0}| \ge 2\ell$,
and the algorithm proceeds to execute Lines~\ref{line:tgone-dist}-\ref{line:tgone-update-size}.
By Lemma~\ref{lemma:dist}, in Line~\ref{line:tgone-dist}, 
$|\mathcal V_l| \ge \frac{|V_{l, 0}|}{2\ell}$ for each $l\in I$ .
Consider the index $j\in I$ where $i_j^* = i^*$.
Then, $O_j$ consists of the first $i^*$ elements in $\mathcal V_j$
by Line~\ref{line:tgone-subset}.
By Lemma~\ref{lemma:prefix-prob},
with probability greater than $1/2$,
either $i^* = m_0$ or at least an $\frac{\epsi}{2}$-fraction
of elements $x\in \mathcal V_j$ satisfy $\marge{x}{A_j}< \tau_{j,0}$.
Consequently, either $m_0 = 0$ after Line~\ref{line:tgone-update-size},
or, after the \update procedure in Line~\ref{line:tgone-update},
one of the following holds:
$|V_l| \le \left(1-\frac{\epsi}{4\ell}\right)|V_{l, 0}|$,
or $\tau_{j} < \tau_{j, 0}$.
Therefore, Property 1 holds in this case. 

\textbf{Proof of Property 2.}
At any iteration of the while,
either $|I|$ different elements or $|I|$ pairwise disjoint sets with same size $i^*$
are added to solution sets $\{A_l: l\in I\}$.
Therefore, Property 2 holds.

\textbf{Proof of Property 3.}
At any iteration, if $\tau_l$ is not updated on Line~\ref{line:tgone-update-2},
then prior to this iteration, all the elements outside of the solutions
have marginal gain less than $\frac{\tau_l^{(x)}}{1-\epsi}$.
Thus, for any $x \in A_l\setminus A_{l, 0}$, $y\in \uni\setminus \left(\bigcup_{j\in [\ell]} A_{j, 0}\right)$,
it holds that $\marge{y}{A_{l, (x)}} < \frac{\tau_l^{(x)}}{1-\epsi}$
by submodularity. Property 3 holds in this case.

Otherwise, if $\tau_l$ is updated on Line~\ref{line:tgone-update-2},
only one element is added to each solution set during this iteration.
Let $x = A_l\setminus A_{l, 0}$.
For any $j\in [\ell]$ and $j\neq l$,
it holds that $A_{j, (x)} = A_j$ if $j < l$,
or $A_{j, (x)} = A_{j, 0}$ if $j > l$.
Since elements are added to each pair of solutions in sequence within the for loop
in Lines~\ref{line:tgone-for-begin}-\ref{line:tgone-for-end},
by the \update procedure,
for any $y \in \uni \setminus \left(\bigcup_{j\in [\ell]} A_{j, (x)}\right)$,
it holds that $\marge{y}{A_{l, (x)}} < \frac{\tau_l^{(x)}}{1-\epsi}$.
Therefore, Property 3 also holds in this case.

\textbf{Proof of Property 4.}
First, we prove $A_l'\subseteq A_l$ by induction.
At the beginning of the algorithm,
$A_l'$ and $A_l$ are initialized as empty sets.
Clearly, the property holds in the base case.
Then, suppose that $A_{l,0}'\subseteq A_{l,0}$.
There are three possible cases of updating $A_{l,0}'$ and $A_{l,0}$ at any iteration:
1) $A_l'=A_{l,0}'$ and $A_l = A_{l,0}$,
2) $A_l' = A_{l,0}' + x_l$ and $A_l = A_{l,0} + x_l$ in Line~\ref{line:tgone-update-A},
or 3) $A_l' = A_{l,0}' \cup S_l'$ and $A_l = A_{l,0} \cup S_l$ in Line~\ref{line:tgone-update-A-2}.
Clearly, $A_l'\subseteq A_l$ holds in all cases.

Next, we prove the rest of Property 4.

If $A_l'=A_{l,0}'$ and $A_l = A_{l,0}$, 
then $\marge{A_l'}{A_{l, 0}'} = \marge{A_l}{A_{l, 0}}=0$.
Property 4 holds.

If $A_{l,0}'$ and $A_{l,0}$ are updated in Line~\ref{line:tgone-update-A},
by submodularity, $\marge{A_l'}{A_{l, 0}'} = \marge{x_l}{A_{l, 0}'} \ge \marge{x_l}{A_{l, 0}}=\marge{A_l}{A_{l, 0}} \ge \tau_l^{(x_{l})}$.
Therefore, Property 4 also holds.

If $A_{l,0}'$ and $A_{l,0}$ are updated in Line~\ref{line:tgone-update-A-2},
we know that $A_l' = A_{l,0}' \cup S_l'$ and $A_l = A_{l,0} \cup S_l$.
Suppose the elements in $S_l$ and $S_l'$ retain their original order within $\mathcal V_l$.
For each $x\in S_l$, let $S_{l,(x)}$, $\mathcal V_{l, (x)}$ and $A_{l, (x)}$
be the largest prefixes of $S_l$, $\mathcal V_l$ and $A_l$ that do not include $x$, respectively.
Moreover, let $S_{l, (x)}' = S_{l, (x)}\cap S_l'$ and $A_{l, (x)}' = A_{l, (x)}\cap A_l'$.
Say an element $x\in S_l$ \textbf{true} if $B_l[(x)] = \textbf{true}$,
where $B_l[(x)]$ is the $i$-th element in $B_l$ if $x$ is the $i$-th element in $\mathcal V_l$.
Similarly, say an element $x\in S_l$ \textbf{false} if $B_l[(x)] = \textbf{false}$,
and \textbf{none} otherwise.

Following the above definitions, for any \textbf{true} or \textbf{none} element $x\in S_l$,
by Line~\ref{line:tgone-subset}, it holds that $S_{l, (x)} \subseteq \mathcal V_{l, (x)}$.
Then, by Line~\ref{line:prefix-B-true} and submodularity,
\begin{equation*}
\marge{x}{A_{l, (x)}} = \marge{x}{A_{l, 0} \cup S_{l, (x)}}
\ge \marge{x}{A_{l, 0} \cup \mathcal V_{l, (x)}} \ge \left\{
\begin{aligned}
&\tau_l^{(x)}, \text{ if } x \text{ is \textbf{true} element}\\
&0, \text{ if } x \text{ is \textbf{none} element}
\end{aligned}\right.
\end{equation*}
Since \textbf{true} elements are selected at first and $i_j^*\ge i^*$,
there are more than $(1-\epsi)i^*$ \textbf{true} elements in $S_l$.
Therefore,
\begin{align*}
\marge{A_l'}{A_{l, 0}'} &= \sum_{x\in A_l'\setminus A_{l, 0}', x\text{ is \textbf{true} element}} \marge{x}{A_{l, (x)}'}
+\sum_{x\in A_l'\setminus A_{l, 0}', x\text{ is \textbf{none} element}} \marge{x}{A_{l, (x)}'}\\
&\ge \sum_{x\in A_l'\setminus A_{l, 0}', x\text{ is \textbf{true} element}} \marge{x}{A_{l, (x)}}
+\sum_{x\in A_l'\setminus A_{l, 0}', x\text{ is \textbf{none} element}} \marge{x}{A_{l, (x)}}\\
&\ge (1-\epsi) |A_l\setminus A_{l, 0}| \tau_l^{(x)}, \text{for any } x\in A_l\setminus A_{l, 0}\\
&= (1-\epsi)\sum_{x\in A_l\setminus A_{l, 0}} \tau_l^{(x)}.
\end{align*}
The third part of Property 4 holds.

To prove the second part of Property 4, consider any \textbf{false} element $x\in S_l$.
By Line~\ref{line:tgone-subset}, it holds that $\mathcal V_{l, (x)} = S_{l, (x)}$.
Then, by Line~\ref{line:prefix-B-false}
\begin{equation}\label{ineq:tgone-false}
\marge{x}{A_{l, (x)}} = \marge{x}{A_{l, 0} \cup S_{l, (x)}}
= \marge{x}{A_{l, 0} \cup \mathcal V_{l, (x)}} < 0.
\end{equation}
By Line~\ref{line:tgone-subset-2}, all the elements in $S_l\setminus S_l'$ are
\textbf{false} elements.
Then,
\begin{align*}
\ff{A_l} - \ff{A_{l, 0}} &= \sum_{x\in S_l'}\marge{x}{A_{l, (x)}} + \sum_{x\in S_l\setminus S_l'}\marge{x}{A_{l, (x)}}\\
&< \sum_{x\in S_l'}\marge{x}{A_{l, (x)}} \tag{Inequality~\ref{ineq:tgone-false}}\\
&\le \sum_{x\in S_l'}\marge{x}{A_{l, (x)}'} \tag{Submodularity}\\
& = \ff{A_l'} - \ff{A_{l, 0}'}.
\end{align*}
\end{proof}

By Lemma~\ref{lemma:tgone-iteration},
we are ready to prove Lemma~\ref{lemma:ptgone}.
\begin{proof}[Proof of Lemma~\ref{lemma:ptgone}]
\textbf{Proof of Property 1.}
By Property 2 and 3 in Lemma~\ref{lemma:tgone-iteration},
this property holds immediately.

\textbf{Proof of Property 2.}
For any $l\in [\ell]$,
since $O_l\cap A_j  = \emptyset$ for each $j\neq l$,
$O_l\setminus A_l$ is outside of any solution set.
If $|A_l| = m$, by Property 4 of Lemma~\ref{lemma:tgone-iteration},
\begin{align*}
\marge{O_l}{A_l} &\le \sum_{y \in O_l\setminus A_l}\marge{y}{A_l}\\
&\le \sum_{x \in A_l}\tau_l^{(x)}/(1-\epsi)\tag{Property 3 in Lemma~\ref{lemma:tgone-iteration}}\\
& \le \frac{\marge{A_l'}{\emptyset}}{(1-\epsi)^2}.\tag{Property 5 in Lemma~\ref{lemma:tgone-iteration}}
\end{align*}
If $|A_l| < m$, then the threshold value for solution $A_l$ has been updated to be less than $\tau_{\min}$.
Therefore, for any $y\in O_l\setminus A_l$,
it holds that $\marge{y}{A_l} < \frac{\tau_{\min}}{1-\epsi}$.
Then,
\begin{align*}
\marge{O_l}{A_l} \le \sum_{y \in O_l\setminus A_l}\marge{y}{A_l}
\le \frac{m\tau_{\min}}{1-\epsi}.
\end{align*}
Therefore, Property 2 holds by summing the above two inequalities.

\textbf{Proof of Property 3 and 4.}
Let $a_{l, j}$ be the $j$-th element added to $A_l$,
$\tau_l^j$ be the threshold value of $\tau_l$ when $a_{l, j}$ is added to $A_l$,
and $A_{l, j}$ be $A_l$ after $a_{l, j}$ is added to $A_l$.
Let $c_l^* = \max\{c\in [m]:A_{l, c}\subseteq O_l\}$.

In the following, we analyze these properties together under two cases,
similar to the analysis of Alg.~\ref{alg:ptgtwo}.
For the case where $\ell = 2$, 
let $O_{1} = S\setminus A_2$, and $O_2 = S\setminus A_1$,
unifying the notations used in Property 3 and 4.
Note that, the only difference between the two analyses is that,
a small portion (no more than $\epsi$ fraction) of elements in the solution returned by Alg.~\ref{alg:ptgone}
do not have marginal gain greater than the threshold value.

\textbf{Case 1: $c_{l_1}^*\le c_{l_2}^*$; left half part in Fig.~\ref{fig:gdtwo}.}

First, we bound $\marge{O_{l_1}}{A_{l_2}}$.
Consider elements in $A_{l_1, c_{l_1}^*} \subseteq O_{l_1}$.
Let $A_{l_1, c_{l_1}^*} = \{o_1, \ldots, o_{c_{l_1}^*}\}$.
For each $1\le j \le c_{l_1}^*$, 
since $o_j$ is added to $A_{l_1}$ with threshold value $\tau_{l_1}^{j}$
and the threshold value starts from the maximum marginal gain $M$,
clearly, $o_j$ has been filtered out with threshold value $\tau_{l_1}^{j}/(1-\epsi)$.
Then, by submodularity,
\begin{equation}\label{inq:ptgone-case1-1}
\marge{A_{l_1, c_{l_1}^*} }{A_{l_2}} \le \marge{A_{l_1, c_{l_1}^*} }{\emptyset}
 = \sum_{j=1}^{c_{l_1}^*}\marge{o_j}{A_{l_1, j-1}}
 \le \sum_{j=1}^{c_{l_1}^*} \tau_{l_1}^{j}/(1-\epsi).
\end{equation}

Next, consider the elements in $O_{l_1}\setminus A_{l_1, c_{l_1}^*}$.
Order the elements in $O_{l_1}\setminus A_{l_1, c_{l_1}^*}$ as $\{o_1, o_2, \ldots\}$ such that $o_j \not \in A_{l_1, c_{l_1}^*+j}$.
(Refer to the gray block with a dotted edge in the top left corner of Fig.~\ref{fig:gdtwo} for $O_{l_1}$.
If $c_{l_1}^*+j$ is greater than $|A_{l_1}|$,
$A_{l_1, c_{l_1}^*+j}$ refers to $A_{l_1}$.)
Note that, since $A_{l_1, c_{l_1}^*} \subseteq O_{l_1}$,
it follows that $|O_{l_1}\setminus A_{l_1, c_{l_1}^*}| \le m - c_{l_1}^*$.

When $1 \le j \le |A_{l_2}| - c_{l_1}^*$,
since each $o_j$ is either added to $A_{l_1}$ or not in any solution set
and $\tau_{l_2}$ is initialized with the maximum marginal gain $M$,
$o_j$ is not considered to be added to $A_{l_2}$ with threshold value $\tau_{l_2}^{c_{l_1}^* + j}/(1-\epsi)$
by Property 3 of Lemma~\ref{lemma:tgone-iteration}.
Therefore, it holds that 
\begin{equation}\label{inq:ptgone-case1-2}
\marge{o_j}{A_{l_2, c_{l_1}^*+j-1}} < \frac{\tau_{l_2}^{c_{l_1}^* + j}}{1-\epsi} , \forall 1\le j\le |A_{l_2}|-c_{l_1}^*.
\end{equation}

When $|A_{l_2}| < m$ and $|A_{l_2}|-c_{l_1}^* < j\le m-c_{l_1}^*$,
the algorithm ends with $\tau_{l_2} < \tau_{\min}$ and
$o_j$ is never considered to be added to $A_{l_2}$.
Thus, it holds that
\begin{equation}\label{inq:ptgone-case1-3}
\marge{o_j}{A_{l_2}} < \frac{\tau_{\min}}{1-\epsi}, 
\forall |A_{l_2}|-c_{l_1}^* < j \le m-c_{l_1}^*.
\end{equation}

Then,
\begin{align*}
\marge{O_{l_1}}{A_{l_2}} &\le \marge{A_{l_1, c_{l_1}^*}}{A_{l_2}}  + \sum_{o_j \in O_{l_1}\setminus A_{l_1, c_{l_1}^*}}\marge{o_j}{A_{l_2}} \tag{Proposition~\ref{prop:sum-marge}}\\
&\le \marge{A_{l_1, c_{l_1}^*}}{\emptyset} + \sum_{j = 1}^{|A_{l_2}|-c_{l_1}^*}\marge{o_j}{A_{l_2, , c_{l_1}^*+j-1}} + \sum_{j=|A_{l_2}|-c_{l_1}^*+1}^{m-c_{l_1}^*} \marge{o_j}{A_{l_2}} \tag{submodularity}\\
&\le \sum_{j=1}^{c_{l_1}^*} \frac{\tau_{l_1}^{j}}{1-\epsi} + \sum_{j=c_{l_1}^*+1}^{|A_{l_2}|} \frac{\tau_{l_2}^{j}}{1-\epsi} + \frac{m \cdot \tau_{\min}}{1-\epsi} \numberthis \label{inq:ptgone-case1-4}
\end{align*}
where the last inequality follows from 
Inequalities~\eqref{inq:ptgone-case1-1}-\eqref{inq:ptgone-case1-3}.

Similarly, we bound $\marge{O_{l_2}}{A_{l_1}}$ below.
Consider elements in $A_{l_2, c_{l_1}^*} \subseteq O_{l_2}$.
Let $A_{l_2, c_{l_1}^*} = \{o_1, \ldots, o_{c_{l_1}^*}\}$.
For each $1\le j \le c_{l_1}^*$, 
since $o_j$ is added to $A_{l_2}$ with threshold value $\tau_{l_2}^{j}$
and the threshold value starts from the maximum marginal gain $M$,
clearly, $o_j$ has been filtered out with threshold value $\tau_{l_2}^{j}/(1-\epsi)$.
Then, by submodularity,
\begin{equation}\label{inq:ptgone-case1-5}
\marge{A_{l_2, c_{l_1}^*} }{A_{l_1}} \le \marge{A_{l_2, c_{l_1}^*} }{\emptyset}
 = \sum_{j=1}^{c_{l_1}^*}\marge{o_j}{A_{l_2, j-1}}
 \le \sum_{j=1}^{c_{l_1}^*} \tau_{l_2}^{j}/(1-\epsi).
\end{equation}

Next, consider the elements in $O_{l_2}\setminus A_{l_2, c_{l_1}^*}$.
Order the elements in $O_{l_2}\setminus A_{l_2, c_{l_1}^*}$ as $\{o_1, o_2, \ldots\}$ such that $o_j \not \in A_{l_2, c_{l_1}^*+j-1}$.
(See the gray block with a dotted edge in the bottom left corner of Fig.~\ref{fig:gdtwo} for $O_{l_2}$.
If $c_{l_1}^*+j-1$ is greater than $|A_{l_2}|$,
$A_{l_2, c_{l_1}^*+j-1}$ refers to $A_{l_2}$.)
Note that, since $A_{l_2, c_{l_1}^*} \subseteq O_{l_2}$,
it follows that $|O_{l_2}\setminus A_{l_2, c_{l_1}^*}| \le m - c_{l_1}^*$.

When $1 \le j \le |A_{l_1}|-c_{l_1}^*$,
since each $o_j$ is either added to $A_{l_2}$ or not in any solution set,
and $\tau_{l_1}$ is initialized with the maximum marginal gain $M$,
$o_j$ is not considered to be added to $A_{l_1}$ with threshold value $\tau_{l_1}^{c_{l_1}^* + j}/(1-\epsi)$
by Property 3 of Lemma~\ref{lemma:tgone-iteration}.
Therefore, it holds that 
\begin{equation}\label{inq:ptgone-case1-6}
\marge{o_j}{A_{l_1, c_{l_1}^*+j-1}} < \frac{\tau_{l_1}^{c_{l_1}^* + j}}{1-\epsi} , \forall 1\le j\le |A_{l_2}|-c_{l_1}^*.
\end{equation}

When $|A_{l_1}| < m$ and $|A_{l_1}|-c_{l_1}^* < j\le m-c_{l_1}^*$,
this iteration ends with $\tau_{l_1} < \tau_{\min}$
and $o_j$ is never considered to be added to $A_{l_1}$.
Thus, it holds that
\begin{equation}\label{inq:ptgone-case1-7}
\marge{o_j}{A_{l_1}} < \frac{\tau_{\min}}{1-\epsi}, \forall |A_{l_1}|-c_{l_1}^* < j \le m-c_{l_1}^*.
\end{equation}

Then,
\begin{align*}
\marge{O_{l_2}}{A_{l_1}} &\le \marge{A_{l_2, c_{l_1}^*}}{A_{l_1}}  + \sum_{o_j \in O_{l_2}\setminus A_{l_2, c_{l_1}^*}}\marge{o_j}{A_{l_1}} \tag{Proposition~\ref{prop:sum-marge}}\\
&\le \marge{A_{l_2, c_{l_1}^*}}{\emptyset} + \sum_{j = 1}^{|A_{l_1}|-c_{l_1}^*}\marge{o_j}{A_{l_1, c_{l_1}^*+j-1}} + \sum_{j=|A_{l_1}|-c_{l_1}^*+1}^{m-c_{l_1}^*} \marge{o_j}{A_{l_1}} \tag{submodularity}\\
&\le \sum_{j=1}^{c_{l_1}^*} \frac{\tau_{l_2}^{j}}{1-\epsi} + \sum_{j=c_{l_1}^*+1}^{|A_{l_1}|} \frac{\tau_{l_1}^{j}}{1-\epsi} + \frac{m\cdot \tau_{\min}}{1-\epsi}  \numberthis \label{inq:ptgone-case1-8}
\end{align*}
where the last inequality follows from Inequalities~\ref{inq:ptgone-case1-5}-\ref{inq:ptgone-case1-7}.

By Inequalities~\eqref{inq:ptgone-case1-4} and~\eqref{inq:ptgone-case1-8},
\begin{equation}\label{inq:ptgone-case1-final}
\marge{O_{l_1}}{A_{l_2}}+\marge{O_{l_2}}{A_{l_1}}
\le \sum_{j=1}^{|A_{l_1}|} \frac{\tau_{l_1}^{j}}{1-\epsi} + \sum_{j=1}^{|A_{l_2}|} \frac{\tau_{l_2}^{j}}{1-\epsi} + \frac{2m\cdot \tau_{\min}}{1-\epsi}
\end{equation}

\textbf{Case 2: $c_{l_1}^* > c_{l_2}^*$; right half part in Fig.~\ref{fig:gdtwo}.}

First, we bound $\marge{O_{l_1}}{A_{l_2}}$.
Consider elements in $A_{l_1, c_{l_2}^*+1} \subseteq O_{l_1}$.
Let $A_{l_1, c_{l_2}^*+1} = \{o_1, \ldots, o_{c_{l_2}^*+1}\}$.
For each $1\le j \le c_{l_2}^*+1$, 
since $o_j$ is added to $A_{l_1}$ with threshold value $\tau_{l_1}^{j}$
and the threshold value starts from the maximum marginal gain $M$,
clearly, $o_j$ has been filtered out with threshold value $\tau_{l_1}^{j}/(1-\epsi)$.
Then, by submodularity,
\begin{equation}\label{inq:ptgone-case2-1}
\marge{A_{l_1, c_{l_2}^*+1} }{A_{l_2}} \le \marge{A_{l_1, c_{l_2}^*+1} }{\emptyset}
 = \sum_{j=1}^{c_{l_2}^*+1}\marge{o_j}{A_{l_1, j-1}}
 \le \sum_{j=1}^{c_{l_2}^*+1} \tau_{l_1}^{j}/(1-\epsi).
\end{equation}

Next, consider the elements in $O_{l_1}\setminus A_{l_1, c_{l_2}^*+1}$.
Order the elements in $O_{l_1}\setminus A_{l_1, c_{l_2}^*+1}$ as $\{o_1, o_2, \ldots\}$ such that $o_j \not \in A_{l_1, c_{l_1}^*+j}$.
(Refer to the gray block with a dotted edge in the top right corner of Fig.~\ref{fig:gdtwo} for $O_{l_1}$.
If $c_{l_1}^*+j$ is greater than $|A_{l_1}|$,
$A_{l_1, c_{l_1}^*+j}$ refers to $A_{l_1}$.)
Note that, since $A_{l_1, c_{l_2}^*+1} \subseteq O_{l_1}$,
it follows that $|O_{l_1}\setminus A_{l_1, c_{l_2}^*+1}| \le m - c_{l_2}^*-1$.

When $1 \le j \le |A_{l_2}| - c_{l_2}^* - 1$,
since each $o_j$ is either added to $A_{l_1}$ or not in any solution set
and $\tau_{l_2}$ is initialized with the maximum marginal gain $M$,
$o_j$ is not considered to be added to $A_{l_2}$ with threshold value $\tau_{l_2}^{c_{l_2}^* + j}/(1-\epsi)$.
Therefore, it holds that 
\begin{equation}\label{inq:ptgone-case2-2}
\marge{o_j}{A_{l_2, c_{l_2}^*+j-1}} < \frac{\tau_{l_2}^{c_{l_2}^* + j}}{1-\epsi}, \forall 1\le j\le |A_{l_2}\setminus G_{i-1}| - c_{l_2}^* - 1.
\end{equation}

When $|A_{l_2}| < m$ and $|A_{l_2}|- c_{l_2}^* - 1 < j\le m- c_{l_2}^* - 1$,
this iteration ends with $\tau_{l_2} < \tau_{\min}$ and
$o_j$ is never considered to be added to $A_{l_2}$.
Thus, it holds that
\begin{equation}\label{inq:ptgone-case2-3}
\marge{o_j}{A_{l_2}} < \frac{\tau_{\min}}{1-\epsi}, 
\forall |A_{l_2}|- c_{l_2}^* - 1 < j \le m- c_{l_2}^* - 1.
\end{equation}

Then,
\begin{align*}
\marge{O_{l_1}}{A_{l_2}} &\le \marge{A_{l_1, c_{l_2}^*}}{A_{l_2}}  + \sum_{o_j \in O_{l_1}\setminus A_{l_1, c_{l_2}^*+1}}\marge{o_j}{A_{l_2}} \tag{Proposition~\ref{prop:sum-marge}}\\
&\le \marge{A_{l_1, c_{l_2}^*}}{\emptyset} + \sum_{j = 1}^{|A_{l_2}|- c_{l_2}^* - 1}\marge{o_j}{A_{l_2, c_{l_2}^*+j-1}} + \sum_{j=|A_{l_2}|- c_{l_2}^*}^{m- c_{l_2}^* - 1} \marge{o_j}{A_{l_2}} \tag{submodularity}\\
&\le \sum_{j=1}^{c_{l_2}^*+1} \tau_{l_1}^{j}/(1-\epsi)
+ \sum_{j = c_{l_2}^*+1}^{|A_{l_2}|} \tau_{l_2}^{j}/(1-\epsi) + \frac{m \cdot \tau_{\min}}{1-\epsi}
 \numberthis \label{inq:ptgone-case2-4}
\end{align*}
where the last inequality follows from 
Inequalities~\eqref{inq:ptgone-case2-1}-\eqref{inq:ptgone-case2-3}.

Similarly, we bound $\marge{O_{l_2}}{A_{l_1}}$ below.
Consider elements in $A_{l_1, c_{l_2}^*} \subseteq O_{l_2}$.
Let $A_{l_2, c_{l_2}^*} = \{o_1, \ldots, o_{c_{l_2}^*}\}$.
For each $1\le j \le c_{l_2}^*$, 
since $o_j$ is added to $A_{l_2}$ with threshold value $\tau_{l_2}^{j}$
and the threshold value starts from the maximum marginal gain $M$,
clearly, $o_j$ has been filtered out with threshold value $\tau_{l_2}^{j}/(1-\epsi)$.
Then, by submodularity,
\begin{equation}\label{inq:ptgone-case2-5}
\marge{A_{l_2, c_{l_2}^*} }{A_{l_1}} \le \marge{A_{l_2, c_{l_2}^*} }{\emptyset}
 = \sum_{j=1}^{c_{l_2}^*}\marge{o_j}{A_{l_2, j-1}}
 \le \sum_{j=1}^{c_{l_2}^*} \tau_{l_2}^{j}/(1-\epsi).
\end{equation}

Next, consider the elements in $O_{l_2}\setminus A_{l_2, c_{l_2}^*}$.
Order these elements as $\{o_1, o_2, \ldots\}$ such that $o_j \not \in A_{l_2, c_{l_2}^*+j}$.
(See the gray block with a dotted edge in the bottom right corner of Fig.~\ref{fig:gdtwo} for $O_{l_2}$.
If $c_{l_2}^*+j$ is greater than the number of elements added to $A_{l_2}$,
$A_{l_2, c_{l_2}^*+j}$ refers to $A_{l_2}$.)
Note that, since $A_{l_2, c_{l_2}^*} \subseteq O_{l_2}$,
it follows that $|O_{l_2}\setminus A_{l_2, c_{l_2}^*}| \le m - c_{l_2}^*$.

Furthermore, for the case where $\ell  = 2$, as considered in Property 4,
we have $O_1 = S\setminus A_2$ and $O_2 = S\setminus A_1$ for a given $S\subseteq \uni$
where $|S| \le m$.
Since $c_{l_1}^* > c_{l_2}^* \ge 0$,
it follows that $c_{l_1}^*\ge 1$, which implies $|O_2| = |S\setminus A_1|\le m-1$.
In this case, it holds that $|O_{l_2}\setminus A_{l_2, c_{l_2}^*}| \le m - c_{l_2}^*-1$.

When $1 \le j \le |A_{l_1}|- c_{l_2}^* - 1$, 
since each $o_j$ is either added to $A_{l_2}$ or not in any solution set by Claim~\ref{claim:par-A}
and $\tau_{l_1}$ is initialized with the maximum marginal gain $M$,
$o_j$ is not considered to be added to $A_{l_1}$ with threshold value $\tau_{l_1}^{c_{l_2}^* + j+1}/(1-\epsi)$.
Therefore, it holds that 
\begin{equation}\label{inq:ptgone-case2-6}
\marge{o_j}{A_{l_1, c_{l_2}^*+j}} < \frac{\tau_{l_1}^{c_{l_2}^* + j+1}}{1-\epsi}, \forall 1\le j\le |A_{l_2}|- c_{l_2}^* - 1.
\end{equation}

If $|A_{l_1}| = m$,
consider the last element $o_{m-c_{l_2}^*}$ in $O_{l_2}\setminus A_{l_2, c_{l_2}^*}$.
Since $o_{m-c_{l_2}^*} \not\in A_{l_2}$ and $o_{m-c_{l_2}^*} \not\in A_{l_1}$, $o_{m-c_{l_2}^*}$ is not considered to be added to 
$A_{l_1}$ with threshold value $\tau_{l_1}^j/(1-\epsi)$ for any $j \in [m]$.
Then,
\begin{equation}\label{inq:ptgone-case2-7}
\marge{o_{m-c_{l_2}^*}}{A_{l_1}} < \frac{\sum_{j=1}^m \tau_{l_1}^j}{(1-\epsi)m}.
\end{equation}
Else, $|A_{l_1}| < m$ and 
this iteration ends with $\tau_{l_1} < \frac{\epsi M}{k}$.
For any $|A_{l_1}|- c_{l_2}^* - 1 < j\le m- c_{l_2}^*$,
$o_j$ is never considered to be added to $A_{l_1}$.
Thus, it holds that
\begin{equation}\label{inq:ptgone-case2-8}
\marge{o_j}{A_{l_1}} < \frac{\tau_{\min}}{1-\epsi}, 
\forall |A_{l_1}|- c_{l_2}^* - 1 < j \le m- c_{l_2}^*.
\end{equation}

Then,
\begin{align*}
&\marge{O_{l_2}}{A_{l_1}} \le \marge{A_{l_2, c_{l_2}^*}}{A_{l_1}}  + \sum_{o_j \in O_{l_2}\setminus A_{l_2, c_{l_2}^*}}\marge{o_j}{A_{l_1}} \tag{Proposition~\ref{prop:sum-marge}}\\
&\le \left\{
\begin{aligned}
&\marge{A_{l_2, c_{l_2}^*}}{\emptyset} + \sum_{j = 1}^{|A_{l_1}\setminus G_{i-1}|- c_{l_2}^* - 1}\marge{o_j}{A_{l_1, c_{l_2}^*+j-1}} + \sum_{j=|A_{l_1}\setminus G_{i-1}|- c_{l_2}^*}^{m-c_{l_2}^*} \marge{o_j}{A_{l_1}}, &&\text{ if } |O_{l_2}| = m\\
&\marge{A_{l_2, c_{l_2}^*}}{\emptyset} + \sum_{j = 1}^{|A_{l_1}\setminus G_{i-1}|- c_{l_2}^* - 1}\marge{o_j}{A_{l_1, c_{l_2}^*+j-1}} + \sum_{j=|A_{l_1}\setminus G_{i-1}|- c_{l_2}^*}^{m-c_{l_2}^*-1} \marge{o_j}{A_{l_1}}, &&\text{otherwise}
\end{aligned}
\right. \tag{submodularity}\\
&\le\left\{
\begin{aligned}
	&\sum_{j=1}^{c_{l_2}^*} \frac{\tau_{l_2}^j}{1-\epsi} + \sum_{j=c_{l_2}^*+2}^{|A_{l_2}|} \left(1+\frac{1}{m}\right) \frac{\tau_{l_1}^j}{1-\epsi} + \frac{m\cdot \tau_{\min}}{1-\epsi}, &&\text{ if } |O_{l_2}| = m\\
	&\sum_{j=1}^{c_{l_2}^*} \frac{\tau_{l_2}^j}{1-\epsi} + \sum_{j=c_{l_2}^*+2}^{|A_{l_2}|} \frac{\tau_{l_1}^j}{1-\epsi} + \frac{m\cdot \tau_{\min}}{1-\epsi}, &&\text{otherwise}
\end{aligned}
\right. \numberthis \label{inq:ptgone-case2-9}
\end{align*}
where the last inequality follows from Inequalities~\eqref{inq:ptgone-case2-5}-\eqref{inq:ptgone-case2-8}.

By Inequalities~\eqref{inq:ptgone-case2-4} and~\eqref{inq:ptgone-case2-9},
\begin{equation}\label{inq:ptgone-case2-final}
\marge{O_{l_1}}{A_{l_2}}+\marge{O_{l_2}}{A_{l_1}}
\le \left\{
\begin{aligned}
	& \left(1+\frac{1}{m}\right) \frac{1}{1-\epsi}\left(\sum_{j=1}^{|A_{l_1}|} \tau_{l_1}^{j} + \sum_{j=1}^{|A_{l_2}|} \tau_{l_2}^{j}\right) + \frac{2m\cdot \tau_{\min}}{1-\epsi}, &&\text{ if } |O_{l_2}| = m \\
	&  \frac{1}{1-\epsi}\left(\sum_{j=1}^{|A_{l_1}|} \tau_{l_1}^{j} + \sum_{j=1}^{|A_{l_2}|} \tau_{l_2}^{j}\right) + \frac{2m\cdot \tau_{\min}}{1-\epsi}, &&\text{ otherwise }
\end{aligned}
\right.
\end{equation}

Overall, in both cases, if $|O_{l_2}| = m$,
\begin{align*}
\marge{O_{l_1}}{A_{l_2}}+\marge{O_{l_2}}{A_{l_1}}
&\le \left(1+\frac{1}{m}\right) \frac{1}{1-\epsi}\left(\sum_{j=1}^{|A_{l_1}|} \tau_{l_1}^{j} + \sum_{j=1}^{|A_{l_2}|} \tau_{l_2}^{j}\right) + \frac{2m\cdot \tau_{\min}}{1-\epsi} \tag{Inequalities~\eqref{inq:ptgone-case1-final} and~\eqref{inq:ptgone-case2-final}}\\
&\le \left(1+\frac{1}{m}\right)\frac{1}{(1-\epsi)^2}\left(\marge{A_{l_1}'}{\emptyset}+\marge{A_{l_2}'}{\emptyset}\right) + \frac{2m\cdot \tau_{\min}}{1-\epsi} \tag{Property 4 of Lemma~\ref{lemma:tgone-iteration}}
\end{align*}
Otherwise, if $|O_{l_2}| < m$,
\begin{align*}
\marge{O_{l_1}}{A_{l_2}}+\marge{O_{l_2}}{A_{l_1}}
&\le \frac{1}{1-\epsi}\left(\sum_{j=1}^{|A_{l_1}|} \tau_{l_1}^{j} + \sum_{j=1}^{|A_{l_2}|} \tau_{l_2}^{j}\right) + \frac{2m\cdot \tau_{\min}}{1-\epsi} \tag{Inequalities~\eqref{inq:ptgone-case1-final} and~\eqref{inq:ptgone-case2-final}}\\
&\le \frac{1}{(1-\epsi)^2}\left(\marge{A_{l_1}'}{\emptyset}+\marge{A_{l_2}'}{\emptyset}\right) + \frac{2m\cdot \tau_{\min}}{1-\epsi} \tag{Property 4 of Lemma~\ref{lemma:tgone-iteration}}
\end{align*}
Property (3) and (4) hold.

\textbf{Proof of Adaptivity and Query Complexity.}
Note that, at the beginning of every iteration,
for any $j\in I$, $V_j$ contains all the elements outside of all solutions that has marginal gain greater than $\tau_j$ with respect to solution $A_j$.
Say an iteration \textit{successful} if either
1) algorithm terminates after this iteration because of $m_0=0$,
2) all the elements in $V_j$ can be filtered out at the end of this iteration
and the value of $\tau_j$ decreases,
or 3) the size of $V_j$ decreases by a factor of $1-\frac{\epsi}{4\ell}$.
Then, by Property 1 of Lemma~\ref{lemma:tgone-iteration},
with a probability of at least $1/2$,
the iteration is successful.
Furthermore, if $\tau_j$ is less than $\tau_{\min}$,
$j$ will be removed from $I$ and 
solutions $A_j$ and $A_j'$ won't be updated anymore.

For each $j \in [\ell]$,
there are at most $\log_{1-\epsi}\left(\frac{\tau_{\min}}{M}\right) \le \epsi^{-1}\log\left(\frac{M}{\tau_{\min}}\right)$ possible threshold values.
And, for each threshold value, with at most 
$\log_{1-\frac{\epsi}{4\ell}}\left(\frac{1}{n}\right) \le 4\ell\epsi^{-1}\log(n)$
successful iterations regarding solution $A_j$,
the threshold value $\tau_j$ will decrease
or the algorithm terminates because of $m_0=0$.
Overall, with at most $4\ell^2\epsi^{-2}\log(n)\log\left(\frac{M}{\tau_{\min}}\right)$
successful iterations,
the algorithm terminates because of $m_0=0$ or $I=\emptyset$.

Next, we prove that, after $N=4\left(\log(n)+ 4\ell^2\epsi^{-2}\log(n)\log\left(\frac{M}{\tau_{\min}}\right)\right)$ iterations,
with a probability of $1-\frac{1}{n}$,
there exists at least $4\ell^2\epsi^{-2}\log(n)\log\left(\frac{M}{\tau_{\min}}\right)$
successful iterations,
or equivalently, the algorithm terminates.
Let $X$ be the number of successful iterations.
Then, $X$ can be regarded as a sum of $N$ dependent Bernoulli trails,
where the success probability is larger than $1/2$.
Let $Y$ be a sum of $N$ independent Bernoulli trials,
where the success probability is equal to $1/2$.
Then, the probability that the algorithm terminates with at most $N$ iterations can be bounded as follows,
\begin{align*}
\prob{\#\text{iterations} > N} &\le \prob{X \le 4\ell^2\epsi^{-2}\log(n)\log\left(\frac{M}{\tau_{\min}}\right)} \\
& \overset{(a)}{\le} \prob{Y \le 4\ell^2\epsi^{-2}\log(n)\log\left(\frac{M}{\tau_{\min}}\right)}\tag{Lemma~\ref{lemma:indep}}\\
&\le e^{- \frac{N}{4}\left(1-\frac{8\ell^2\epsi^{-2}\log(n)\log\left(\frac{M}{\tau_{\min}}\right)}{N}\right)^2} \tag{Lemma~\ref{lemma:chernoff}}\\
&= e^{-\frac{\left(4\log(n)+ 8\ell^2\epsi^{-2}\log(n)\log\left(\frac{M}{\tau_{\min}}\right)\right)^2}{16\left(\log(n)+ 4\ell^2\epsi^{-2}\log(n)\log\left(\frac{M}{\tau_{\min}}\right)\right)}} \le \frac{1}{n}.
\end{align*}
Therefore, with a probability of $1-\frac{1}{n}$,
the algorithm terminates with $\oh{\ell^2\epsi^{-2}\log(n)\log\left(\frac{M}{\tau_{\min}}\right)}$ iterations of the while loop.

In Alg.~\ref{alg:ptgone}, oracle queries occur during calls
to \update and \prefix 
on Line~\ref{line:tgone-update-2},~\ref{line:tgone-prefix} and~\ref{line:tgone-update}.
The \prefix algorithm,
with input $(f, \mathcal V, s, \tau, \epsi)$,
operates with $1$ adaptive rounds
and at most $|\mathcal V|$ queries.
The \update algorithm,
with input $(f, V_0, \tau_0, \epsi)$, 
outputs $(V, \tau)$
with $1+\log_{1-\epsi}\left(\frac{\tau}{\tau_0}\right)$ adaptive rounds
and at most $|V| + n\log_{1-\epsi}\left(\frac{\tau}{\tau_0}\right)$ queries.
Here, $\log_{1-\epsi}\left(\frac{\tau}{\tau_0}\right)$ equals the number of iterations in the while loop within \update.
Notably, every iteration is successful,
as the threshold value is updated.
Consequently, we can regard an iteration of the while loop in \update
as a separate iteration of the while loop in Alg.~\ref{alg:ptgone},
where such iteration only update one threshold value $\tau_j$ and its corresponding candidate set $V_j$.
So, each redefined iteration has no more than $2$ adaptive rounds,
and then the adaptivity of the algorithm should be no more than 
the number of successful iterations, which is
$\oh{\ell^2\epsi^{-2}\log(n)\log\left(\frac{M}{\tau_{\min}}\right)}$.
Since there are at most $\ell n$ queries at each adaptive rounds,
the query complexity is bounded by $\oh{\ell^3\epsi^{-2}n\log(n)\log\left(\frac{M}{\tau_{\min}}\right)}$.

\end{proof}

\subsection{\rev Analysis of Guarantees achieved by \ptgoneshort (Theorem~\ref{thm:ptgone}, Section~\ref{sec:ptg})}
\label{apx:ptgone-guarantee}
\rev
In this section, we provide the analysis of the parallel $(1/4-\epsi)$-approximation algorithm.
\color{black}
\thmptgone*
\begin{proof}[Proof of Theorem~\ref{thm:ptgone}]
The adaptivity and query complextiy are quite straightforward.
In the following, we will analyze the approximation ratio.

Let $S = O$ in Lemma~\ref{lemma:ptgone},
it holds that
\begin{align}
	& \ff{A_l'} \ge \ff{A_l}, \forall l = 1,2 \label{inq:ptg-1}\\
	& A_1 \cap A_2 = \emptyset \label{inq:ptg-2}\\
	& \marge{O}{A_1} + \marge{O}{A_2} \le \frac{1}{(1-\epsi)^2}\left(\ff{A_1'} + \ff{A_2'} \right) + \frac{2\epsi M}{1-\epsi}\label{inq:ptg-3}
\end{align}
Then,
\begin{align*}
	\ff{O} &\le \ff{O\cup A_1} + \ff{O\cup A_2} \tag{Submodularity, Nonnegativity, Inequality~\eqref{inq:ptg-2}} \\
	&\le \ff{A_1} + \ff{A_2} + \frac{1}{(1-\epsi)^2}\left(\ff{A_1'} + \ff{A_2'} \right) + \frac{2\epsi M}{1-\epsi} \tag{Inequality~\eqref{inq:ptg-3}}\\
	&\le 2\left(1+\frac{1}{(1-\epsi)^2}\right)\ff{G} + \frac{2\epsi}{1-\epsi}\ff{O} \tag{Inequality~\eqref{inq:ptg-1} and $G = \argmax\{\ff{A_1'}, \ff{A_2'}\}$}\\
	\Rightarrow \ff{G} &\ge \frac{(1-3\epsi)(1-\epsi)}{2\left((1-\epsi)^2 + 1+\frac{1}{k}\right)}\ff{O}\ge \left(\frac{1}{4}-\epsi\right)\ff{O} 
\end{align*}
\end{proof}

\subsection{\rev Pseudocode and Analysis of \ptgtwoshort (Theorem~\ref{thm:ptgtwo}, Section~\ref{sec:ptg})}
\label{apx:ptgtwo}
\begin{algorithm}[ht]
\Fn{\ptgtwo($f, k, \epsi$)}{
	\KwIn{evaluation oracle $f:2^{\uni} \to \reals$, 
        constraint $k$, constant $\ell$, error $\epsi$}
	\Init{$G\gets \emptyset, \epsi' \gets \frac{\epsi}{2}, m\gets \left\lfloor \frac{k}{\ell} \right\rfloor, M\gets \max_{x\in\uni}\ff{\{x\}}, \tau_{\min}\gets \frac{\epsi'M}{k}$}
	\For{$i\gets 1$ to $\ell$}{
		$\{A_l': l\in [\ell]\} \gets \ptgone(f_{G}, m, \ell, \tau_{\min}, \epsi')$\;
		$G\gets$ a random set in $\{G\cup A_l': l\in [\ell]\}$\;
	}
	\Return{$G$}
	}
\caption{A randomized $(1/e-\epsi)$-approximation algorithm with $\oh{\ell^{3}\epsi^{-2}\log(n)\log(k)}$ adaptivity and $\oh{\ell^4\epsi^{-2}n\log(n)\log(k)}$ query complexity}\label{alg:ptg}
\label{alg:ptgtwo}
\end{algorithm}
\rev
This subsection presents the pseudocode and theoretical analysis of our parallel $(1/e-\epsi)$-approximation algorithm.
\color{black}

First, we provide the following lemma which provides a lower bound
on the gains achieved after every iteration in Alg.~\ref{alg:ptgtwo}.
\begin{lemma}\label{lemma:ptgtwo-recur}
For any iteration $i$ of the outer for loop in Alg.~\ref{alg:ptgtwo},
it holds that
\begin{align*}
\text{1) } & \ex{\ff{O\cup G_i}}\ge \left(1-\frac{1}{\ell}\right) \ex{\ff{O\cup G_{i-1}}}\\
\text{2) } & \ex{\ff{G_i} - \ff{G_{i-1}}}
\ge\frac{1}{1+\frac{\ell}{(1-\epsi')^2}}\left(1-\frac{1}{m+1}\right)\left(\left(1-\frac{1}{\ell}\right)  \ex{\ff{O\cup G_{i-1}}} - \ex{\ff{G_{i-1}}} - \frac{\epsi'}{1-\epsi'}\ff{O}\right).
\end{align*}
\end{lemma}
\begin{proof}[Proof of Lemma~\ref{lemma:ptgtwo-recur}]
Fix on $G_{i-1}$ at the beginning of this iteration,
Since $\{A_l: l\in [\ell]\}$ are pairwise disjoint sets,
by Proposition~\ref{prop:sum-marge}, it holds that
\[\exc{\ff{O\cup G_i}}{G_{i-1}} = \frac{1}{\ell}\sum_{l\in [\ell]}\ff{O\cup G_{i-1}\cup A_l} \ge \left(1-\frac{1}{\ell}\right)\ff{O\cup G_{i-1}}.\]
Then, by unfixing $G_{i-1}$, the first inequality holds.

To prove the second inequality, also consider fix on $G_{i-1}$ at the beginning of iteration $i$.
By Lemma~\ref{lemma:ptgone},
$\{A_l: l\in [\ell]\}$ are paiewise disjoint sets,
and the following inequalities hold,
\begin{align}
&A_l'\subseteq A_l, \marge{A_l'}{\emptyset} \ge \marge{A_l}{\emptyset}, \forall 1\le l \le \ell \label{inq:ptgtwo-1}\\
&\marge{O_{l}}{A_{l}}\le \frac{\marge{A_{l}'}{\emptyset}}{(1-\epsi')^2}+\frac{\epsi' M}{(1-\epsi')\ell}, \forall 1\le l \le \ell \label{inq:ptgtwo-2}\\
&\marge{O_{l_2}}{A_{l_1}} + \marge{O_{l_1}}{A_{l_2}} \le \frac{1+\frac{1}{m}}{(1-\epsi')^2}\left(\marge{A_{l_1}'}{\emptyset}+\marge{A_{l_2}'}{\emptyset}\right) + \frac{2\epsi' M}{(1-\epsi')\ell}, \forall 1\le l_1 < l_2 \le \ell \label{inq:ptgtwo-3}
\end{align}
Then,
\begin{align*}
\sum_{l\in [\ell]}\marge{O}{A_l\cup G_{i-1}} &\le \sum_{l_1\in [\ell]}\sum_{l_2\in [\ell]}\marge{O_{l_1}}{A_{l_2}\cup G_{i-1}}\tag{Proposition~\ref{prop:sum-marge}}\\
& = \sum_{l \in [\ell]}\marge{O_{l}}{A_{l}\cup G_{i-1}} + \sum_{1\le l_1< l_2 \le \ell} \left(\marge{O_{l_1}}{A_{l_2}\cup G_{i-1}} +\marge{O_{l_2}}{A_{l_1}\cup G_{i-1}}\right) \tag{Lemma~\ref{lemma:tg-par-A}}\\
& \le \sum_{l \in [\ell]}\left(\frac{\marge{A_{l}'}{G_{i-1}}}{(1-\epsi')^2}+\frac{\epsi' M}{(1-\epsi')\ell}\right)\\
&\hspace*{2em}+\sum_{1\le l_1< l_2 \le \ell} \left(\frac{\left(1+\frac{1}{m}\right)}{(1-\epsi')^2}\left(\marge{A_{l_1}'}{G_{i-1}}
+\marge{A_{l_2}'}{G_{i-1}}\right)
+\frac{2\epsi' M}{(1-\epsi')\ell}\right)\tag{Inequalities~\eqref{inq:ptgtwo-2} and~\eqref{inq:ptgtwo-3}}\\
&\le \frac{\ell}{(1-\epsi')^2}\left(1+\frac{1}{m}\right)\sum_{l \in [\ell]}\marge{A_{l}'}{G_{i-1}} + \frac{\epsi' \ell}{1-\epsi'}\ff{O}\tag{$M \le \ff{O}$}
\end{align*}
\begin{align*}
\Rightarrow \left(1+\frac{\ell}{(1-\epsi')^2}\right)\left(1+\frac{1}{m}\right) \sum_{l\in [\ell]}\marge{A_l'}{G_{i-1}} &\ge \sum_{l\in [\ell]}\ff{O\cup A_l\cup G_{i-1}} -\ell\ff{G_{i-1}} - \frac{\epsi' \ell}{1-\epsi'}\ff{O} \tag{Inequality~\eqref{inq:ptgtwo-1}}\\
&\ge \left(\ell-1\right)\ff{O\cup G_{i-1}}-\ell\ff{G_{i-1}} - \frac{\epsi' \ell}{1-\epsi'}\ff{O}
\end{align*}
Thus,
\begin{align*}
&\exc{\ff{G_i} - \ff{G_{i-1}}}{G_{i-1}}  = \frac{1}{\ell}\sum_{l \in [\ell]}\marge{A_{l}'}{G_{i-1}}\\
&\ge \frac{1}{1+\frac{\ell}{(1-\epsi')^2}} \frac{m}{m+1}\left(\left(1-\frac{1}{\ell}\right)  \ff{O\cup G_{i-1}} - \ff{G_{i-1}} - \frac{\epsi'}{1-\epsi'}\ff{O}\right)\tag{Proposition~\ref{prop:sum-marge}}
\end{align*}
By unfixing $G_{i-1}$, the second inequality holds.
\end{proof}
\thmptgtwo*
\begin{proof}[Proof of Theorem~\ref{thm:ptgtwo}]
Since the algorithm contains a for loop
which runs \ptgone $\ell = \oh{1/\epsi}$ times,
by Lemma~\ref{lemma:ptgone},
the adaptivity, query complexity and success probability holds immediately.

Next, we provide the analysis of approximation ratio.
By solving the recurrence in Lemma~\ref{lemma:ptgtwo-recur},
we calculate the approximation ratio of the algorithm as follows,
\begin{align*}
&\ex{\ff{G_{i}}}  \ge \left(1-\frac{1}{\ell}\right) \ex{\ff{G_{i-1}}}
+ \frac{1}{1+\frac{\ell}{(1-\epsi')^2}}\left(1-\frac{1}{m+1}\right)\left(\left(1-\frac{1}{\ell}\right)^i - \frac{\epsi'}{1-\epsi'}\right)\ff{O}\\
\Rightarrow& \ex{\ff{G_\ell}} \ge \frac{\ell}{1+\frac{\ell}{(1-\epsi')^2}}\left(1-\frac{1}{m+1}\right)\left(\left(1-\frac{1}{\ell}\right)^\ell - \frac{\epsi'}{1-\epsi'}\left(1-\left(1-\frac{1}{\ell}\right)^\ell\right)\right)\ff{O}\\
&\hspace*{4em} \ge \frac{\ell-1}{1+\frac{\ell}{(1-\epsi')^2}}\left(1-\frac{1}{m+1}\right)\left(e^{-1} - \frac{\epsi'}{1-\epsi'}\left(1-e^{-1}\right)\right)\ff{O}\\
&\hspace*{4em} \ge \frac{1}{1-\frac{\ell}{k}}\left((1-\epsi')^2 - \frac{2}{\ell}\right)\left(1-\frac{\ell}{k}\right)^2\left(e^{-1} - \frac{\epsi'}{1-\epsi'}\left(1-e^{-1}\right)\right) \ff{O}\\
% &\hspace*{4em} \ge \frac{1}{1-\frac{\ell}{k}}\left(1-\epsi' - \frac{2}{\ell}\right)\left(1-\frac{2\ell}{k}\right)\left(e^{-1} - \frac{\epsi'}{1-\epsi'}\left(1-e^{-1}\right)\right) \ff{O}\\
&\hspace*{4em} \ge \frac{1}{1-\frac{\ell}{k}}\left((1-\epsi')^2 - \frac{2}{\ell}-\frac{2(1-\epsi')^2 \ell}{k}\right)\left(e^{-1} - \frac{\epsi'}{1-\epsi'}\left(1-e^{-1}\right)\right) \ff{O}\\
&\hspace*{4em} \ge \frac{1}{1-\frac{\ell}{k}} \left(1-(e+1)\epsi'\right)\left(e^{-1} - \frac{\epsi'}{1-\epsi'}\left(1-e^{-1}\right)\right) \ff{O}\tag{$\ell\ge \frac{2}{e\epsi'}, k\ge \frac{2(1-\epsi')^2\ell}{e\epsi'-\frac{2}{\ell}}$}\\
&\hspace*{4em} \ge \frac{1}{1-\frac{\ell}{k}} \left(e^{-1}-\epsi\right)\ff{O}\tag{$\epsi' = \frac{\epsi}{2}$}.
\end{align*}
By Inequality~\ref{inq:tgtwo-dif-opt},
the approximation ratio of Alg.~\ref{alg:tgtwo} is $e^{-1}-\epsi$.
\end{proof}

%% file: apx_exp.tex
\section{Experimental Setups and Additional Empirical Results}\label{apx:exp}
In the section, we introduce the settings in Section~\ref{sec:exp} further, and discuss more experimental results on \nmon and \mon.

\subsection{Applications}\label{apx:app}
\textbf{Maxcut.}
In the context of the maxcut application, we start with a graph $G=(V, E)$ 
where each edge $ij \in E$ has a weight $w_{ij}$.
The objective is to find a cut that maximizes the total weight of edges crossing the cut.
The cut function $f: 2^V \to \reals$ is defined as follows,
\[f(S) = \sum_{i \in S} \sum_{j \in V\setminus S}w_{ij}, \forall S\subseteq V.\]
This is a non-monotone submodular function.
In our implementation, for simplicity, all edges have a weight of $1$.

\textbf{Revmax.}
In our revenue maximization application,
we adopt the revenue maximization model introduced in \citep{DBLP:conf/www/HartlineMS08}, which we will briefly outline here.
Consider a social network $G=(V, E)$,
where $V$ denotes the buyers.
Each buyer $i$'s value for a good depends on the set of buyers $S$ that already own it, 
which is formulated by 
\[v_i(S)=f_i\left(\sum_{j \in S} w_{ij}\right),\]
where $f_i: \reals \to \reals$ is a non-negative, monotone, concave function, and $w_{ij}$ is drawn independently from a distribution.
The total revenue generated from selling goods to the buyers $S$ is
\[f(S) = \sum_{i \in V\setminus S} f_i\left(\sum_{j \in S} w_{ij}\right).\]
This is a non-monotone submodular function.
In our implementation, 
we randomly choose each $w_{ij}\in (0,1)$,
and $f_i(x) = x^{\alpha_i}$, where $\alpha_i \in (0,1)$ is chosen uniformly randomly.

% \textbf{Imgsum.}
% We follow the setting of Personalized Image Summarization application in~\citet{mirzasoleiman2016fast} with the following objective function
% \[f(S) = \sum_{i \in \uni} \max_{j \in S}s_{ij} - \frac{1}{n}\sum_{i \in S}\sum_{j \in S}s_{ij},\]
% where $s_{ij}$ determines the cosine similarity of image $i$ to image $j$
% with pixel vectors.
% The first term tries to ensure that the set $S$ is a good summary of the dataset,
% while the second promotes diversity within the summary itself. 
% This is a non-monotone, submodular objective function.

\subsection{Datasets}\label{apx:data}
\textbf{er} is a synthetic random graph generated by Erd{\"{o}}s-R{\'{e}}nyi model~\citep{erdds1959random} by setting number of nodes $n=100,000$ and edge probability $p=\frac{5}{n}$.

\textbf{web-Google}~\citep{DBLP:journals/im/LeskovecLDM09}  is a web
graph of $n=875,713$ web pages as nodes and $5,105,039$ hyperlinks
as edges.

\textbf{musae-github}~\citep{rozemberczki2019multiscale} is a social network of GitHub developers with $n=37,700$ developers and $289,003$ edges,
where edges are mutual follower relationships between them.

\textbf{twitch-gamers}~\citep{rozemberczki2021twitch} is a social network of $n=168,114$ Twitch users with $6,797,557$ edges, 
where edges are mutual follower relationships between them.

% \textbf{CIFAR-10}~\citep{krizhevsky2009learning} dataset consists of $50,000$ training images and $10,000$ test images
% where each image 
% is represented by a pixel vector of length 3,072:
% $32 \times 32$ pixels with red, green, and blue channels.
% In this paper, we randomly choose $3,000$ images from the training dataset.

\subsection{Additional Results}\label{apx:nmon}
Fig.~\ref{fig:apx} provides additional results on musae-github dataset with $n=37,700$
and web-Google dataset with $n=875,713$.
It shows that as $n$ and $k$ increase, our algorithms achieve superior on objective values.
The results of query complexity and adaptivity align closely with those discussed in Section~\ref{sec:exp}.
Notably, the number of adaptive round of \ptgtwoshort exceeds $k$ on musae-github,
which may be attributed to the dataset's relatively small size.
\begin{figure}[ht]
    \centering
    \subfigure[musae-github, solution value]{\label{fig:git-val}
    \includegraphics[width=0.31\linewidth]{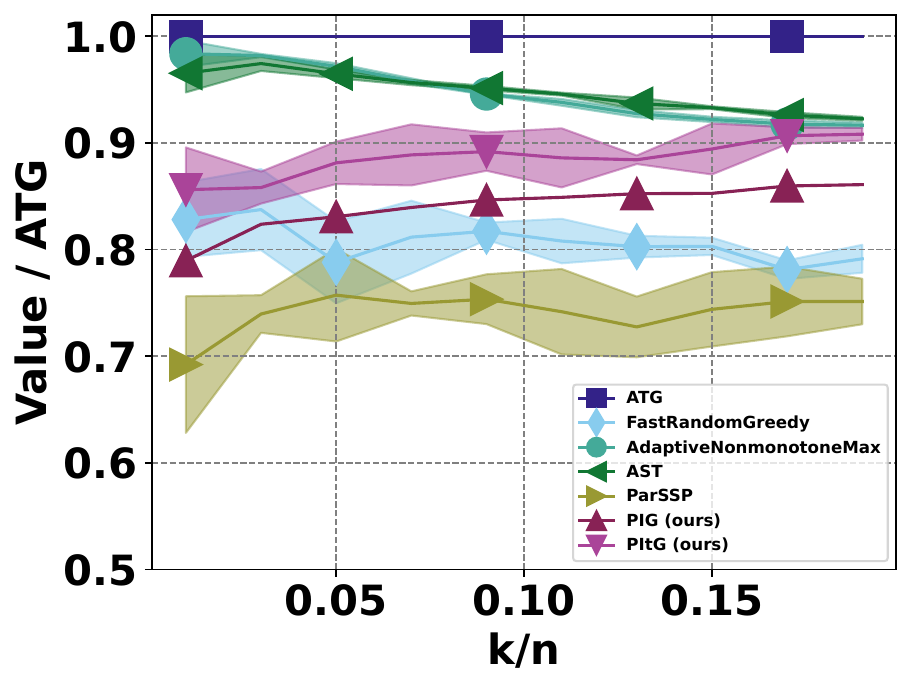}}
    \subfigure[musae-github, query]{\label{fig:git-query}
    \includegraphics[width=0.31\linewidth]{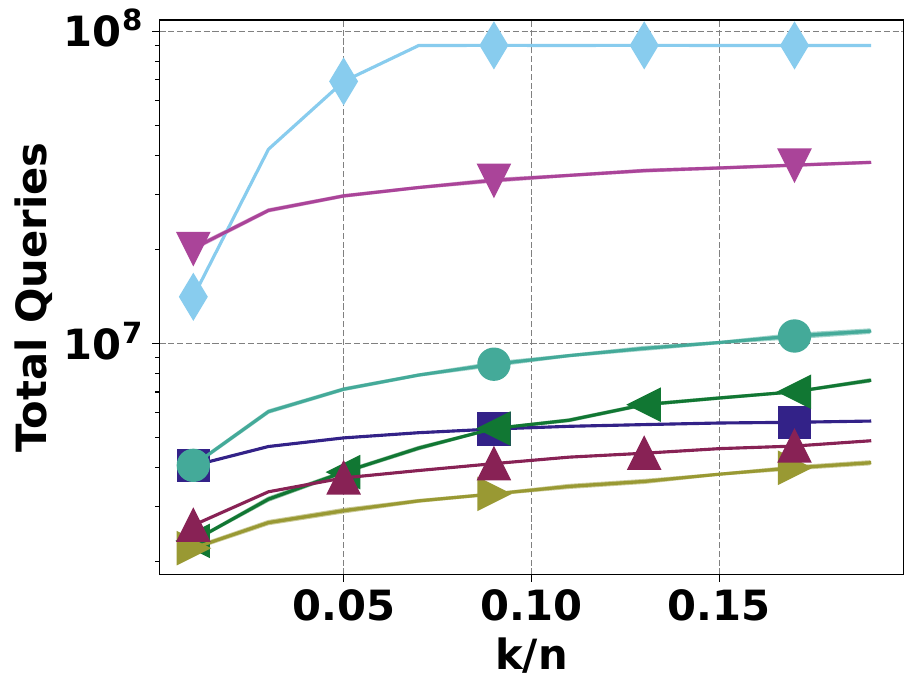}}
    \subfigure[musae-github, round]{\label{fig:git-round}
    \includegraphics[width=0.31\linewidth]{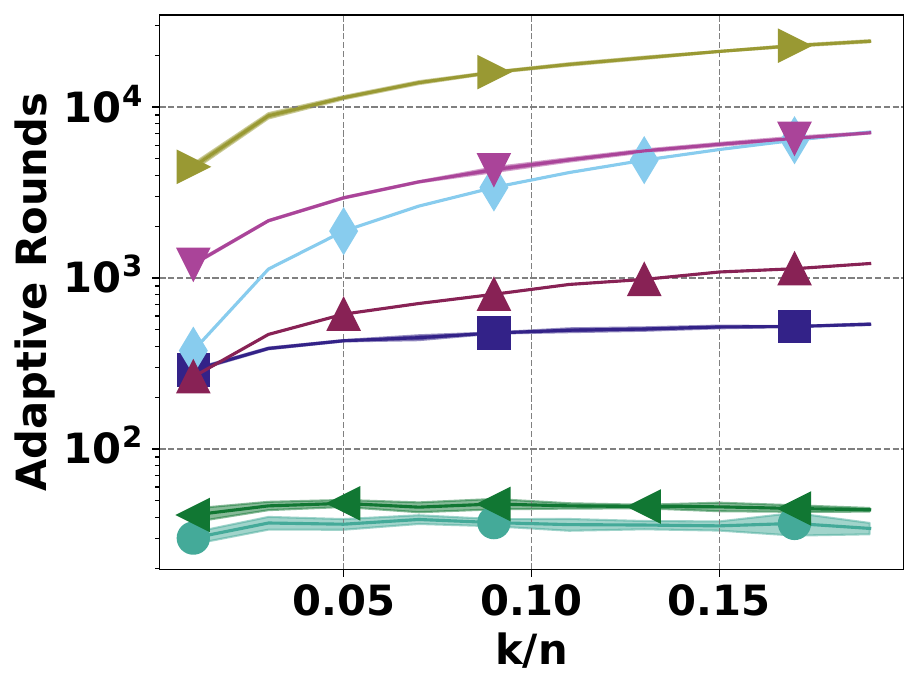}}
    \subfigure[web-Google, solution value]{\label{fig:google-val}
    \includegraphics[width=0.31\linewidth]{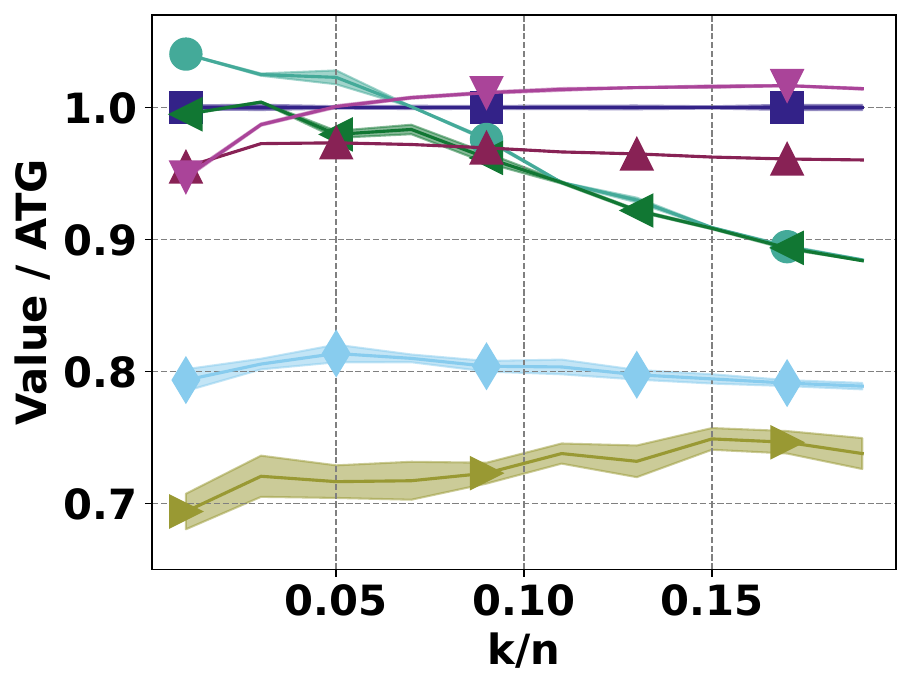}}
    \subfigure[web-Google, query]{\label{fig:google-query}
    \includegraphics[width=0.31\linewidth]{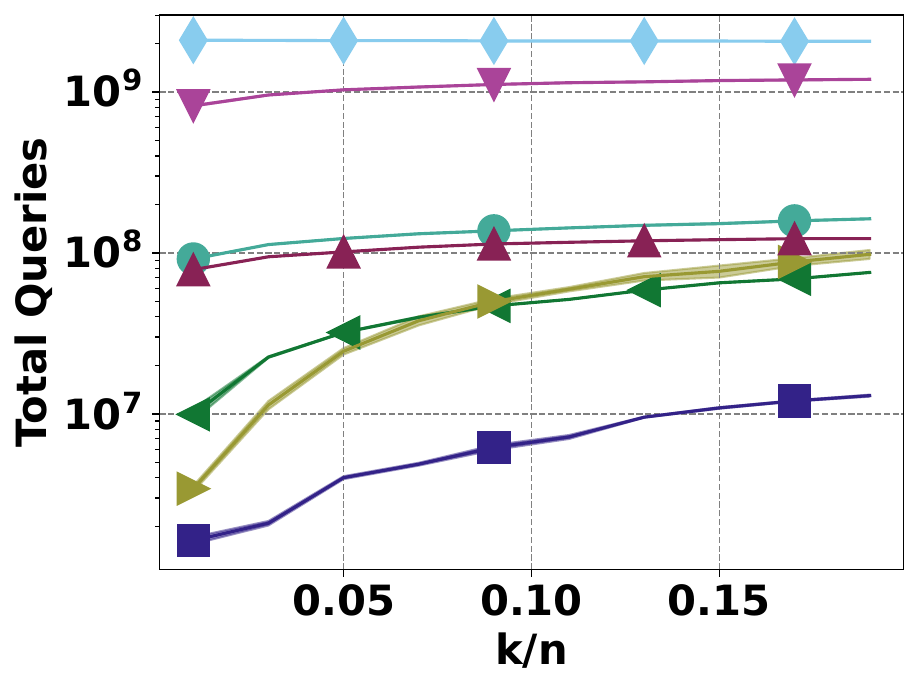}}
    \subfigure[web-Google, round]{\label{fig:google-round}
    \includegraphics[width=0.31\linewidth]{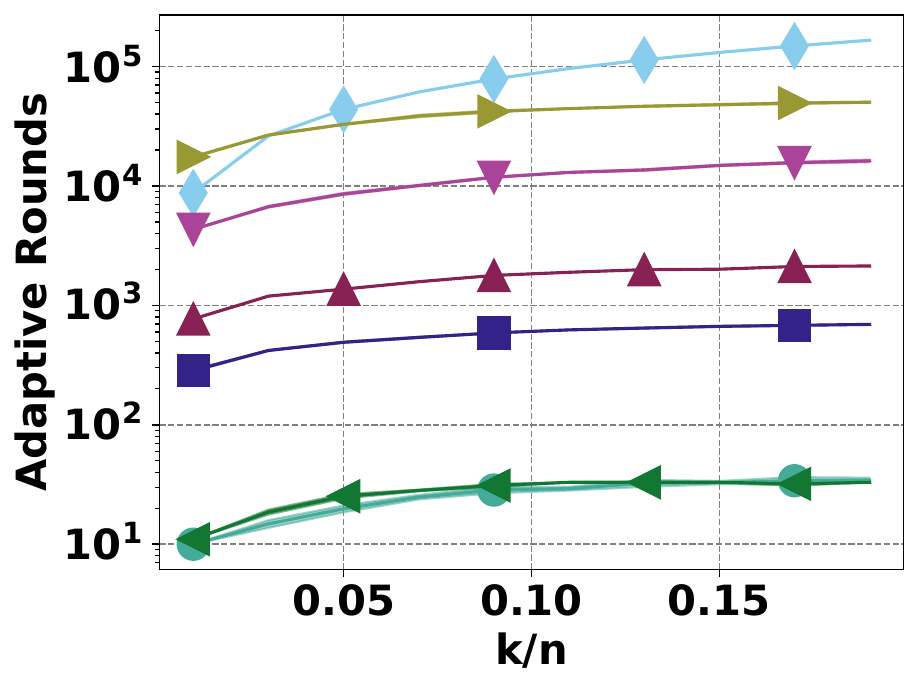}}
    \caption{Results for \revmax on musae-github with $n=37,700$,
    and \maxcut on web-Google with $n=875,713$.}
    \label{fig:apx}
\end{figure}